\documentclass[envcountsect,orivec]{LMCS}
\def\dOi{10(4:20)2014}
\lmcsheading%
{\dOi}
{1--45}
{}
{}
{Aug.~\phantom08, 2013}
{Dec.~30, 2014}
{}

\ACMCCS{[{\bf Software and its engineering}]: Software notations and
  tools---General programming languages---Language types; [{\bf Theory of
      computation}]:  Models of computation---Concurrency---Process calculi} 


\usepackage[dvipsnames]{xcolor}
\usepackage{amsmath}
\usepackage{amssymb}
\usepackage{graphicx}
\usepackage{latexsym}
\usepackage{listings}
\usepackage{multirow}
\usepackage{suffix}
\usepackage{url}
\usepackage{mathptmx}
\usepackage{mathrsfs}
\usepackage{comment}
\usepackage{color}      
\usepackage{tikz}
\usepackage{array}
\usepackage{hyperref}

\newif\iflong\longfalse
\newif\ifsynch\synchfalse
\newif\ifcamera\camerafalse

\synchtrue

\longtrue

\newcommand{\semicolon}{:}
\newcommand{\lrangle}[1]{\langle #1 \rangle}

\newcommand{\parenthtext}[1]{(\textrm{\small #1})}
\newcommand{\brtext}[1]{[\textrm{\small #1}]}
\newcommand{\textinmath}[1]{\textrm{#1}}

\newcommand{\strule}[1]{\textrm{#1}}

\newcommand{\ltsrule}[1]{{\footnotesize \lrangle{\textrm{#1}}}}
\newcommand{\eltsrule}[1]{{\footnotesize \brtext{#1}}}
\newcommand{\erule}[1]{{\footnotesize \{\textrm{#1}\}}}
\newcommand{\trule}[1]{{\footnotesize\brtext{#1}}}
\newcommand{\orule}[1]{{\footnotesize{\brtext{#1}}}}
\newcommand{\mrule}[1]{{\footnotesize{\parenthtext{#1}}}}

\newcommand{\noi}{\noindent}
\newcommand{\Hline}{\rule{\linewidth}{.5pt}}

\newcommand{\bnfis}{\;\;::=\;\;}
\newcommand{\bnfbar}{\;\;\;|\;\;\;}

\newcommand{\Case}[1]{\noi {\bf Case: }#1\\}

\newcommand{\LogAnd}{\texttt{ and }}

\newcommand{\tree}[2]{
\ensuremath{\displaystyle
		\frac
		{
			#1
		}{
			\raisebox{-0.4mm}{$\displaystyle{#2}$}
		}
	}
}

\newcommand{\vect}[1]{\tilde{#1}}

\newcommand{\set}[1]{\{#1\}}
\newcommand{\es}{\emptyset}
\newcommand{\maxset}[1]{\max(#1)}
\newcommand{\setbar}{\ \ |\ \ }

\newcommand{\eval}{\downarrow}

\newcommand{\dom}[1]{\mathtt{dom}(#1)}

\newcommand{\funcbr}[2]{#1\lrangle{#2}}

\newcommand{\freev}[1]{\lrangle{#1}}
\newcommand{\boundv}[1]{(#1)}

\newcommand{\send}[1]{\overline{#1}}

\newcommand{\inact}{\mathbf{0}}
\newcommand{\If}{\sessionfont{if}\ }
\newcommand{\Then}{\sessionfont{then}\ }
\newcommand{\Else}{\sessionfont{else}\ }
\newcommand{\ifthen}[2]{\If #1\ \Then #2\ }
\newcommand{\IfThenElse}[3]{\ifthen{#1}{#2} \Else #3}
\newcommand{\Par}{\;|\;}
\newcommand{\news}[1]{(\nu\ #1)}
\newcommand{\newsp}[2]{(\nu\ #1)(#2)}
\newcommand{\varp}[1]{#1}
\newcommand{\rec}[2]{\mu \varp{#1}. #2}

\newcommand{\bn}[1]{\mathtt{bn}(#1)}
\newcommand{\fn}[1]{\mathtt{fn}(#1)}
\newcommand{\fv}[1]{\mathtt{fv}(#1)}
\newcommand{\bv}[1]{\mathtt{bv}(#1)}

\newcommand{\subj}[1]{\mathtt{subj}(#1)}
\newcommand{\obj}[1]{\mathtt{obj}(#1)}

\newcommand{\relfont}[1]{\mathcal{#1}}
\newcommand{\rel}[3]{#1\ \relfont{#2}\ #3}

\newcommand{\scong}{\equiv}
\newcommand{\acong}{\scong_{\alpha}}
\newcommand{\wb}{\approx}
\newcommand{\swb}{\approx^{s}}

\newcommand{\red}{\longrightarrow}
\newcommand{\Red}{\rightarrow\!\!\!\!\!\rightarrow}

\newcommand{\subst}[2]{\set{#1/#2 }}

\newcommand{\hole}{-}
\newcommand{\context}[2]{#1[#2]}
\newcommand{\Ccontext}[1]{\C[#1]}

\newcommand{\barb}[1]{\downarrow_{#1}}
\newcommand{\Barb}[1]{\Downarrow_{#1}}

\newcommand{\sessionfont}[1]{\mathtt{#1}}
\newcommand{\vart}[1]{\mathsf{#1}}

\newcommand{\ssep}{;}
\newcommand{\shsep}{.}
\newcommand{\outses}{!}
\newcommand{\inpses}{?}
\newcommand{\selses}{\oplus}
\newcommand{\brases}{\&}
\newcommand{\dual}[1]{\overline{#1}}
\newcommand{\cat}{\cdot}

\newcommand{\bacc}[2]{#1 \boundv{#2} \shsep}
\newcommand{\breq}[2]{\send{#1} \boundv{#2} \shsep}

\newcommand{\bout}[2]{#1 \outses \freev{#2} \ssep}
\newcommand{\binp}[2]{#1 \inpses \boundv{#2} \ssep}
\newcommand{\bsel}[2]{#1 \selses #2 \ssep}

\newcommand{\bbras}[2]{#1 \brases #2}

\newcommand{\role}[1]{[#1]}

\newcommand{\srole}[2]{#1\role{#2}}

\newcommand{\fromto}[2]{\role{#1} \role{#2}}
\newcommand{\sfromto}[3]{#1\fromto{#2}{#3}}

\newcommand{\sout}[3]{\srole{#1}{#2} \outses \freev{#3} \ssep}
\newcommand{\sinp}[3]{\srole{#1}{#2} \inpses \boundv{#3} \ssep}

\newcommand{\ssel}[3]{\srole{#1}{#2} \selses #3 \ssep}

\newcommand{\sbraP}[2]{\srole{#1}{#2} \brases \lPi}

\newcommand{\acc}[3]{#1 \role{#2} \boundv{#3} \shsep}
\newcommand{\req}[3]{\send{#1} \role{#2} \boundv{#3} \shsep}

\newcommand{\out}[4]{\sfromto{#1}{#2}{#3} \outses \freev{#4} \ssep}
\newcommand{\inp}[4]{\sfromto{#1}{#2}{#3} \inpses \boundv{#4} \ssep}

\newcommand{\sel}[4]{\sfromto{#1}{#2}{#3} \selses #4 \ssep}
\newcommand{\bra}[4]{\sfromto{#1}{#2}{#3} \brases \set{#4}}

\newcommand{\braP}[3]{\sfromto{#1}{#2}{#3} \brases \lPi}

\newcommand{\typingred}{\red}

\newcommand{\cohses}[2]{\mathtt{co}(#1(#2))}
\newcommand{\coherent}[1]{\mathtt{co}(#1)}
\newcommand{\fcoherent}[1]{\mathtt{fco}(#1)}

\newcommand{\emp}{\epsilon}

\newcommand{\gtfont}[1]{\mathtt{#1}}
\newcommand{\gsep}{.}

\newcommand{\globaltype}[1]{\lrangle{#1}}

\newcommand{\roles}[1]{\mathtt{roles}(#1)}

\newcommand{\fromtogt}[2]{#1 \rightarrow #2 \semicolon}

\newcommand{\valuegt}[3]{\fromtogt{#1}{#2} \lrangle{#3} \gsep}
\newcommand{\selgt}[3]{\fromtogt{#1}{#2} \set{#3}}
\newcommand{\selgtG}[2]{\fromtogt{#1}{#2} \lGi}
\newcommand{\recgt}[2]{\mu \vart{#1}. #2}
\newcommand{\vargt}[1]{\vart{#1}}
\newcommand{\inactgt}{\gtfont{end}}

\newcommand{\projset}[1]{\mathtt{proj}(#1)}

\newcommand{\gcong}{\equiv}
\newcommand{\govcong}{\cong_g^s}
\newcommand{\govcongs}{\cong_g^s}
\newcommand{\gperm}{\simeq}

\newcommand{\projsymb}{\lceil}
\newcommand{\proj}[2]{#1 \projsymb #2}

\newcommand{\tfont}[1]{\mathtt{#1}}
\newcommand{\tsep}{;}

\newcommand{\chtype}[1]{\lrangle{#1}}

\newcommand{\outtype}{\outses}
\newcommand{\inptype}{\inpses}
\newcommand{\seltype}{\selses}
\newcommand{\bratype}{\brases}

\newcommand{\trec}[2]{\mu\vart{#1}.#2}
\newcommand{\tvar}[1]{\vart{#1}}

\newcommand{\tinact}{\tfont{end}}

\newcommand{\tout}[2]{\role{#1} \outtype \lrangle{#2} \tsep}
\newcommand{\tinp}[2]{\role{#1} \inptype (#2) \tsep}
\newcommand{\tsel}[2]{\role{#1} \seltype \set{#2}}
\newcommand{\tsels}[2]{\role{#1} \seltype #2}
\newcommand{\tselT}[1]{\role{#1} \seltype \lTi}
\newcommand{\tbra}[2]{\role{#1} \bratype \set{#2}}
\newcommand{\tbras}[2]{\role{#1} \bratype #2}
\newcommand{\tbraT}[1]{\role{#1} \bratype \lTi}

\newcommand{\btout}[1]{\outtype \lrangle{#1} \tsep}
\newcommand{\btinp}[1]{\inptype (#1) \tsep}
\newcommand{\btsel}[1]{\seltype \set{#1}}

\newcommand{\btbra}[1]{\bratype \set{#1}}

\newcommand{\mtout}[2]{\role{#1} \outtype \lrangle{#2}}

\newcommand{\Ga}{\Gamma}
\newcommand{\De}{\Delta}
\newcommand{\proves}{\vdash}
\newcommand{\hastype}{\triangleright}

\newcommand{\Decat}[1]{\De \cat #1}
\newcommand{\Gacat}[1]{\Ga \cat #1}

\newcommand{\typed}[1]{#1:}
\newcommand{\typedrole}[2]{\typed{\srole{#1}{#2}}}

\newcommand{\typedprocess}[3]{#1 \proves #2 \hastype #3}

\newcommand{\Eproves}[3]{#1 \proves \typed{#2} #3}
\newcommand{\Gproves}[2]{\Eproves{\Ga}{#1}{#2}}

\newcommand{\tprocess}[3]{#1 \proves #2 \hastype #3}
\newcommand{\Gtprocess}[2]{\tprocess{\Ga}{#1}{#2}}

\newcommand{\noGtprocess}[2]{#1 \hastype #2}

\newcommand{\globalenvI}{E}

\newcommand{\SGtprocess}[2]{\Gtprocess{#1}{#2, \globalenvI}}

\newcommand{\fromtolts}[2]{\fromto{#1}{#2}}

\newcommand{\outlts}{\outses}
\newcommand{\inplts}{\inpses}
\newcommand{\sellts}{\selses}
\newcommand{\bralts}{\brases}

\newcommand{\actreq}[3]{\send{#1} \role{#2} \boundv{#3}}

\newcommand{\actreqs}[3]{\send{#1} \role{\set{#2}} \boundv{#3}}

\newcommand{\actacc}[3]{#1 \role{#2} \boundv{#3}}
\newcommand{\actaccs}[3]{#1 \role{\set{#2}} \boundv{#3}}

\newcommand{\actout}[4]{#1 \fromtolts{#2}{#3} \outlts \freev{#4}}

\newcommand{\actbout}[4]{#1 \fromtolts{#2}{#3} \outlts \boundv{#4}}

\newcommand{\actbdel}[5]{#1 \fromtolts{#2}{#3} \outlts \boundv{\srole{#4}{#5}}}

\newcommand{\actinp}[4]{#1 \fromtolts{#2}{#3} \inplts \freev{#4}}

\newcommand{\actsel}[4]{#1 \fromtolts{#2}{#3} \sellts #4}

\newcommand{\actbra}[4]{#1 \fromtolts{#2}{#3} \bralts #4}

\newcommand{\actval}[4]{#1: #2 \rightarrow #3:#4}
\newcommand{\actgsel}[4]{#1: #2 \rightarrow #3:#4}

\newcommand{\trans}[1]{\stackrel{#1}{\longrightarrow}}
\newcommand{\Trans}[1]{\stackrel{#1}{\Longrightarrow}}

\newcommand{\barbreq}[1]{\barb{#1}}
\newcommand{\barbout}[3]{\barb{\sfromto{#1}{#2}{#3}}}

\newcommand{\comp}{\asymp}
\newcommand{\coh}{\asymp}
\newcommand{\bistyp}{\rightleftharpoons}

\newcommand{\ordercup}{\bowtie}

\newcommand{\envtrans}[1]{\trans{#1}}

\newcommand{\govwb}{\approx_g^s}
\newcommand{\govwbs}{\approx_g^s}

\newcommand{\true}{\sessionfont{tt}}
\newcommand{\false}{\sessionfont{ff}}

\newcommand{\bool}{\sessionfont{bool}}

\newcommand{\PP}{\ensuremath{P}}
\newcommand{\Q}{\ensuremath{Q}}
\newcommand{\R}{\ensuremath{R}}

\newcommand{\s}{\ensuremath{s}}

\newcommand{\Ia}{\ensuremath{a}}
\newcommand{\Iu}{\ensuremath{u}}

\newcommand{\x}{\ensuremath{x}}
\newcommand{\y}{\ensuremath{y}}

\newcommand{\va}{\ensuremath{v}}
\newcommand{\e}{\ensuremath{e}}
\newcommand{\n}{\ensuremath{n}}

\newcommand{\p}{\ensuremath{\mathtt{p}}}
\newcommand{\q}{\ensuremath{\mathtt{q}}}
\newcommand{\A}{\ensuremath{A}}

\newcommand{\G}{\ensuremath{G}}
\newcommand{\gG}{\globaltype{\G}}
\newcommand{\U}{\ensuremath{U}}
\newcommand{\So}{\ensuremath{S}}
\newcommand{\T}{\ensuremath{T}}

\newcommand{\C}{\ensuremath{C}}
\newcommand{\E}{\ensuremath{E}}

\newcommand{\suc}{\textrm{succ}}

\newcommand{\lPi}{\set{l_i:\PP_i}_{i \in I}}
\newcommand{\lGi}{\set{l_i:\G_i}_{i \in I}}
\newcommand{\lTi}{\set{l_i:\T_i}_{i \in I}}

\newif\ifny\nyfalse
\newcommand{\NY}[1]
{\ifny{\color{purple}{#1}}\else{#1}\fi}

\newif\ifdk\dktrue

\newif\ifrhu\rhutrue

\newcommand{\ENCan}[1]{\langle #1 \rangle}

\newcommand{\syntaxvspace}{\\[1mm]}

\newcommand{\CODE}[1]{{\tt #1}}

\newcommand{\instrumentzero}{\mathtt{i0}}
\newcommand{\instrument}{\mathtt{i}}
\newcommand{\agentone}{\mathtt{a_1}}
\newcommand{\agenttwo}{\mathtt{a_2}}
\newcommand{\user}{\mathtt{u}}

\newcommand{\transs}[1]{\stackrel{#1}{\longrightarrow_{s}}}

\newcommand{\typedqrolei}[2]{#1^{i}[#2]:}
\newcommand{\typedqroleo}[2]{#1^{o}[#2]:}

\newcommand{\sqrolei}[2]{#1^{i}[#2]}

\newcommand{\actqinpi}[4]{#1^{i}[#2][#3]?\lrangle{#4}}

\newcommand{\mtcat}[1]{*(#1)}

\newcommand{\outp}[1]{\mathtt{out}(#1)}
\newcommand{\inpp}[1]{\mathtt{inp}(#1)}

\title[Globally Governed Session Semantics]{Globally Governed Session Semantics}

\author[D.~Kouzapas]{Dimitrios Kouzapas\rsuper a}
\address{{\lsuper{a,b}}Imperial College London}
\email{\{Dimitrios.Kouzapas,n.yoshida\}@imperial.ac.uk}
\thanks{{\lsuper a}University of Glasgow, Dimitrios.Kouzapas@glasgow.ac.uk (On leave).}	

\author[N. Yoshida]{Nobuko Yoshida\rsuper b}	
\address{\vspace{-18 pt}}	


\begin{document}

\begin{abstract}
	This paper proposes a bisimulation theory based on multiparty
	session types where a choreography specification governs the
	behaviour of session typed processes and their observer. The
	bisimulation is defined with the observer cooperating with the
	observed process in order to form complete global session scenarios
	and usable for proving correctness of optimisations for globally
	coordinating threads and processes. The induced bisimulation is
	strictly more fine-grained than the standard session bisimulation.
	The difference between the governed and standard bisimulations only
	appears when more than two interleaved multiparty sessions exist.
	This distinct feature enables to reason real 
	scenarios in the large-scale distributed system where multiple 
	choreographic sessions need to be interleaved. 
	The compositionality of the governed bisimilarity is proved 
	through the soundness and completeness with respect to the
	governed reduction-based congruence.
	Finally, its usage is
	demonstrated by a thread transformation governed under multiple
	sessions in a real usecase in the large-scale
	cyberinfrustracture.
\end{abstract}

\maketitle


\section{Introduction}\label{sec:intro}
\paragraph{\bf Backgrounds. \ } Modern society increasingly depends on distributed
software infrastructure such as the
backend of popular Web portals,
global E-science cyberinfrastructure,
e-healthcare and e-governments.
An application in such distributed environments
is typically organised into
many components that communicate
through message passing.
Thus an application is naturally designed as
a collection of interaction scenarios, or
{\em multiparty sessions}, each
following an interaction pattern, or {\em choreographic protocol}.
The theory for multiparty
session types \cite{HYC08} captures these two
natural abstraction units, representing the situation
where two or more multiparty sessions (choreographies) can interleave
for a single point application, with each message clearly
identifiable as belonging to a specific session.

This article introduces a behavioural theory which can reason about
distributed processes that are controlled globally
by multiple choreographic sessions.
Typed behavioural theory has been one of
the central topics of the study of the $\pi$-calculus 
throughout its history,
for example, reasoning about various encodings
into the typed $\pi$-calculi \cite{PiSa96b,Yoshida96}.
Our theory treats the mutual effects of multiple choreographic
sessions which are shared 
among distributed participants as their common
knowledge or agreements, 
reflecting the origin of choreographic frameworks \cite{CDL}. 
These features related to multiparty session type discipline 
make our theory distinct from any type-based
bisimulations in the literature and 
also applicable to a real choreographic usecase
from a large-scale distributed system.  
We say that our bisimulation is {\em globally governed},
since it uses global multiparty specifications to regulate the
conversational behaviour of distributed processes.

\paragraph{\bf Multiparty session types. \ }
To illustrate the idea for globally governed semantics, 
we first explain
the mechanisms of multiparty session types \cite{HYC08}.
Consider the simplest communication protocol where the participant
implementing 
$1$ sends a message of type $\bool$ to the participant
implementing $2$. 
A {\em global type} \cite{HYC08} is used to describe the
protocol as:
\[
	G_1  \ = \ \valuegt{1}{2}{\bool} \inactgt
\]
where $\rightarrow$ denotes
the flow of communication and $\inactgt$
denotes protocol termination.

Global type $G_1$ is used as an agreement for the
type-check specifications of both Server$_1$ and Server$_2$.  
We type-check implementations of both servers 
against the 
{\em projection} of $G_1$ into local session types.
The following type  
\[
	\tout{2}{\bool} \tinact
\]
is the local session type from the point of
view of participant $1$, that describes the 
output of a $\bool$-type value towards participant $2$,
while local type:
\[
	\tinp{1}{\bool} \tinact
\]
is the local session type from 
the point of view of $2$ that
describes the input of a $\bool$-type value from $1$.

\paragraph{\bf Resource management usecase. } We will use a simplified usecase, UC.R2.13 ``Acquire Data From Instrument'', c.f.~\S~\ref{sec:app:ooi} of \cite{ooi},
to explain the main intuition of globally governed semantics
and give insights on how our theory can reason about choreographic 
interactions.

\ifcamera
\noindent\Hline\\[2mm]
\else
\begin{figure}[t]
\centering 
\fi
	\begin{tikzpicture}[scale=1.1]

	\draw[rounded corners] (-5.8, 1.8) rectangle (-4.6, 1.2);
	\draw[rounded corners] (-4.4, 1.8) rectangle (-1.8, 1.2);
	\draw[rounded corners] (-1.2, 1.8) rectangle (0, 1.2);

	\draw[dash pattern= on 8pt off 4pt] (-5.2, 1.2) -- (-5.2, -1);
	\draw[dash pattern= on 8pt off 4pt] (-3.8, 1.2) -- (-3.8, -1);
	\draw[dash pattern= on 8pt off 4pt] (-2.4, 1.2) -- (-2.4, -1);
	\draw[dash pattern= on 8pt off 4pt] (-0.6, 1.2) -- (-0.6, -1);

	\node at (-5.2, 1.5) {\scriptsize $\textsf{Server}_1$};
	\node at (-3.1, 1.5) {\scriptsize $\textsf{Server}_2$};
	\node at (-0.6, 1.5) {\scriptsize $\textsf{Client}_3$};

	\draw[<-, very thick, blue] (-0.6, 0.7) -- (-5.2, 0.7);
	\draw[->, very thick, violet] (-5.2, 0.1) -- (-2.4, 0.1);
	\draw[<-, very thick, blue] (-0.6, -0.5) -- (-3.8, -0.5);

	\node at (-1.5, 0.9) {\scriptsize $\actout{s_1}{1}{3}{v}$};
	\node at (-3.1, 0.3) {\scriptsize $\actout{s_2}{1}{2}{w}$};
	\node at (-1.5, -0.3) {\scriptsize $\actout{s_1}{2}{3}{v}$};

	\draw[rounded corners] (1, 1.8) rectangle (2.2, 1.2);
	\draw[rounded corners] (2.4, 1.8) rectangle (3.6, 1.2);
	\draw[rounded corners] (4.2, 1.8) rectangle (5.4, 1.2);

	\draw[dash pattern= on 8pt off 4pt] (1.6, 1.2) -- (1.6, -1);
	\draw[dash pattern= on 8pt off 4pt] (3, 1.2) -- (3, -1);
	\draw[dash pattern= on 8pt off 4pt] (4.8, 1.2) -- (4.8, -1);

	\node at (1.6, 1.5) {\scriptsize $\textsf{Server}_1$};
	\node at (3, 1.5) {\scriptsize $\textsf{Server}_2$};
	\node at (4.8, 1.5) {\scriptsize $\textsf{Client}_3$};

	\draw[<-, very thick, blue] (4.8, 0.7) -- (1.6, 0.7);
	\draw[->, very thick, violet] (1.6, 0.1) -- (3, 0.1);
	\draw[<-, very thick, blue] (4.8, -0.5) -- (3, -0.5);

	\node at (3.9, 0.9) {\scriptsize $\actout{s_1}{1}{3}{v}$};
	\node at (2.3, 0.3) {\scriptsize $\actout{s_2}{1}{2}{w}$};
	\node at (3.9, -0.3) {\scriptsize $\actout{s_1}{2}{3}{v}$};

\end{tikzpicture}
\ifcamera
Resource Managment Example: (a) before optimisation;
          (b) after optimisation\\
\Hline
\else
	\caption{
Resource Managment Example: (a) before optimisation;
          (b) after optimisation \label{fig:clexample}}
\end{figure}
\fi

Consider the scenario in Figure~\ref{fig:clexample}(a) where
a single threaded Client$_3$ (participant $3$)
uses two services: from the single threaded Server$_1$ (participant $1$)
and from the dual threaded Server$_2$ (participant $2$). In 
Figure~\ref{fig:clexample}, the vertical lines represent the threads
in the participants. 
Additionally Server$_1$ sends an internal signal to Server$_2$.
The communication patterns are described with the use of
the following global protocols: 
\begin{eqnarray*}
	G_a &=& \valuegt{1}{3}{\mathit{ser}}\valuegt{2}{3}{\mathit{ser}} \inactgt \\
	G_b &=& \valuegt{1}{2}{\mathit{sig}} \inactgt
\end{eqnarray*}
with $G_a$ describing the communication between
the two servers and the Client$_3$ and $G_b$ describing
the internal communication between the two servers.
Participants $1, 2$ and $3$ are assigned to processes 
$P_1$, $P_2$ and $P_3$, respectively in order to implement 
a usecase as:
\begin{eqnarray*}
P_1 &=&	\acc{a}{1}{x} \acc{b}{1}{y} \sout{x}{3}{v} \sout{y}{2}{w} \inact \\
P_2 &=&	\acc{a}{2}{x} \req{b}{2}{y} (\sinp{y}{1}{z}\inact \Par \sout{x}{3}{v} \inact)\\
P_3 &=&	\req{a}{3}{x} \sinp{x}{1}{z} \sinp{x}{2}{y} \inact
\end{eqnarray*}
Shared name $a$ establishes the session
corresponding to $G_a$, where
Client$_3$ ($P_3$) uses prefix $\overline{a}[3](x)$
to initiate a session that involves 
three processes:
Server$_1$ ($P_1$) and Server$_2$ ($P_2$) participate to
the session with prefixes $a[1](x)$ and $a[2](x)$, respectively.
In a similar way the
session corresponding to $G_b$ is established between
Server$_1$ and Server$_2$.

The above scenario is subject to an optimisation due to
the fact that the internal signal between Server$_1$ and
Server$_2$ is invisible to clients because the communication link 
created after the session initiation is local. The optimisation
is illustrated in Figure~\ref{fig:clexample}(b),
where we require a single threaded service for Server$_2$
to avoid the overhead of an extra thread creation.
The new implementation of participant $2$ is:
\[
R_2  = \acc{a}{2}{x} \req{b}{2}{y} \sinp{y}{1}{z} \sout{x}{3}{v} \inact
\]
It is important to note that both $P_2$ and $R_2$
are typable under both $G_a$ and $G_b$.

The motivation of this work is set
when we compare the two server interfaces 
that are exposed towards Client$_3$, here implemented by process $P_3$.
The two different interfaces are given by
the two different implementations $P_1 \Par P_2$ and
$P_1 \Par R_2$.

\paragraph{\bf Untyped and linear bisimulations. \ }
It is obvious that in the untyped setting \cite{SangiorgiD:picatomp},
$P_1\Par P_2$ and $P_1\Par R_2$ are
{\em not} bisimilar 
since $P_2$ can exhibit the output action $\actout{s_a}{2}{3}{v}$
before the input action $\actinp{s_b}{2}{1}{w}$ (assuming that
variable $x$ is substituted
 with session $s_a$ and
variable $y$ is substituted with session $s_b$ after
the session initiation actions take place). More concretely 
we can analyse their transitions as follows where 
$\actacc{a}{\set{1,2}}{s_a}$ is the label to start the session with Client$_3$ while $\actout{s_a}{2}{3}{v}$ is an output of value $v$ from Server${}_2$
to Client$_3$: 
\begin{eqnarray*}
	P_1 \Par P_2 \trans{\actacc{a}{\set{1,2}}{s_a}} \trans{\tau} \trans{\actout{s_a}{2}{3}{v}} & & \\
	P_1 \Par R_2 \trans{\actacc{a}{\set{1,2}}{s_a}} \trans{\tau} \stackrel{\actout{s_a}{2}{3}{v}}{\not\longrightarrow}
\end{eqnarray*}
The same transitions are observed if we restrict the transition
semantics to respect the traditional linearity principle based on  
session local types \cite{DBLP:conf/forte/KouzapasYH11}.
The full definition of the multiparty session bisimulation 
which follows the usual linearity property is found in 
\S~\ref{sec:synchronous}. We also give the detailed 
interaction patterns in Example~\ref{ex:swb} to explain why both
untyped and linear bisimulations cannot equate $P_1 \Par P_2$ and 
$P_1 \Par R_2$. 

\paragraph{\bf Globally governed bisimulation. \ }
In the global setting, 
the behaviours of 
$P_1 \Par P_2$ and $P_1 \Par R_2$ are constrained
by the specification of the global protocols, $G_a$ and
$G_b$. The service provided by Server$_2$
is available to Client$_3$
only after Server$_1$ sends a signal to Server$_2$. 
Protocol $G_a$ dictates that action
$\actout{s_a}{2}{3}{v}$ can only happen after
action $\actinp{s_b}{2}{1}{w}$ in $P_2$.
This is because Client$_3$, which uses the
service interface as a global observer, is
also typed by the global protocol $G_a$ and
can only interact with action $\actout{s_a}{2}{3}{v}$
from Server$_1$
after it interacts with action $\actout{s_a}{1}{3}{v}$
from Server$_2$.

Hence in a globally typed setting, 
processes $P_1 \Par P_2$ and $P_1 \Par R_2$ are
not distinguishable by Client$_3$ and thus the 
thread optimisation of $R_2$ is behaviourally justified.

Note that processes $P_2$ and $R_2$
(i.e.~without the parallel composition $P_1$)
are not observationaly equivalent under any set
of session typed or untyped observational semantics.
The governed bisimulation between $P_1 \Par P_2$ and $P_1\Par R_2$
is achieved if we introduce an internal
message of the session created on shared channel $b$ between
processes $P_1$ and $P_2$ and $P_1$ and $R_2$, respectively.

\paragraph{\bf Changing a specification. \ }
A global protocol directly affects the behaviour of processes.
We change global type $G_a$ with global type:
\begin{eqnarray*}
	G_a' = \valuegt{2}{3}{\mathit{ser}} \valuegt{1}{3}{\mathit{ser}}\inactgt
\end{eqnarray*}
Processes $P_1 \Par P_2$ and $P_1 \Par R_2$ are also typable under
protocol $G_a'$ but now
process $R_2$ can perform both the output to Client$_3$ and the input
from Server$_1$
concurrently and according to the protocol $G_a'$ that states that
Client$_3$ can receive a message from Server$_2$ first.
Hence $P_1\Par P_2$ and $P_1 \Par R_2$ are no longer
equivalent under global type $G_a$.

The above example gives an insight for our development of
an equivalence theory that takes into account a global type as a specification.
The interaction scenario between processes refines the behaviour of processes.
To achieve such a theory of process equivalence we require to 
observe the labelled transitions together with the information provided
by the global types.
Global types define additional
knowledge about how an observer
(in the example above the observer is Client$_3$)
will collaborate with the observed
processes so that different global types
(i.e.~global witnesses) can
induce the different behaviours.

\paragraph{\bf Contributions and outline. \ }
\label{subsec:outline}
This article introduces two classes of 
typed bisimulations based on
multiparty session types.
The
first bisimulation definition is based on 
the typing information derived by 
local (endpoint) types,
hence it resembles
the standard linearity-based bisimulation 
for session types (\cite{DBLP:conf/forte/KouzapasYH11}).
The second bisimulation definition, which
we call {\em globally governed session bisimilarity}, uses
the information from global multiparty session types to derive
the interaction pattern of the global observer. 
We prove that both bisimilarities 
coincide with a corresponding standard contextual
equivalence \cite{HondaKYoshida95}
(see Theorems~\ref{the:scoincidence} and \ref{the:gcoincidence}).
The globally governed semantics give a
more fine-grained bisimilarity equivalence comparing to 
the locally typed linear bisimulation relation. 
We identify the condition when the two semantic theories 
coincide (see Theorem \ref{lem:full-abstraction}).
Interestingly our next theorem (Theorem \ref{thm:coincidence}) shows
that both bisimilarity relations differ
only when processes implement
two or more interleaved global types.
This feature makes the theory for govern multiparty bisimulation
applicable to real situations where 
multiple choreographies are used to compose a single, large application. 
We demonstrate the use of governed
bisimulation through the running example,
which is applicable to a thread optimisation of a real usecase
from a large scale distributed system \cite{ooi}. 

This article is a full version of the extended abstract published in
\cite{KY13} and the first author's thesis \cite{dkphdthesis}.
Here we include the detailed definitions, 
expanded explanations, full reasoning of the usecases and complete proofs.
The rest of the paper is organised as follows: 
Section \ref{sec:sync-calculus} introduces a synchronous version 
of the calculus for the multiparty sessions. Section \ref{sec:typing}
defines a typing system and proves its subject reduction theorem. 
Section \ref{sec:synchronous} presents the typed behavioural theory 
based on the local types 
for the synchronous multiparty sessions and proves that the typed 
bisimulation and reduction-based barbed congruence coincide.
Section \ref{sec:governed} 
introduces the semantics for globally governed behavioural theory and proves 
the three main theorems of this article. 
Section \ref{sec:app:ooi} presents a real-world usecase based on
UC.R2.13 ``Acquire Data
From Instrument'' from the 
Ocean Observatories Initiative (OOI)~\cite{ooi} and 
shows that the governed bisimulation can justify an optimisation 
of network services. Finally, Section \ref{sec:related} concludes with 
the related work. 
The appendix includes
the full proofs.

\paragraph{\bf Acknowledgement.\ } 
The first author is funded by EPSRC Post Doctal Fellowship. 
The work has been partially sponsored by the
Ocean Observatories Initiative, EPSRC EP/K011715/1, EP/K034413/1,
EP/G015635/1 and EP/L00058X/1, and EU project FP7-612985 UpScale. 
We thank Rumyana Neykova for her suggestion of the usecase. 

\section{Synchronous Multiparty Sessions}
\label{sec:sync-calculus}

This section defines a synchronous version of the 
multiparty session types. The syntax and typing 
follow the work in \cite{BettiniCDLDY08} without the
definition for session endpoint queues (which
 are used for asynchronous communication).
We choose to define our theory in the
synchronous setting, since it allows the simplest 
formulations for demonstrating the essential concepts of 
bisimulations. 
An extension of the current theory to the asynchronous
semantics is part of the work in \cite{dkphdthesis},
where session endpoint queues are used to provide asynchronous
interaction patterns between session endpoints.

\subsection{Syntax} 
\label{sec:synch-syntax}

\begin{figure}[t]
	\[
	\begin{array}{ccc}
		\begin{array}{lclr}
			\PP & \bnfis  &	\req{\Iu}{\p}{\x} \PP &\text{Request} \syntaxvspace
			& \bnfbar & \acc{\Iu}{\p}{\x} \PP & \text{Accept} \syntaxvspace
			& \bnfbar & \sout{c}{\p}{\e} \PP &\text{Sending} \syntaxvspace
			& \bnfbar & \sinp{c}{\p}{\x} \PP &\text{Receiving} \syntaxvspace
			& \bnfbar & \ssel{c}{\p}{l} \PP &\text{Selection} \syntaxvspace
			& \bnfbar & \sbraP{c}{\p} &\text{Branching}
			\\[4mm]
			\Iu & \bnfis & \x \bnfbar \Ia &\text{Identifier}\syntaxvspace
			n & \bnfis & s \bnfbar \Ia &\text{Name}\syntaxvspace
			\e & \bnfis & \multicolumn{2}{l}{
				\va \bnfbar \x \bnfbar
				\e = \e' \bnfbar \dots
			}
		\end{array}
		&
		\begin{array}{lclr}
			& \bnfbar & \IfThenElse{\e}{\PP}{\Q} & \text{Conditional} \syntaxvspace
			& \bnfbar & \PP \Par \Q & \text{Parallel} \syntaxvspace
			& \bnfbar & \inact & \text{Inaction} \syntaxvspace
			& \bnfbar & \news{n}\PP & \text{Hiding} \syntaxvspace
			& \bnfbar & \rec{X}{\PP} & \text{Recursion} \syntaxvspace
			& \bnfbar & \varp{X} & \text{Variable} 
			\\[4mm]
			c & \bnfis & \srole{\s}{\p} \bnfbar \x & \text{Session}\syntaxvspace
			\va & \bnfis & \Ia \bnfbar \true \bnfbar \false \bnfbar s[\p]
		        & \text{Value} \syntaxvspace
			& & & 
			\text{Expression}
		\end{array}
	\end{array}
	\]
	\caption{
		\label{fig:synch-syntax}
		Syntax for synchronous multiparty session calculus 
	}
\end{figure}

The syntax of the synchronous multiparty session
calculus is defined in Figure~\ref{fig:synch-syntax}.
We assume two disjoint countable sets of names: one ranges over
{\em shared names} $a, b, \dots$ and another ranges over 
{\em session names} $s, s_1, \dots$.
Variables range over $\x, \y, \dots$.
{\em Roles} 
range over the natural numbers and are denoted as
$\p, \q, \dots$, labels range over $l, l_1, \dots$
and constants range over $\true, \false, \dots$.
In general, when the sessions are interleaved, 
a participant may implement more than one roles. We often call 
$\p,\q,...$ participants when there is no confusion. 
Symbol $\Iu$ is used to range over shared 
names or variables.
Session endpoints are denoted as $\srole{\s}{\p}, \srole{\s}{\q}, \dots$.
The symbol $c$ ranges over 
session endpoints or variables.
Values range over shared names,
constants
(we can extend constants to include natural numbers $1, 2, \dots$) 
and session endpoints $\srole{s}{\p}$.
Expressions $\e, \e', \dots$ are either values, 
logical operations on expressions or name matching 
operations $n=n'$.

The first two prefixes are the primitives for session initiation: 
$\req{\Iu}{\p}{\x} \PP$ initiates a new session through 
identifier $u$  (which represents a shared interaction point) 
on the other multiple participants, each of the shape 
$\acc{\Iu}{\q}{\x} Q_\q$ where $1 \leq \q \leq \p - 1$. 
The (bound) variable $\x$ will be substituted with
the channel used for session communication.
The basic session endpoint communication
(i.e.\ the communication that takes place
between two endpoints)
is performed with
the next two pairs of prefixes:
the prefixes for sending ($\sout{c}{\p}{e} P$) 
and receiving ($\sinp{c}{\p}{\x} P$)
a value and the prefixes for the selecting 
($\ssel{c}{\p}{l} P$) and
branching  ($\sbraP{c}{\p}$)
processes,
where
the former prefix 
chooses one of the branches offered by the latter prefix. 
Specifically, process $\sout{c}{\p}{\e} \PP$
denotes the intention of sending a value to role $\p$;
in a similar way, process $\sinp{c}{\p}{\x} \PP$ 
denotes the intention of receiving a value from
role $\p$. 
A similar interaction pattern holds 
for the selection/branching pair of communication, with the
difference that the intention is to send (respectively receive) 
a label that eventually determines 
the reduction of the branching process.
Process $\inact$ is the inactive process.
The conditional process $\IfThenElse{e}{P}{Q}$
offers processes $\PP$ and $Q$ according
to the evaluation of expression $e$. Process $\PP \Par Q$
is the parallel composition, while $\news{n} P$ restricts
name $n$ in the scope of $P$. We use the primitive recursor
$\rec{X} P$ with $\varp{X}$ as the recursive variable.
We write $\fn{P}$/$\bn{P}$ and $\fv{P}$/$\bv{P}$ to 
denote the set 
of free/bound names and free/bound variables, respectively
in process $P$.
A process is closed if it is a term with no
free variables. 

\subsection{Operational semantics}

\begin{figure}[t]
\[
	\begin{array}{c}
		\PP \scong \PP \Par \inact \quad \quad \PP \Par \Q \scong \Q \Par \PP \quad \quad (\PP \Par \Q) \Par \R \scong \PP \Par (\Q \Par \R)
		\quad \quad \PP \acong \Q \quad \quad \rec{X}{\PP} \scong \PP \subst{\rec{X}{\PP}}{X}
		\\[2mm]
		\news{\n}\news{\n'} \PP \scong 
\news{\n'}\news{\n} \PP 
		\quad \quad \news{\n} \inact \scong \inact 
\quad\quad		\newsp{\n}{\PP} \Par \Q \scong \newsp{\n}{\PP \Par \Q} \quad \n \notin \fn{\Q}
	\end{array}
\]

\caption{
	\label{fig:synch-str-cong}
	Structural Congruence for Synchronous Multiparty Session Calculus
}
\end{figure}

The operational semantics is defined in
Figure~\ref{fig-synch-operational}. 
It uses the usual structure congruence relation 
(denoted by $\equiv$), which is defined as 
the least congruence relation that respects the rules in
Figure~\ref{fig:synch-str-cong}.
Structural congruence
is an associative and commutative monoid over the
parallel ($\Par$) operation with the inactive process 
($\inact$) being the neutral element. It 
respects alpha-renaming and recursion unfolding. 
The order of
the name restriction operators has no effect with respect
to structural congruence. Finally, 
the scope opening rule extends the 
scope of the restriction on a name from a process to a paralleled process,
provided that the name does not occur free in the latter. We often 
write $\newsp{\n_1\n_2\cdots\n_m}{\PP}$ for 
$\newsp{\n_1}{\newsp{\n_2}{\cdots \newsp{\n_m}{\PP}}}$. 

\ifcamera{
	\small
\else
	\begin{figure}[t]
\fi
	\newcommand{\redvspace}{\\[2mm]}
	\[
	\begin{array}{rclcl}
		\req{\Ia}{n}{\x} \PP_n \Par \Pi_{i=\{1,..,n-1\}} 
		\acc{\Ia}{i}{\x} P_{i}  &\red& \newsp{\s}{\PP_n \subst{\srole{\s}{n}}{x} \Par 
		\Pi_{i=\{1,..,n-1\}} P_i\subst{\srole{\s}{i}}{x}} & \quad & \orule{Link} 
		\redvspace

		\out{\s}{\p}{\q}{\e} \PP \Par \inp{\s}{\q}{\p}{\x} \Q & \red &
		\PP \Par \Q \subst{\va}{\x} \quad (e \eval \va) & \quad & \orule{Comm} 
		\redvspace


		\ssel{s[\p]}{\q}{l_k} \PP \Par \braP{\s}{\q}{\p} & \red & 
		\PP \Par \PP_k \quad (k \in I)  & \quad & \orule{Label}
		\redvspace

\multicolumn{5}{c}{
		\IfThenElse{\e}{\PP}{\Q} \red \PP \quad (\e \eval
                \true)  
\quad \orule{If-F} \quad 


		\IfThenElse{\e}{\PP}{\Q} \red \Q \quad (\e \eval \false) \quad \orule{If-T}
}		\redvspace



	\end{array}
	\]
	\vspace{-3mm}
	\[
		\tree{\PP \red \PP'}{\news{n}\PP \red \news{n} \PP'} \quad \orule{Res}
		\quad \quad 
		\tree{\PP \red \PP'}{\PP \Par Q \red \PP'\Par Q} \quad \orule{Par}
		\quad \quad 
		\tree{\PP \equiv \PP'\red Q' \equiv Q}{\PP \red Q} \quad \orule{Str}
	\]
\ifcamera
	}
\else
	\caption{
		\label{fig-synch-operational}
		Operational semantics for synchronous multiparty session
		calculus 
	}
	\end{figure}
\fi

The reduction semantics of the calculus is defined in 
Figure~\ref{fig-synch-operational}.
Rule $\orule{Link}$ defines synchronous session initiation
and requires that all session endpoints must be present for
a synchronous reduction, where 
each role $\p$ creates a session endpoint $\srole{s}{\p}$
on a fresh session name $s$. The maximum participant with
the maximum role ($\req{a}{n}{x} P$)  
is responsible for requesting a session initiation.
Rule $\orule{Comm}$ describes the semantics for sending 
a value to the corresponding receiving process.
We assume an evaluation context $e\downarrow v$, where
expression $e$ evaluates to value $v$.
Session endpoint $\srole{s}{\p}$ sends a value $v$
to session endpoint $\srole{s}{\q}$, which in its turn
waits to receive a value from role $\p$.
The interaction between selection and branching processes 
is defined 
by rule $\orule{Label}$, where we expect the selection part
$\srole{s}{\p}$
of the interaction to choose the continuation on the branching 
part $\srole{s}{\q}$ of the interaction.

The rest of the rules follow the usual
$\pi$-calculus rules. Rule $\orule{If}$ 
evaluates the boolean expression $e$. If the latter is true
it proceeds with the first branch process defined; otherwise
it proceeds with the second branch process defined. 
Rules $\orule{Res}$ and $\orule{Par}$ are inductive rules
on the parallel and name restriction operators. Finally,
rule $\orule{Str}$ closes the reduction relation under
structural congruence closure.
We write $\Red$ for $(\red\cup \equiv)^\ast$ \cite{SangiorgiD:picatomp,HondaKYoshida95}.

\section{Typing for Synchronous Multiparty Sessions}
\label{sec:typing}
This section defines a typing system 
for the synchronous multiparty session calculus. The system 
is a synchronous version of one presented in \cite{BettiniCDLDY08}. 
We first define multiparty session types 
and then summarise the typing system for the synchronous
multiparty session calculus. At the end of the section, 
we \NY{prove} the subject reduction 
theorem (Theorem~\ref{the:subject}).

\subsection{Global and local types}
\label{sec:global-local}
We first give the definition of the global type and  
then define the local session type
as a projection of the global type.

\paragraph{\bf Global types,} 
ranged over by $G, G',\dots$ describe the whole 
conversation scenario of a multiparty session
as a type signature. 
The grammar of the global types 
is given in Figure~\ref{fig:types} (left). 
\begin{figure}[t]
\noindent\begin{tabular}{c|c}
	$
	\begin{array}{rcll}
		\text{Global}\hspace*{-0.3cm} \\[1mm]
  \G & \bnfis & \valuegt{\p}{\q}{\U} \G' & \strule{exchange}\\
		& \bnfbar & \selgtG{\p}{\q} & \text{branching}\\
		& \bnfbar & \recgt{t}{\G} & \text{recursion}\\
		& \bnfbar & \vargt{t} & \text{variable} \\
		& \bnfbar & \inactgt & \text{end} \\[2mm]

		\text{Exchange}\hspace*{-0.8cm}\\[1mm]
		\U & \bnfis & \So \bnfbar \T\\[2mm]
		\text{Sort}\\[1mm]
 \So & \bnfis & \bool \bnfbar \globaltype{\G}
	\end{array}
	$
&
	$
	\begin{array}{rclcr}
		\text{Local}\\
 \T & \bnfis & \tout{\p}{\U} T & & \strule{send}\\
		  & \bnfbar & \tinp{\p}{\U} T &  & \strule{receive}\\
		  & \bnfbar & \tselT{\p} &  & \strule{select}\\
		  & \bnfbar & \tbraT{\p} &  & \strule{branch}\\
		  & \bnfbar & \trec{t}{\T} & & \strule{recursion}\\
		  & \bnfbar & \tvar{t} &  & \strule{variable}\\
		  & \bnfbar & \tinact &  & \strule{end}\\
	\end{array}
	$
\end{tabular}
\caption{Global and local types} \label{fig:types}
\end{figure}

The global type $\valuegt{\p}{\q}{\U} \G'$ describes
the interaction where
role $\p$ sends a message of type $\U$ 
to role $\q$ and then the interaction described 
in $\G'$ takes place. The \emph{exchange type} $\U,\U',...$ 
consists of {\em sort} types $S, S', \ldots$ for values 
(either base types or global types), and 
{\em local session} types $\T, \T', \ldots$  for
channels (local types are defined in the next paragraph). 
Type $\selgt{\p}{\q}{l_i:\G_i}_{i\in I}$ describes the
interaction where
role $\p$ selects one of the labels $l_i$ against
role $\q$. 
If $l_j$ is selected, the interaction described in $G_j$ 
takes place.
We assume that $\p \not= \q$. Type $\recgt{t}{\G}$ 
is the recursive type. Type variables 
($\vart{t}, \vart{t'}, \dots$) are guarded,
i.e., type variables only appear under some prefix and
do not appear free in the exchange types $U$.
We take an \emph{equi-recursive} view of recursive types, 
not distinguishing between $\mu \vart{t}.\G$ and its 
unfolding $G\subst{\mu \vart{t}.\G}{\vart{t}}$
\cite[\S 21.8]{PierceBC:typsysfpl}. 
We assume that $G$ in the grammar of sorts has 
no free type variables (i.e. type variables do not appear 
freely in carried types in exchanged global types).
Type $\inactgt$ represents the 
termination of the session.

\paragraph{\bf Local types} 
Figure~\ref{fig:types} (right) defines the syntax of local types. 
They correspond to the communication actions, representing
sessions from the view-point of a single role.
The {\em send type}  $\tout{\p}{\U} T$
expresses the sending to $\p$ of a value 
of type \U, followed by the communications of \T.
Similarly, the {\em select type} 
$\tselT{\p}$ represents the transmission of label
$l_i$ to role $\p$. 
Label $\l_i$ is chosen in the set $\set{l_i}_{i \in I}$
and the selection prefix is
followed by the communications described by $\T_i$.
The {\em receive} and {\em branch types} are described
as dual types for the send and select \NY{types}, respectively.
In the receive type \NY{$\tinp{\p}{U}T$} we expect that
the typed process will receive a value of type $U$
from role $\p$
while in the branch type $\tbra{\p}{l_i:T_i}_{i \in I}$
we expect a selection label $l_j \in \set{l_i}_{i \in I}$
from role $\p$.
The rest of the local types are the same as global types.
Recursion $\trec{t}{T}$ uses type variables $\vart{t}$ to
perform a recursion via substitution 
\NY{$T\subst{\trec{t}{T}}{\vart{t}}$}. The inactive
type is written as $\tinact$.

We define the roles 
occurring in a global type and the
roles 
occurring in a local type.

\begin{defi}[Roles]
$ $
	\begin{itemize}
		\item	We define $\roles{\G}$ as
			the set of roles in protocol $\G$:
			\[
			\begin{array}{l}
\roles{\inactgt} = \es \qquad \quad \roles{\vargt{t}} = \es \qquad \quad
					\roles{\trec{t}{G}} = \roles{G}
					\\[1mm]
					\roles{\valuegt{\p}{\q}{U} G} = \set{\p, \q} \cup \roles{G}
					\\[1mm]
					\roles{\selgt{\p}{\q}{l_i:G_i}_{i \in I}} = \set{\p, \q} \cup \set{\roles{G_i} \setbar i \in I}
				\end{array}
			\]

		\item	We define $\roles{T}$ on local types as:
		\[
			\begin{array}{l}
				\roles{\tinact} = \es \qquad \qquad \roles{\vart{t}} = \es \qquad \qquad
				\roles{\trec{t}{T}} = \roles{T}
				\\[1mm]
				\roles{\tout{\p}{U} T} = \set{\p} \cup \roles{T} \qquad 
				\roles{\tinp{\p}{U} T} = \set{\p} \cup \roles{T}
				\\[1mm]
				\roles{\tsel{\p}{l_i:T_i}_{i \in I}} = \set{\p} \cup \set{\roles{T_i} \setbar i \in I}
				\\[1mm]
				\roles{\tbra{\p}{l_i:T_i}_{i \in I}} = \set{\p} \cup \set{\roles{T_i} \setbar i \in I}
			\end{array}
		\]
	\end{itemize}
\end{defi}

\paragraph{\bf Global and local projections}
\label{subsec:projection}

The relation between global and local types is formalised by 
the usual projection function \cite{HYC08,BettiniCDLDY08},
where the projection of a global type $G$ over a role $\p$ results
in a local type 
$T$.\footnote{For a simplicity of the presentation, we take the projection function from
  \cite{HYC08,BettiniCDLDY08}, which does not use
  the mergeability operator \cite{DYBH12}. The extension does not affect to the
  whole theory.}

\begin{defi}[Global projection and projection set]
\label{def:projection} \rm
	The projection of a global type $G$ 
	onto a role $\p$ results in a local type
	and it is defined by induction on $G$:
	\[
	\small
	\begin{array}{rcl}
		\proj{\valuegt{\p'}{\q}{\U} \G}{\p}
		& = & \left\{
		\begin{array}{lcl}
			\tout{\q}{\U} \proj{\G}{\p} & \quad & \p = \p'\\
			\tinp{\p'}{\U} \proj{\G}{\p} & \quad & \p = \q \\
			\proj{\G}{\p} & \quad & \textinmath{otherwise}
		\end{array} \right.
		\\[8mm]

		\proj{\selgtG{\p'}{\q}}{\p}
		& = & \left\{
		\begin{array}{lcl}
			\tsels{\q}{\set{l_i:\proj{\G_i}{\p}}_{i \in I}} & \quad & \p = \p'\\
			\tbras{\p'}{\set{l_i:\proj{\G_i}{\p}}_{i \in I}} & \quad & \p = \q \\
			\proj{\G_1}{\p} & \quad & \text{if}\  \forall j\in I.\
              		\proj{\G_1}{\p}=\proj{\G_j}{\p}
		\end{array} \right.
		\\[8mm]

		\proj{(\recgt{t}{\G})}{\p} & = & \left\{
		\begin{array}{lcl}
			\trec{t}{(\proj{\G}{\p})}	& \quad &  \proj{\G}{\p} \not = \vart{t}\\
			\tinact				& \quad & \textinmath{otherwise}
		\end{array} \right.
		\\[6mm]

		\multicolumn{3}{c}{
			\proj{\vargt{t}}{\p} = \tvar{t} \qquad \quad \proj{\inactgt}{\p} = \tinact 
		}
	\end{array}
	\]
	A {\em projection} of $G$ is defined as $\projset{G} = \set{\proj{G}{\p} \setbar \p \in \roles{\G}}$. 
\end{defi}
Inactive $\inactgt$ and recursive variable $\vargt{t}$ types
are projected to their respective local types.
Note that we assume global types have
  roles starting from $1$ up to some $n$, without skipping numbers in
  between.

We project $\valuegt{\p'}{\q}{\U} \G$ to
party $\p$ as a sending local type if $\p = \p'$ and as
a receiving local type if $\p = \q$. In any case the
continuation of the projection is $\proj{\G}{\p}$.

For $\selgtG{\p}{\q}$, the projection is the 
select local type for $\p = \p'$ and the branch local type  $\p = \q$.
Otherwise we use the projection of one of the 
global types $\set{\G_i \setbar i \in I}$ (all types $\G_i$
should have the same projection with respect to $\p$).

The first side condition of the recursive type $\recgt{t}{\G}$ ensures 
that it does not project to an invalid local type $\recgt{t}{\vart{t}}$. 

We use the local projection function
to project a local type $\T$ onto a role
$\p$ to produce binary session types. 
We use the local projection to extract
technical results (well-formed linear environment)
later in this section.

We use the usual binary session types, 
\cite{honda.vasconcelos.kubo:language-primitives,yoshida.vasconcelos:language-primitives}, 
for the definition of local projection:
\begin{defi}[Binary session types]
	\[
		B \bnfis \btout{\U} B \bnfbar \btinp{\U} B \bnfbar \btsel{l_i: B_i} \bnfbar \btbra{l_i:B_i} \bnfbar 
		\trec{t}{B} \bnfbar \vart{t} \bnfbar{\tinact}
	\]
\end{defi}

\begin{defi}[Local projection]
\label{def:localprojection}
	The projection of a local type $T$ 
	onto a role $\p$ is defined by induction on $T$:
	\[
	\begin{array}{rcl}
		\proj{\tout{\p}{\U} T}{\q} & = & \left\{
		\begin{array}{lcl}
			\btout{\U} \proj{\T}{\q} & \quad & \textrm{\q = \p}\\
			\proj{\T}{\q} & \quad & \textrm{otherwise} 
		\end{array}
		\right.
		\\[6mm]

		\proj{\tinp{\p}{\U} T}{\q} & = & \left\{
		\begin{array}{lcl}
			\btinp{\U} \proj{\T}{\q} & \quad & \textrm{\q = \p}\\
			\proj{\T}{\q} & \quad & \textrm{otherwise}
		\end{array}
		\right.
		\\[6mm]

		\proj{\tsel{\p}{l_i: T_i}_{i \in I}}{\q} & = & \left\{
		\begin{array}{lcl}
			\btsel{l_i:\proj{\T_i}{\q}}_{i \in I} & \quad & \textrm{\q = \p}\\
			\proj{\T_1}{\q} & \quad & 
			\text{if} \ \forall i\in I. \proj{\T_i}{\q} = \proj{\T_1}{\q}
		\end{array}
		\right.
		\\[6mm]

		\proj{\tbra{\p}{l_i: T_i}_{i \in I}}{\q} & = & \left\{
		\begin{array}{lcl}
			\btbra{l_i:\proj{\T_i}{\q}}_{i \in I} & \quad & \textrm{\q = \p}\\
			\proj{\T_1}{\q} & \quad & 
			\text{if} \ \forall i\in I. \proj{\T_i}{\q} = \proj{\T_1}{\q}
		\end{array}
		\right.
		\\[6mm]

		\proj{(\recgt{t}{\T})}{\p} & = & \left\{
		\begin{array}{lcl}
			\trec{t}{(\proj{\T}{\p})}	& \quad & \proj{\T}{\p} \not = \vart{t}\\
			\tinact				& \quad & \textinmath{otherwise}
		\end{array} \right.
		\\[6mm]

		\multicolumn{3}{c}{
			\proj{\vart{t}}{\p} = \tvar{t} \qquad \quad
			\proj{\tinact}{\p} = \tinact
		}
	\end{array}
	\]
\end{defi}
An inactive local type, a recursive variable and a recursive type are projected
to their corresponding binary session type syntax. 
\NY{Types $\tout{\p}{\U} T$ and $\tinp{\p}{\U} T$ 
are projected with respect to $\q$
to binary send and receive session types, respectively under
the condition $\p = \q$, 
and continue with the projection of $T$ on $\q$. 
If $\p \not= \q$, the local projection continues with the projection of $T$.
A similar argument is applied for $\tsel{\p}{l_i: T_i}_{i \in I}$ and $\tbra{\p}{l_i: T_i}_{i \in I}$
if $\p = \q$. 
For the case $\p \not= \q$, we project one of the continuations in
$\set{T_i}_{i \in I}$ since we expect all the projections of $\set{T_i}_{i \in I}$
to be the same \cite{HYC08,BettiniCDLDY08}.}

We inductively define the notion of duality as 
a relation over the projected local types:
\begin{defi}[Duality]\rm
We define the duality function over binary session types as:
\[
\begin{array}{c}
	\tinact  =  \dual{\tinact} \quad \vart{t} = \dual{\vart{t}} \quad
	\dual{\trec{t}{B}} = \trec{t}{\dual{B}} \quad
	\dual{\btout{U} B} = \btinp{U} \dual{B} \quad
	\dual{\btinp{U} B} = \btout{U} \dual{B} \\
	\dual{\btsel{l_i:B_i}_{i \in I}} = \btbra{l_i:\dual{B_i}}_{i \in I} \quad
	\dual{\btbra{l_i:B_i}_{i \in I}} = \btsel{l_i:\dual{B_i}}_{i \in I}
\end{array}
\]
\end{defi}
We assume only session types with tail recursion as in \cite{HYC08,BettiniCDLDY08} 
(note that the inductive duality
on non-tail recursive session types,  
i.e.~session prefixes that carry a recursive
variable as an object) is shown to be unsound \cite{DBLP:journals/corr/BernardiH13}). 

The result of the following proposition is used on the well-formedness
criteria of a linear environment.
\begin{prop}
	\label{prop:bindual}
	If $\p, \q \in \roles{G}$ with $\p \not= \q$ then
	$\proj{(\proj{\G}{\p})}{\q} = \dual{\proj{(\proj{\G}{\q})}{\p}}$. 
\end{prop}
\begin{proof}
	The proof is done by induction on the structure of global types.
\end{proof}

\subsection{Typing system}
\label{subsec:typing}
We define the typing \NY{system} for
the synchronous multiparty session calculus.
The typing judgements for expressions and processes are of the
shapes:
\[
	\Gproves{\e}{\So} \quad \textinmath{ and } \quad \Gtprocess{\PP}{\De}
\]
where $\Gamma$ is the shared environment 
which associates variables to sort types (i.e.~base types or 
global types), shared names to
global types and process variables to session environments;
and  $\Delta$ is the session environment (or linear environment) 
which associates channels to session types.
Formally we define: 
\[
	\Ga \bnfis \es \bnfbar \Gacat{\typed{u} \So} \bnfbar \Gacat{\typed{\varp{X}} \De}
	\quad \textinmath{and} \quad
	\De \bnfis \es \bnfbar \Decat{\typed{c} \T}
\]
assuming we can write $\Gacat{\typed{u} \So}$ if $u\not\in\dom{\Gamma}$. 
We extend this to a concatenation for typing environments as 
$\De \cat \De'=  \De \cup \De'$.
We use the following definition to declare the coherence of
session environments.
\begin{defi}[Coherency]\rm 
\label{def:coherency}
	Typing $\De$ is {\em coherent with respect to session $s$}
	(notation $\cohses{\De}{s}$) if
	for all $\srole{s}{\p}: T_\p\ \in \De$ there exists
	$\srole{s}{\q}: T_\q \in \De$ such that
	$\proj{T_\p}{\q} = \dual{\proj{T_\q}{\p}}$.
	A typing $\De$ is {\em coherent}
	(notation $\coherent{\De}$) if 
	it is coherent with respect to all $s$ in its domain.
We say a typing $\De$ is {\em fully coherent}
(notation $\fcoherent{\De}$) if
it is coherent and
if $\srole{s}{\p}:T_\p \in \De$ then
for all	$\q \in \roles{T_\p}$, 
$\srole{s}{\q} : T_\q \in \De$.
\end{defi}

\begin{figure}
\newcommand{\typvspace}{\\[5mm]}
\[
\small
\begin{array}{cc}
	\Eproves{\Gacat{\typed{u} S}}{u}{S} \quad \trule{Name} 
&
	\Gproves{\true,\false}{\bool} \quad \trule{Bool}
	\typvspace

	\tree{\Gproves{\e_i}{\bool}}{\Gproves{\e_1 \LogAnd \e_2}{\bool}} \ \trule{And}
&
	\tree{
		\NY{\Gproves{n_1}{U} \quad \Gproves{n_2}{U}}
	}{
		\NY{\Gproves{n_1 = n_2}{\bool}}
	}\ \trule{Match}
	\typvspace

	\tree{
		\begin{array}{l}
			\Gproves{\Ia}{\gG} \quad \Gtprocess{\PP}{\Decat \typed{\x} \proj{\G}{\p}\\
			\max(\roles{G})= \p}
		\end{array}
	}{
		\Gtprocess{\req{\Ia}{\p}{\x}{\PP}}{\De}
	} \ \trule{MReq}
&
	\tree{
	\begin{array}{l}
		\Gproves{\Ia}{\gG} \quad \Gtprocess{\PP}{\Decat \typed{\x} \proj{\G}{\p}}\\
		1 \leq \p < \max(\roles{G})
	\end{array}
	}{
		\Gtprocess{\acc{\Ia}{\p}{\x}{\PP}}{\De}
	} \ \trule{MAcc}
	\typvspace

	\tree{
		\Gproves{\e}{\So} \quad \Gtprocess{\PP}{\Decat \typed{c} \T}
	}{
		\Gtprocess{\sout{c}{\q}{\e} \PP}{\Decat \typed{c} \tout{\q}{\So} \T}
	} \ \trule{Send}
&
	\tree{
		\tprocess{\Gacat \typed{\x} \So}{\PP}{\Decat \typed{c} \T}
	}{
		\Gtprocess{\sinp{c}{\q}{\x} \PP}{\Decat \typed{c} \tinp{\q}{\So} \T}
	} \ \trule{Recv}
	\typvspace

	\tree{
		\Gtprocess{\PP}{\Decat \typed{c} \T}
	}{
		\Gtprocess{\sout{c}{\q}{c'} \PP}{\Decat \typed{c} \tout{\q}{\T'} \T \cat \typed{c'} \T' }
	} \ \trule{Deleg}
&
	\tree{
		\Gtprocess{\PP}{\Decat \typed{c} \T \cat \typed{x} \T'}
	}{
		\Gtprocess{\sinp{c}{\q}{\x} \PP}{\Decat \typed{c} \tinp{\q}{\T'} \T}
	} \ \trule{SRecv}
	\typvspace

	\tree{
		\Gtprocess{\PP}{\Decat \typed{c} \T}
	}{
		\Gtprocess{\ssel{c}{\q}{l} \PP}{\Decat \typed{c} \tsel{\q}{l:T}}
	} \ \trule{Sel}
&
	\tree{
		\Gtprocess{\PP_i}{\Decat \typed{c} \T_i \quad \forall\ i \in I}
	}{
		\Gtprocess{\sbraP{c}{\q}}{\Decat \typed{c} \tbraT{\q}}
	} \ \trule{Bra}
	\typvspace

	\tree{
		\Gtprocess{\PP_1}{\De_1} \quad \Gtprocess{\PP_2}{\De_2} \quad \dom{\De_1} \cap \dom{\De_2} = \es
	}{
		\Gtprocess{\PP_1 \Par \PP_2}{\De_1 \cat \De_2}
	} \ \trule{Conc}
&
	\tree{
		\Gproves{\e}{\bool} \quad \Gtprocess{\PP}{\De} \quad \Gtprocess{\Q}{\De}
	}{
		\Gtprocess{\IfThenElse{\e}{\PP}{\Q}}{\De}
	} \ \trule{If}
	\typvspace

	\Gtprocess{\inact}{\es} \quad \trule{Inact}
&
	\tree{
		\Gtprocess{\inact}{\De}
	}{
		\Gtprocess{\inact}{\De \cat \typed{c} \tinact}
	} \ \trule{Complete}
	\typvspace


	\tree{
		\tprocess{\Gacat \typed{a}\gG}{\PP}{\De}
	}{
		\Gtprocess{\news{a} \PP}{\De}
	} \ \trule{NRes} 
&
	\tree{
		\begin{array}{c}
			\fcoherent{\set{\typedrole{\s}{1} T_1 \dots  \typedrole{\s}{n} \T_n}}\\
			\Gtprocess{\PP}{\Decat \typedrole{\s}{1} T_1 \dots  \typedrole{\s}{n} \T_n } 
		\end{array}
	}{
		\Gtprocess{\news{\s} \PP}{\De}
	} \ \trule{SRes}
	\typvspace

	\tprocess{\Gacat{\typed{\varp{X}} \De}}{\varp{X}}{\De} \quad \trule{Var}
&
	\tree{
		\tprocess{\Gacat{\typed{\varp{X}} \De}}{\PP}{\De}
	}{
		\Gtprocess{\rec{X}{\PP}}{\De}
	} \ \trule{Rec} 
\end{array}
\]
\caption{
	\label{fig:synch-typing}
	Typing system for synchronous multiparty session calculus 
}
\end{figure}

Figure~\ref{fig:synch-typing} defines the typing system.
The typing rules presented here are essentially identical to 
the communication typing system for programs 
in \cite{BettiniCDLDY08}, 
due to the fact that our calculus is synchronous 
(i.e.~we do not use session endpoint queues).

Rule $\trule{Name}$ types a shared name or
a shared variable in environment $\Ga$.
Rule $\trule{Bool}$ assigns the 
type $\bool$ to constants $\true, \false$.
Logical expressions are also typed with the $\bool$
type via rule $\trule{And}$, etc. 
\NY{Rule $\trule{Match}$
ensures that the name matching operator has
the boolean type.}

Rules $\trule{MReq}$ and $\trule{MAcc}$ type the
request and accept processes,
respectively. Both rules check that the type of session variable
$x$ agrees with the global type of the session initiation name $a$
that is
projected on the corresponding role $\p$.
Furthermore, in rule $\trule{MReq}$ we require that
the initiating role $\p$ is the maximum among the roles of the
global type $G$ of $a$ while in rule $\trule{MAcc}$ we require
that role $\p$ is less than the maximum role of global type $G$.
Note that these session initiation rules allow processes to contain 
some but not all of the roles in a session. 
Rules $\trule{MReq}$ and $\trule{MAcc}$ ensure that a parallel composition
of processes that implement all the roles in a session can proceed
with a sound session initiation.

Rules $\trule{Send}$ and $\trule{Recv}$ are used to type the 
send, $\sout{c}{\p}{\va} P$, and receive, $\sinp{c}{\p}{x} P$, session prefixes.
Both rules prefix the local type of $c$ in the linear environment
the send type $\tout{\p}{U} T$, and receive type $\tinp{\p}{U} T$, respectively.
The typing is completed with the check of the object types
$v$ and $x$, respectively, in the shared environment $\Ga$. 
The delegation of a session endpoint is typed under rules $\trule{Deleg}$ and $\trule{Srecv}$.
Both rules prefix the local type of $c$ in the linear environment
with the send and receive 
prefixes, respectively, in a similar way with rules
$\trule{Send}$ and $\trule{Recv}$.
They check 
the type for the delegating object in the linear environment $\De$, and 
a delegation respects the linearity of the
delegating endpoint (in this case $c'$ and $x$, respectively), 
i.e.~when an endpoint is sent ($\trule{Deleg}$) it should not be present
in the linear environment of the continuation $P$. Similarly, when an
endpoint is received ($\trule{Srecv}$) the receiving endpoint should not
be present before the reception.


Rules $\trule{Sel}$ and $\trule{Bra}$ type selecting and 
branching processes, respectively. 
A selection prefix is typed on a select local type, while
a branching prefix is typed on a branch local type.
A selection prefix with label $l$ for role $c$ uses label
$l$ to select the continuation on name $c$ in the select local type.
A branching process with labels $l_i$ branches the local types 
of $c$ in the corresponding $P_i$ in the branch local type.
Furthermore, all $P_i$ processes should have the same linear type
on names other than $c$.

Rule $\trule{Conc}$ types a parallel composition of processes.
The disjointness condition on typing environments $\De_1$ and $\De_2$
ensures the linearity of the environment $\De_1 \cat \De_2$.
Rule $\trule{If}$ types conditional process, 
where we require that the expression $\e$ to be of $\bool$ 
type and that the branching processes have the
same linear environment.
Rule
$\trule{Inact}$ types the empty process with the empty linear
environment. Rule $\trule{Complete}$ does an explicit weakening on the linear typing of 
an inaction process to achieve a complete linear environment. A complete
linear environment is defined as the linear typing where
every session endpoint is mapped to the inactive local type $\tinact$.
Rule $\trule{Nres}$
defines the typing for shared name restriction. A restricted shared name should 
be present in the shared environment $\Ga$ before restriction and should 
not appear in $\Ga$ after the restriction. 
Rule $\trule{Sres}$
uses the full coherency property to restrict a session name. Full
coherency requires that all session endpoints of a session are present
in the linear environment before restriction
and furthermore, it requires that their local projections are mutually dual. A
restricted session name does not appear in the domain of the linear
environment.

Rule $\trule{Var}$ returns the linear environment $\De$
that is assigned to process variable $\varp{X}$ in environment $\Ga$.
Rule $\trule{Rec}$ checks that the process $P$ 
under the recursion has the same
linear environment $\De$ as the recursion variable $\varp{X}$.

Further examples of typing and typable processes can be found in 
\cite{CDYP13}.

Finally, we call the typing judgement $\Gtprocess{\PP}{\De}$ 
{\em coherent} if $\coherent{\De}$. 



\subsection{Type soundness}
\label{subsec:typesoundness}

We proceed with the proof of a subject reduction
theorem to show the soundness of the typing system.

Before we state the subject reduction theorem we
define the reduction semantics for local types extended
to include session environments.
The reduction on a session environment of a process 
shows the change on the session environment
after a possible reduction on the process.
We use the approach from
\cite{HYC08,BettiniCDLDY08} to define session environment
reduction.
\begin{defi}[Session environment reduction]
\label{def:lin_red}
$ $
	\begin{enumerate}
		\item	$\set{\typedrole{s}{\p} \tout{\q}{U} T \cat \typedrole{s}{\q} \tinp{\p}{U} T'}
			\typingred
			\set{\typedrole{s}{\p} T \cat \typedrole{s}{\q} T'}$.

		\item	$\set{\typedrole{s}{\p} \tsel{\q}{l_i:T_i}_{i \in I} \cat
			\typedrole{s}{\q} \tbra{\p}{l_j:T'_j}_{j \in J}}
			\typingred 
			\set{\typedrole{s}{\p} T_k \cat \typedrole{s}{\q} T'_k}\quad I \subseteq J, k \in I$.

		\item	$\De \cup \De' \typingred \De \cup \De''$ if
			$\De' \typingred \De''$. 
	\end{enumerate}
\end{defi}
%
Note that the second rule of the session environment
reduction makes the reduction $\Delta \red^* \Delta'$
non-deterministic (i.e.~not always confluent).
%
The typing system satisfies 
the subject reduction theorem \cite{BettiniCDLDY08}: 
\begin{thm}[Subject reduction]
	\label{the:subject}
	If $\Gtprocess{P}{\De}$ is coherent and $P \Red P'$ then 
	$\Gtprocess{P'}{\De'}$ is coherent with $\De \typingred^\ast \De'$.
\end{thm}

\begin{proof}
	See Appendix~\ref{app:sr}.
\end{proof}

\section{Synchronous Multiparty Session Semantics}
\label{sec:synchronous}
This section presents the session typed behavioural 
theory for synchronous multiparty sessions.
The typed bisimulation uses a
labelled transition system (LTS) on environment tuples
($\Gamma,\Delta$) to control the behaviour of untyped processes. 
The LTS on environments introduces a constrain that captures 
accurately multiparty session interactions and lies
at the heart of the session typed semantics.
The bisimulation theory presented in this section is extended in the 
next section to define a bisimulation theory that uses 
a more fine-grained LTS, defined using 
the additional typing information of the global observer.

\subsection{Labelled transition system}
\label{sec:LTS}

We use the following labels ($\ell, \ell', \dots$)
to define the labelled transition system:
\begin{eqnarray*}
	\ell	&\bnfis& \actreq{\Ia}{\A}{\s} \bnfbar 
		\actacc{\Ia}{\A}{\s} \bnfbar \actout{\s}{\p}{\q}{\va} 
		\bnfbar \actbout{\s}{\p}{\q}{\Ia}  \\
		& \bnfbar &\actbout{\s}{\p}{\q}{\srole{\s'}{\p'}} \bnfbar
		\actinp{\s}{\p}{\q}{\va} \bnfbar \actsel{\s}{\p}{\q}{l} \bnfbar 
		\actbra{\s}{\p}{\q}{l} \bnfbar \tau
\end{eqnarray*}
Note that label $\actout{\s}{\p}{\q}{\srole{\s'}{\p'}}$ is subsumed in 
label $\actout{\s}{\p}{\q}{\va}$. 
Symbol $A$ denotes a {\em role set}, which is a set of 
roles.
Labels $\actreq{\Ia}{\A}{\s}$ and $\actacc{\Ia}{\A}{\s}$ define
the accept and request of a fresh session $\s$ 
by roles in set $\A$, respectively.
Actions on session channels are denoted with the labels 
$\actout{\s}{\p}{\q}{\va}$ and $\actinp{\s}{\p}{\q}{\va}$ 
for output and input of value $\va$
from $\p$ to $\q$ on session $\s$. Bound output values can
be shared channels or session endpoints (delegation).
$\actsel{\s}{\p}{\q}{l}$ and $\actbra{\s}{\p}{\q}{l}$ define
the select and branch, respectively. Label $\tau$ is
the unobserved transition.

Dual label definition is used to define the parallel
rule in the label transition system:
\begin{defi}[Dual Labels]
	\label{def:comp}
We define a duality relation $\comp$ between two labels which is 
generated by the following axioms and synmetric ones:
	\[
		\actout{\s}{\p}{\q}{\va} \comp \actinp{\s}{\q}{\p}{\va} \qquad \quad
		\actbout{\s}{\p}{\q}{\va} \comp \actinp{\s}{\q}{\p}{\va} \qquad \quad
		\actsel{\s}{\p}{\q}{l} \comp \actbra{\s}{\q}{\p}{l}
	\]
\end{defi}
Dual labels are input and output (respectively select and branch)
on the same session channel and on complementary roles. For
example, in $\actout{\s}{\p}{\q}{\va}$ and $\actinp{\s}{\q}{\p}{\va}$,
role $\p$ sends to $\q$ and role $\q$ receives from $\p$.

We define the notion of complete role set, used for
defining session initiation transition semantics later:
\begin{defi}[Complete role set]\rm
	We say the role set $A$ is {\em complete with respect to $n$}
	if $n = \maxset{\A}$ and $\A = \set{1, 2,\dots, n}$.
\end{defi}
A complete role set means that all global protocol 
participants are present in the set.
For example, $\set{1,3,4}$ is not complete, but
$\set{1,2,3,4}$ is with respect to 4 and not complete with number $n >
4$.
We use $\fn{\ell}$ and $\bn{\ell}$ to denote 
a set of free and bound names in $\ell$ and set
$\mathsf{n}(\ell)=\bn{\ell}\cup \fn{\ell}$. 

\paragraph{\bf Labelled transition system for processes\ }
\label{sec:tupedLTS}

\begin{figure}[t]
\[
	\begin{array}{c}
		\begin{array}{lrclllrcl}
			\ltsrule{Req} & \req{\Ia}{\p}{\x} \PP &
        		\trans{\actreqs{\Ia}{\p}{\s}} & \PP \subst{\s\role{\p}}{x}
			& \qquad &
			\ltsrule{Acc} & \acc{\Ia}{\p}{\x} \PP &
        		\trans{\actaccs{\Ia}{\p}{\s}} & \PP \subst{\s\role{\p}}{x}
			\\[2mm]

			\ltsrule{Send} & \out{\s}{\p}{\q}{e} \PP &
        		\trans{\actout{\s}{\p}{\q}{\va}} & \PP \quad (e\downarrow v)
			& \qquad &
			\ltsrule{Rcv} & \inp{\s}{\p}{\q}{\x} \PP &
        		\trans{\actinp{\s}{\p}{\q}{\va}} &  \PP \subst{\va}{x}
			\\[2mm]


			\ltsrule{Sel} & \sel{\s}{\p}{\q}{l} \PP &
		        \trans{\actsel{\s}{\p}{\q}{l}} & \PP
			& \qquad &
			\ltsrule{Bra} & \braP{\s}{\p}{\q}
        		& \trans{\actbra{\s}{\p}{\q}{l_k}} & \PP_k
		\end{array}
		\\[6mm]

		\begin{array}{c}
			\ltsrule{Tau} \quad
			\tree{
				\PP \trans{\ell} \PP' \quad \Q \trans{\ell'} \Q' \quad \ell \comp \ell'
			}{
				\PP \Par \Q \trans{\tau} \newsp{\bn{\ell} \cap \bn{\ell'}}{\PP' \Par \Q'}
			}
			\qquad \quad
		 	\ltsrule{Par} \quad
			\tree{
				\PP \trans{\ell} \PP' \quad \bn{\ell} \cap \fn{\Q} = \es
			}{
				\PP \Par \Q \trans{\ell} \PP' \Par \Q
			}
			\\[6mm]

		 	\ltsrule{Res} \ 
			\tree{
				\PP \trans{\ell} \PP' \quad 
n \notin \fn{\ell}
			}{
				\news{n} \PP \trans{\ell} \news{n} \PP'
			}
			\\[6mm]
			\ltsrule{OpenS} \ 
			\tree{
				\PP
                                \trans{\actout{\s}{\p}{\q}{\srole{\s'}{\p'}}} \PP' \quad s\not = s'
			}{
				\news{\s'} \PP \trans{\actbout{\s}{\p}{\q}{\srole{\s'}{\p'}}} \PP'
			}\qquad
		 	\ltsrule{OpenN} \ 
			\tree{
				\PP \trans{\actout{\s}{\p}{\q}{\Ia}} \PP'
			}{
				\news{\Ia} \PP \trans{\actbout{\s}{\p}{\q}{\Ia}} \PP'
			}
			\\[6mm]

		 	\ltsrule{Alpha} \quad
			\tree{
				\PP \acong \PP' \quad \PP' \trans{\ell} \Q' 
			}{
				\PP \trans{\ell} \Q
			}
			\qquad \quad
			\ltsrule{AccPar} \quad
				\tree{
					\PP_1 \trans{\actacc{\Ia}{\A}{\s}} \PP_1' \quad \PP_2 \trans{\actacc{\Ia}{\A'}{\s}} \PP_2' \quad \A \cap \A'= \es
				}{
					\PP_1 \Par \PP_2 \trans{\actacc{\Ia}{\A \cup \A'}{\s}} \PP_1' \Par \PP_2'
				}
		\end{array}
		\\[6mm]

		\begin{array}{lcl}
			\ltsrule{ReqPar} &
			\tree{
				\PP_1 \trans{\actacc{\Ia}{\A}{\s}} \PP_1' \quad \PP_2
			        \trans{\actreq{\Ia}{\A'}{\s}} \PP_2' \quad \A \cap \A'= \es, \ \A \cup \A'  \textinmath{ not complete w.r.t } \max{(A')}
			}{
				\PP_1 \Par \PP_2 \trans{\actreq{\Ia}{\A \cup \A'}{\s}} \PP_1' \Par \PP_2'
		}
		\\[6mm]

			\ltsrule{TauS} &
			\tree{
				\PP_1 \trans{\actacc{\Ia}{\A}{\s}} \PP_1' \quad
	                        \PP_2 \trans{\actreq{\Ia}{\A'}{\s}} \PP_2' \quad
	                        \A \cap \A'= \es, \ \A \cup \A'  \textinmath{ complete w.r.t } \max{(A')}
			}{
				\PP_1 \Par \PP_2 \trans{\tau} \newsp{\s} {\PP_1' \Par \PP_2'}
			}
		\end{array}
	\end{array}
\]
We omit the synmetric case of $\ltsrule{Par}$ and conditionals.
\caption{
	\label{fig-synch-lts}
Labelled transition system for processes 
}
\end{figure}

Figure \ref{fig-synch-lts} defines the untyped labelled transition system.
Rules $\ltsrule{Req}$ and $\ltsrule{Acc}$ define 
that processes $\req{a}{\p}{x} P$ and $\acc{a}{\p}{x} P$
produce the accept and request
labels, respectively for a fresh session $s$ on role $\p$.
Rules $\ltsrule{Send}$ and $\ltsrule{Rcv}$ predict that
processes $\out{s}{\p}{\q}{v} P$ and $\inp{s}{\p}{\q}{x} P$
produce the send and
receive label, respectively for value $v$ from role $\p$ to role $\q$ 
in session $s$. 
Similarly, rules $\ltsrule{Sel}$ and $\ltsrule{Bra}$
define that the select and branch labels are observed on
processes $\sel{s}{\p}{\q}{l} P$ and $\bra{s}{\p}{\q}{l_i: P_i}$ respectively.

The last three rules collect and 
synchronise the multiparty participants together. 
Rule $\ltsrule{AccPar}$ accumulates the accept participants 
and records them into role set $A$. 
Rule $\ltsrule{ReqPar}$ accumulates the accept participants and
the request participant into role set $A$. 
Note that the request 
action role set always includes the maximum 
role number among the participants. 

Finally, rule $\ltsrule{TauS}$ checks that a role set
is complete (i.e.~all roles are present), 
thus a new session can be created under the
$\tau$-action (synchronisation).
Other rules follow the usual inductive rules for the
$\pi$-calculus labelled transition system.

Rule $\ltsrule{Tau}$ synchronises two processes that
exhibit dual labels, while rules $\ltsrule{Par}$ and  
$\ltsrule{Res}$ close the labelled transition system under 
the parallel composition and name restriction
operators. Note that $\ltsrule{Par}$ allows the output transition to the role
$\q$ even the endpoint at $\q$ is in $Q$. This will be disallowed by 
the environment LTS defined in Figure~\ref{fig-synch_envtrans} later. 

Rules $\ltsrule{OpenS}$ and $\ltsrule{OpenN}$
are used for name extrusion. Finally, rule $\ltsrule{Alpha}$
closes the LTS under structural congruence. 

We write $\Trans{}$ for the reflexive and
transitive closure of $\trans{}$, $\Trans{\ell}$ for the transitions
$\Trans{}\trans{\ell}\Trans{}$ and $\Trans{\hat{\ell}}$ for $\Trans{\ell}$ if
$\ell\not = \tau$ otherwise $\Trans{}$. 

\paragraph{\bf Typed labelled transition relation}

We define the typed LTS on the basis of the untyped one.
This is realised by defining 
an environment labelled transition system,
defined in Figure~\ref{fig-synch_envtrans}, which
uses the same labels defined for the untyped
LTS.
We write $(\Ga, \De)\trans{\ell} (\Ga', \De')$ as
the notation where an environment $(\Ga, \De)$ 
allows an action $\ell$ to take place, 
resulting in environment $(\Ga', \De')$.

\begin{figure}[t]
\[
	\begin{array}{cc}
\small
		\tree {
			\Gamma(a) = \globaltype{G} \quad s \notin \dom{\De} 
		}{
			(\Ga, \De) \envtrans{\actacc{a}{A}{s}}
			(\Ga, \De \cat \set{\srole{\s}{i}: \proj{G}{i}}_{i \in A})
		}
\ \erule{Req}

\qquad
		\tree {
			\Gamma(a) = \globaltype{G} \quad s \notin \dom{\De} 
		}{
			(\Ga, \De) \envtrans{\actreq{a}{A}{s}}
			(\Ga, \De \cat \set{\srole{\s}{i}: \proj{G}{i}}_{i \in A})
		}
\erule{Acc}
		\\[9mm]

		\tree {
			\Gproves{v}{U} \quad \srole{s}{\q} \notin \dom{\De}\hspace{-1cm}
		}{
			(\Ga, \De \cat \typedrole{s}{\p} \tout{\q}{U} T) \envtrans{\actout{s}{\p}{\q}{v}}
			(\Ga, \De \cat \typedrole{s}{\p} T)
		}
\erule{Send}
\\[9mm]

		\tree {
			\srole{s}{\q} \notin \dom{\De} \quad a\not\in \dom{\Gamma}\hspace{-1cm}
		}{
			(\Ga, \De \cat \typedrole{s}{\p} \tout{\q}{U} T) \envtrans{\actbout{\s}{\p}{\q}{\Ia}}
			(\Ga\cat{a}:{U}, \De \cat \typedrole{s}{\p} T)
		}
\erule{OpenN}
		\\[9mm]

		\tree {
			\Gproves{v}{U} \quad \srole{s}{\q} \notin \dom{\De}
		}{
			(\Ga, \De \cat \typedrole{s}{\p} \tinp{\q}{U} T)
			\envtrans{\actinp{s}{\p}{\q}{v}}
			(\Ga, \De \cat \typedrole{s}{\p} T)
		}
\erule{Recv}
\\[9mm]
		\tree {
			a \not\in \dom{\Gamma} \quad \srole{s}{\q} \notin \dom{\De}
		}{
			(\Ga, \De \cat \typedrole{s}{\p} \tinp{\q}{U} T)
			\envtrans{\actinp{s}{\p}{\q}{a}}
			(\Ga\cat a:U, \De \cat \typedrole{s}{\p} T)
		}
\erule{RecvN}
		\\[9mm]

		\tree {
			\srole{s}{\q} \notin \dom{\De}
		}{
			(\Ga, \De \cat \typedrole{s'}{\p'} T'\cat \typedrole{s}{\p} \tout{\q}{T'} T)
			\envtrans{\actout{s}{\p}{\q}{s'[\p']}}
			(\Ga, \De  \cat \typedrole{s}{\p} T) 
		}
\erule{SendS}
\\[9mm]		\tree {
			\srole{s}{\q} \notin \dom{\De}
		}{
			(\Ga, \De  \cat \typedrole{s}{\p}\tout{\q}{T'} T)
			\envtrans{\actbout{s}{\p}{\q}{\srole{s'}{\p'}}}
			(\Ga, \De \cat \typedrole{s}{\p} T\cat  \set{\typedrole{s'}{\p_i} T_i})
		}
\erule{OpenS}
\\[9mm]

		\tree {
			\srole{s}{\q}, \srole{s'}{\p'}  \notin \dom{\De}
		}{
			(\Ga, \De \cat \typedrole{s}{\p} \tinp{\q}{T'} T)
			\envtrans{\actinp{s}{\p}{\q}{\srole{s'}{\p'}}}
			(\Ga, \De \cat \typedrole{s'}{\p'} T' \cat \typedrole{s}{\p} T)
		}
\erule{RecvS}
\\[9mm]

		\tree {
			\srole{s}{\q} \notin \dom{\De}
		}{
			(\Ga, \De \cat \typedrole{s}{\p} \tselT{\q})
			\envtrans{\actsel{s}{\p}{\q}{l_k}}
			(\Ga, \De \cat \typedrole{s}{\p} T_k)
		}
\erule{Sel}
		\\[9mm]

		\tree {
			\srole{s}{\q} \notin \dom{\De}
		}{
			(\Ga, \De \cat \typedrole{s}{\p} \tbraT{\q})
			\envtrans{\actbra{s}{\p}{\q}{l_k}}
			(\Ga, \De \cat \typedrole{s}{\p} T_k)
		}
\erule{Bra}
\qquad
		\tree {
			\De \red \De' \quad \mbox{or} \quad \De = \De'
		}{
			(\Ga, \De) \envtrans{\tau} (\Ga, \De')
		}
\erule{Tau}
	\end{array}
\]
\caption{
	\label{fig-synch_envtrans}
	Labelled transition system for environments  
}
\end{figure}

The intuition for this definition is that
the observables on session channels occur when 
the corresponding endpoint is not present in the 
linear typing environment $\De$,
and the type of an action's object respects 
the environment $(\Ga, \De)$.
In the case when new names are created or 
received, the environment $(\Ga, \De)$ is extended 
according to the new name.

Rule \erule{Req} says that a reception of a message via
$a$ is possible when $a$'s type $\globaltype{G}$ is recorded into $\Gamma$ 
and the resulting session environment records 
projected types from $G$ ($\set{\srole{\s}{i}: \proj{G}{i}}_{i \in A}$).  
Rule \erule{Acc} is for the send of a message via $a$ and it is dual to the first rule. 
The next four rules are free value output \erule{Send}, bound name output \erule{OpenN}, 
free value input \erule{Recv} and name input \erule{RecvN}.   
The rest of rules are free session output \erule{SendS}, bound session output \erule{OpenS}, and 
session input \erule{RecvS} as well as selection \erule{Sel} and branching \erule{Bra} rules.  
The bound session output \erule{OpenS} records a set of session types $s'[\p_i]$ at opened
session $s'$. 
Rule \erule{Tau} follows the reduction rules for linear
session environment defined
in \S~\ref{subsec:typesoundness}
($\Delta=\Delta'$ is the case for the reduction at hidden sessions).
Note that if $\Delta$ already contains destination ($s[\q]$), the 
environment cannot perform the visible action, but only
the final $\tau$-action. 

The typed LTS requires that a process can perform an untyped
action $\ell$ and that its typing environment $(\Ga, \De)$ can
match the action $\ell$.

\begin{defi}[Typed transition]\rm
\label{def:lts}
	A {\em typed transition relation} is a typed relation
	$\typedprocess{\Gamma_1}{P_1}{\Delta_1} \trans{\ell}\typedprocess{\Gamma_2}{P_2}{\Delta_2}$
	if (1) $P_1 \trans{\ell}P_2$ and 
	(2) $(\Gamma_1,\Delta_1)\trans{\ell}(\Gamma_2, \Delta_2)$
	with $\typedprocess{\Gamma_i}{P_i}{\Delta_i}$.
\end{defi}
%


\subsection{Synchronous multiparty behavioural theory}
We begin with the definition of the typed relation 
as the binary relation over closed, coherent and typed processes.
\begin{defi}[Typed relation]\rm
\label{def:typerel}
	A relation $\mathcal{R}$ as {\em typed relation} if
	it relates two closed, coherent typed terms, written:
\[
	\Gtprocess{P_1}{\De_1}\ \mathcal{R}\ \Gtprocess{P_2}{\De_2}
\]
	We often write $\Gtprocess{P_1}{\De_1}\ \mathcal{R}\ \noGtprocess{P_2}{\De_2}$.
\end{defi}
Next we define the {\em barb} \cite{DBLP:journals/tcs/AmadioCS98} with respect to the judgement:
\begin{defi}[Barbs]\rm
	We write:
	\begin{itemize}
		\item	$\Gtprocess{P}{\De} \barbout{s}{\p}{\q}$ if
			$P \equiv \newsp{\vect{a}\vect{s}}{\out{\s}{\p}{\q}{v} R \Par Q}$ with 
$s \notin \vect{s}$ and $\srole{s}{\q} \notin \dom{\De}$.
		\item	$\Gtprocess{P}{\De} \barbreq{a}$ if 
			$P \equiv \newsp{\vect{a}\vect{s}}{\req{a}{n}{x} R \Par Q}$ with 
			$a \notin \vect{a}$. Then we write $m$ for either $a$ or $s[\p][\q]$.
	\end{itemize}
	We define $\Gtprocess{\PP}{\De} \Barb{m}$ if there exists 
	$Q$ such that  $\PP \Red Q$ and $\Gtprocess{Q}{\De'} \barb{m}$.
\end{defi}

The set of contexts is defined as follows: 
\[
\begin{array}{lcl}
	C &\bnfis& \hole \bnfbar C \Par P \bnfbar P \Par C \bnfbar \news{n} C \bnfbar 
	\IfThenElse{\e}{C}{C'} \bnfbar \rec{X}{C} \bnfbar\\
	& & \bout{\s}{\va}{C} \bnfbar \binp{\s}{\x}{C} \bnfbar \bsel{\s}{l}{C} \bnfbar
	\bbras{\s}{\set{l_i:C_i}_{i \in I}} \bnfbar \breq{\Ia}{\x} C \bnfbar \bacc{\Ia}{\x} C
\end{array}
\]
\noi $\context{C}{\PP}$ substitutes 
process $\PP$ for each hole ($\hole$) in context $C$.
We say $\context{C}{\PP}$ is {\em closed} if
 $\fv{\context{C}{\PP}}=\emptyset$. 
%
\begin{defi}[Linear environment convergence]\rm 
	We write $\De_1 \bistyp \De_2$ if there exists $\De$ such that
	$\De_1 \typingred^\ast \De$ and $\De_2 \typingred^\ast \De$. 
\end{defi}
We expect processes with the {\em same} behaviour to have linear environments
that converge since they should follow the same session behaviour.

Note that the convergence condition is not related with the
fact that reduction on a linear environment is non-deterministic 
(see Definition~\ref{def:lin_red}) and two 
linear environments $\De_1$ and $\De_2$ which are non-deterministic 
may converge. 

We define the reduction-closed congruence based on the 
definition of barb and \cite{HondaKYoshida95}. 

\begin{defi}[Reduction-closed congruence]\rm 
\label{def:src}
	A typed relation $\relfont{R}$ is a {\em reduction-reduction congruence}
	if it satisfies the following conditions for each 
	$\rel{\Gtprocess{P_1}{\De_1}}{R}{\noGtprocess{P_2}{\De_2}}$:
	\begin{enumerate}
		\item	$\De_1 \bistyp \De_2$ 
		\item	$\Gtprocess{P_1}{\De_1} \Barb{m}$ iff 
			$\Gtprocess{P_2}{\De_2} \Barb{m}$.
		\item	
			\begin{itemize}
				\item	$P_1 \Red P_1'$ implies that there exists $P_2'$
					such that  $P_2 \Red P_2'$ and
					$\rel{\Gtprocess{P_1'}{\De_1'}}{R}{\noGtprocess{P_2'}{\De_2'}}$
					with $\De_1' \bistyp \De_2'$. 
				\item	the symmetric case.
			\end{itemize}
		\item	For all closed context $\C$ and for all
                  $\De_1'$ and $\De_2'$,  
			such that $\Gtprocess{\Ccontext{P_1}}{\De_1'}$ and 
			$\Gtprocess{\Ccontext{P_2}}{\De_2'}$ then
                        $\De_1'\bistyp \De_2'$ and  
			$\rel{\Gtprocess{\Ccontext{P_1}}{\Delta_1'}}{R}{\Gtprocess{\Ccontext{P_2}}{\Delta_2'}}$. 
\end{enumerate}
	The union of all reduction-closed congruence relations is denoted as 
$\cong^{s}$.
\end{defi}

We now define the synchronous multiparty session bisimilarity as
the greatest fixed point on the weak labelled transition relation
for the pairs of co-directed processes.

\begin{defi}[Synchronous multiparty session bisimulation]\rm
\label{def:swb}
	A typed relation $\mathcal{R}$ over closed processes 
	is a (weak) {\em synchronous multiparty session bisimulation} 
	or often a {\em synchronous bisimulation} if,  whenever 
	$\Gtprocess{P_1}{\De_1}\ \mathcal{R}\
        \noGtprocess{P_2}{\De_2}$, it holds:
	\begin{enumerate}
		\item	$\Gtprocess{P_1}{\De_1} \trans{\ell} \tprocess{\Ga'}{P_1'}{\De_1'}$ implies 
			$\Gtprocess{P_2}{\De_2} \Trans{\hat{\ell}} \tprocess{\Ga'}{P_2'}{\De_2'}$ such that 
			$\tprocess{\Ga'}{P_1'}{\De_1'}\ \mathcal{R}\ \noGtprocess{P_2'}{\De_2'}$. 
		\item The symmetric case.
	\end{enumerate}
	The maximum bisimulation exists which we call {\em synchronous 
	bisimilarity}, denoted by $\wb^s$. We sometimes leave
	environments implicit, writing e.g.~$P\wb^s Q$. 

\end{defi}
\iflong\else
By using the following lemma, we derive 
Theorem~\ref{the:scoincidence}.
See Appendix \ref{app:Bisimulation_typing_relation}.
\fi
\begin{lem}
	\label{lemma:Bisimulation_typing_relation}
	If $\Gtprocess{P_1}{\De_1} \wb^s \noGtprocess{P_2}{\De_2}$ then
	$\De_1 \bistyp \De_2$.
\end{lem}

\begin{proof}
	The proof uses the co-induction method and can be found
	in Appendix \ref{app:Bisimulation_typing_relation}.
\end{proof}

%
%
\begin{thm}[Soundness and completeness]
	\label{the:scoincidence}
	$\cong^{s} \ = \ \wb^s$. 
\end{thm}

\begin{proof}
	The proof is a simplification of the proof of 
	Theorem \ref{the:gcoincidence} in 
	Appendix \ref{app:bis-complete}. 
\end{proof}

\begin{exa}[Synchronous multiparty bisimulation] 
	\label{ex:swb}

	We use the running example from the introduction,
	\S~\ref{sec:intro},
	for a demonstration of the bisimulation semantics 
	developed in this section.
	In the introduction we considered transition under
	the untyped setting \cite{SangiorgiD:picatomp}. 
        If we follow the typed labelled
	transition system developed in this section we obtain
	similar interaction patterns. 

	Recall the definition of processes $P_1$, $P_2$, $P_3$ and $R_2$ from
	\S~\ref{sec:intro}. The linear types for these processes
	are empty since they have no free session names.
	\begin{eqnarray*}
		\Ga \proves P_1 \hastype \es, \quad 
		\Ga \proves P_2 \hastype \es, \quad 
		\Ga \proves P_3 \hastype \es, \quad \text{and}\quad 
		\Ga \proves R_2 \hastype \es
	\end{eqnarray*}
where $\Gamma=a:G_a\cat b:G_b$ with 
\begin{eqnarray*}
	G_a &=& \valuegt{1}{3}{U}\valuegt{2}{3}{U} \inactgt \\
	G_b &=& \valuegt{1}{2}{U} \inactgt
\end{eqnarray*}
	We follow the untyped LTS from Figure~\ref{fig-synch-lts} to obtain
the following processes 
	by $\ltsrule{Acc}$ and $\ltsrule{Req}$:
	\begin{eqnarray*}
		P_1 &\trans{\actacc{a}{\set{1}}{s_a}}& P_1' = \acc{b}{1}{y} \out{s_a}{1}{3}{v} \sout{y}{2}{w} \inact\\
		P_2 &\trans{\actacc{a}{\set{2}}{s_a}}& P_2' = 
			\req{b}{2}{y} (\sinp{y}{1}{z}  \inact \Par \out{s_a}{2}{3}{v} \inact) \\
		P_3 &\trans{\actreq{a}{\set{3}}{s_a}}& P_3' = \inp{s_a}{3}{1}{z} \inp{s_a}{3}{2}{y} \inact\\
		R_2 &\trans{\actacc{a}{\set{2}}{s_a}}& R_2' = 
			\req{b}{2}{y} (\sinp{y}{1}{z} \out{s_a}{2}{3}{v} \inact)
	\end{eqnarray*}
	\noi The corresponding environment transitions are defined as:
	\begin{eqnarray*}
		(\Ga, \es) &\trans{\actacc{a}{\set{1}}{s_a}}& (\Ga, \ \typedrole{s_a}{1} \tout{3}{U} \tinact)\\ 
		(\Ga, \es) &\trans{\actacc{a}{\set{2}}{s_a}}& (\Ga, \ \typedrole{s_a}{2} \tout{3}{U} \tinact)\\
		(\Ga, \es) &\trans{\actacc{a}{\set{3}}{s_a}}& (\Ga, \ \typedrole{s_a}{3} \tinp{1}{U} \tinp{2}{U} \tinact)\\
	\end{eqnarray*}
	\noi 
We can now observe the typed transitions as follows:
	\begin{eqnarray*}
		\tprocess{\Ga}{P_1}{\es} &\trans{\actacc{a}{\set{1}}{s_a}}& \tprocess{\Ga}{P_1'}{\typedrole{s_a}{1} \tout{3}{U} \tinact} \\
		\tprocess{\Ga}{P_2}{\es} &\trans{\actacc{a}{\set{2}}{s_a}}& \tprocess{\Ga}{P_2'}{\typedrole{s_a}{2} \tout{3}{U} \tinact}\\
		\tprocess{\Ga}{P_3}{\es} &\trans{\actacc{a}{\set{3}}{s_a}}& \tprocess{\Ga}{P_3'}{\typedrole{s_a}{3} \tinp{1}{U} \tinp{2}{U} \tinact}
	\end{eqnarray*}
	By $\ltsrule{AccPar}$, we have
	\[
		P_1 \Par P_2 \trans{\actacc{a}{\set{1, 2}}{s_a}} P_1' \Par P_2'
	\]
By $\ltsrule{ReqPar}$,	another process combination would invoke
	\[
		P_1 \Par P_3 \trans{\actreq{a}{\set{1, 3}}{s_a}} P_1' \Par P_3'
	\]
	If we compose the missing process in either of both
	processes, 
	the role set $\set{1,2, 3}$ is now complete with respect to $3$, so that 
	by synchronisation $\ltsrule{TauS}$ we may observe:
	\[
		P_1 \Par P_2 \Par P_3 \trans{\tau} 
\newsp{s_a}{P_1' \Par P_2' \Par P_3'}
	\]
	Furthermore, we can also observe the corresponding typed transition,
	since $(\Ga, \De)$ can always perform a $\tau$ action:
	\[
		\Gtprocess{P_1 \Par P_2 \Par P_3}{\es} \trans{\tau}
		\Gtprocess{\newsp{s_a}{P_1' \Par P_2' \Par P_3'}}{\es}
	\]

        Now we demonstrate the intuition given in \S~\ref{sec:intro}, i.e.~the bisimulation developed in this section cannot equate  
$P_1 \Par P_2$ and $P_1 \Par R_2$. 
We have the following transitions:
	\[
		\begin{array}{l}
			\Gtprocess{P_1' \Par P_2'}{\De_0} \trans{\tau}
			\\
			\hspace{1cm}
			\Gtprocess{\newsp{s_b}{\out{s_a}{1}{3}{v} \out{s_b}{1}{2}{w} \inact \Par
			\inp{s_b}{2}{1}{z} \inact \Par \out{s_a}{2}{3}{v} \inact} = Q_1}{\De_0} \\
			\Gtprocess{P_1' \Par R_2'}{\De_0} \trans{\tau}
			\\
			\hspace{1cm}
			\Gtprocess{\newsp{s_b}{\out{s_a}{1}{3}{v} 
                        \out{s_b}{1}{2}{w} \inact \Par
			\inp{s_b}{2}{1}{z} \out {s_a}{2}{3}{v} \inact} = Q_2}{\De_0}
		\end{array}
	\]
	\noi with $\De_0 = \typedrole{s_a}{1} \tout{3}{U} \tinact \cat \typedrole{s_a}{2} \tout{3}{U} \tinact$.

	From this point, we can check:
	\[
		\Gtprocess{Q_1}{\De_0} \not\wb^s \Gtprocess{Q_2}{\De_0} 
	\]
	\noi due to the fact that $(\Ga, \De_0) \trans{\actout{s_a}{2}{3}{v}}$ and
	\begin{eqnarray*}
		\Gtprocess{Q_1}{\De_0} \trans{\actout{s_a}{2}{3}{v}} & &\\
		\Gtprocess{Q_2}{\De_0} \stackrel{\actout{s_a}{2}{3}{v}}{\not\longrightarrow}
	\end{eqnarray*}

	The next result distinguishes the semantics of the typed
	equivalence semantics developed in this section from the untyped equivalence semantics \cite{SangiorgiD:picatomp}.
	\[
		\Gtprocess{Q_1\Par P_3'}{\Delta_0 \cat \typedrole{s_a}{3} \tinp{1}{U}
		\tinp{2}{U} \tinact} \wb^s \noGtprocess{\inact}{
		\typedrole{s_a}{a}\tinact \cat \typedrole{s_a}{2}\tinact \cat \typedrole{s_a}{3}\tinact}
	\]
	\noi since
	\[
		(\Gamma,\Delta) \stackrel{\ell}{\not\longrightarrow}
	\]
	for any $\ell\not=\tau$
	with $\Delta=\Delta_0 \cat \typedrole{s_a}{3} \tinp{1}{U} \tinp{2}{U} \tinact$ 
	(by the condition of \erule{Send} in Figure~\ref{fig-synch_envtrans}). 

	However
	the untyped labelled transition semantics do not equate
	the two processes $Q_1\Par P_3' \not\wb \inact$
	since $Q_1\Par P_3'\trans{\actout{\s_a}{1}{3}{\va}}$. 
\end{exa}


\section{Globally Governed Behavioural Theory}
\label{sec:governed}

We introduce the semantics for globally governed behavioural theory.
In the previous section, the local typing ($\De$) is used to 
constrain the untyped LTS and give rise to a local typed LTS.
In a multiparty distributed environment, 
communication follows the global protocol,
which controls both an observed process and its observer. 
The local typing is not sufficient to maintain the consistency 
of transitions of a
process with respect to a global protocol. In this section we refine
the environment LTS with a {\em global environment} $E$ to give a more
fine-grained control over the LTS of the processes.
We then show a bisimulation-based reasoning technique
which equates 
the two processes $P_1 \ | \ P_2$ and $P_1 \ | \ R_2$ in \S~\ref{sec:intro} 
by the governed bisimulation, which cannot be equated by the standard 
synchronous typed bisimulation $\wb^s$ studied in the previous section. 

%
%
\subsection{Global environments and configurations}
\label{subsec:globalenv}

We define a {\em global environment} ($E,E',...$) as
a mapping from session names to global types.
\[
	E \bnfis E \cat \s:\G \bnfbar \es
\]
We extend the projection definition on global environments $E$ as follows:
\[
	\projset{E} = \bigcup_{s:G \in E} \projset{s:G}
\]
where $\projset{s:G}$ associates the projection of type $G$ with session name $s$
as follows: $\projset{\typed{s} \G} = \set{\typedrole{\s}{\p} \proj{\G}{\p}\setbar \p \in \roles{\G}}$. 
Note that $E$ is a mapping from a session channel to a global type,
while $\Gamma$ is a mapping from a shared channel to a global type. 

We define a labelled reduction relation over
global environments
which corresponds to $\Delta_0 \red \Delta'_0$
defined in \S~\ref{subsec:typesoundness}.
We use the labels:
\[
	\lambda \bnfis \actval{s}{\p}{\q}{U} \bnfbar \actgsel{s}{\p}{\q}{l}
\]
to annotate reductions over global environments.
We define $\outp{\lambda}$ and $\inpp{\lambda}$ as:
\[
\begin{array}{rcrcl}
	\outp{\actval{s}{\p}{\q}{U}} &=& \outp{\actgsel{s}{\p}{\q}{l}} &=& \p\\
	\inpp{\actval{s}{\p}{\q}{U}} &=& \inpp{\actgsel{s}{\p}{\q}{l}} &=& \q\\
\end{array}
\]
and write 
$\p \in \ell$ if $\p \in \set{\outp{\ell}} \cup \set{\inpp{\ell}}$. 
 
\begin{defi}[Global environment reduction]
	\label{def:gltsrules}
	We define the relation $E \trans{\lambda} E'$ as the smallest 
	relation generated by the following rules:
	\[
	\begin{array}{c}
		\set{s:\valuegt{\p}{\q}{U} G} \trans{\actval{s}{\p}{\q}{U}} \set{s:G} \quad \mrule{Inter}
		\qquad
		\set{s:\selgt{\p}{\q}{l_i:G_i}_{i \in I}} \trans{\actgsel{s}{\p}{\q}{l_k}} \set{s:G_k} \quad \mrule{SelBra}
		\\[3mm]

		\tree{
			\set{s:G} \trans{\lambda} \set{s:G'} \quad \p, \q \notin \lambda
		}{
			\set{s: \valuegt{\p}{\q}{U} G} \trans{\lambda} \set{s: \valuegt{\p}{\q}{U} G'}
		} \ \mrule{IPerm}
		\\[7mm]
		\tree{
			\forall i \in I \quad \set{s:G_i} \trans{\lambda} \set{s:G_i'}
        		\quad \p, \q \notin \lambda
		}{
			\set{s: \selgt{\p}{\q}{l_i: G_i}_{i\in I}} \trans{\lambda}
		        \set{s: \selgt{\p}{\q}{l_i:G_i'}_{i\in I}}
		}\ \mrule{SBPerm}
		\qquad 
		\tree{
			E \trans{\lambda} E'
		}{
			E \cat E_0 \trans{\lambda} E' \cat E_0
		}\ \mrule{GEnv}
	\end{array}
	\]
\end{defi}
\noindent We often omit the label $\lambda$ by writing
$\red$ for $\trans{\lambda}$ and $\red^*$ for $(\trans{\lambda})^*$.
Rule $\mrule{Inter}$ is the axiom for the input and output interaction between two
parties; rule $\mrule{SelBra}$ reduces on the select/branch choice;
Rules $\mrule{IPerm}$ and $\mrule{SBPerm}$ define the case where we can permute 
action $\lambda$ to be performed under $\p \to \q$ if $\p$ and
$\q$ are not related to the participants in $\lambda$. 
Note that in our synchronous
semantics, we can permute two actions with no relevance in the participating
roles without changing the interaction semantics of the entire global
protocol.
Finaly rule $\mrule{GEnv}$ is a congruence rule over global
environments.

As a simple example of the above LTS, consider the global type:
\[
	s: \valuegt{\p}{\q}{U_1} \selgt{\p'}{\q'}{l_1: \inactgt, l_2: \valuegt{\p'}{\q'}{U_2} \inactgt}
\]
Since $\p, \q, \p', \q'$ are pairwise distinct,
we can apply the second and third rules to obtain:
\[
	s: \valuegt{\p}{\q}{U_1} \selgt{\p'}{\q'}{l_1: \inactgt, l_2: \valuegt{\p'}{\q'}{U_2} \inactgt}
	\trans{\actgsel{\s}{\p'}{\q'}{l_1}}
	s: \valuegt{\p}{\q}{U_1} \inactgt
\]
%


\noindent Next we introduce the {\em governance judgement} which controls the
behaviour of processes by the global environment.
\begin{defi}[Governance judgement]
	Let $\Gtprocess{\PP}{\De}$ be coherent.
	We write $\tprocess{E, \Ga}{\PP}{\De}$ if
	$\exists E'$ such that $E \red^* E'$ and $\De \subseteq \projset{E'}$.
\end{defi}
The global environment $E$ records the knowledge of both the
environment ($\De$) of the observed process $P$ and the environment
of its {\em observer}.
The side conditions ensure that $E$ is coherent with $\Delta$: there
exist  $E'$ reduced from $E$ whose projection should cover the environment
of $P$ since $E$ should include the observer's information together with
the observed process information recorded into $\Delta$. The reason 
that $E$ is allowed to have a few reduction steps behind the local 
environment $\Delta$ is that the observer has more informative 
global knowledge (in the form of a global type) 
before the moment the session was actually reduced to $\Delta$ which 
coincides with the projection of $E'$. 

Next we define the LTS for well-formed environment configurations.
\begin{defi}[Environment configuration]
\label{def:env_conf}
	We write $(E, \Ga, \De)$ if
	$\exists E'$ such that $E \red^* E'$ and $\De \subseteq
        \projset{E'}$.
\end{defi}
The up-to reduction requirement on $E$ allows a global
environment $E$ to configure linear environments $\De$ that
also differ up-to reduction. Specificaly a global
environment $E$ configures pairs of linear environments
that type equivalent processes.

We refined the reduction
relation on $\Delta$ in \S~\ref{subsec:typesoundness} as a labelled
reduction relation on $\Delta$, which is used for defining
a labelled transition system over environment configurations:
\begin{defi}[Linear typing labelled reduction]
$ $
\begin{enumerate}
	\item	$\set{\typedrole{s}{\p} \tout{\q}{U} T \cat \typedrole{s}{\q} \tinp{\p}{U} T'}
		\trans{\actval{\s}{\p}{\q}{U}}
		\set{\typedrole{s}{\p} T \cat \typedrole{s}{\q} T'}$.

	\item	$\set{\typedrole{s}{\p} \tsel{\q}{l_i:T_i}_{i \in I} \cat
		\typedrole{s}{\q} \tbra{\p}{l_j:T'_j}_{j \in J}}
		\trans{\actgsel{s}{\p}{\q}{l_k}}
		\set{\typedrole{s}{\p} T_k \cat \typedrole{s}{\q} T'_k}\ I \subseteq J, k \in I$.

	\item	$\De \cup \De' \trans{\lambda} \De \cup \De''$ if $\De' \trans{\lambda} \De''$.
\end{enumerate}
\end{defi}
Figure~\ref{fig:ELTS} defines an LTS over environment configurations that refines
the LTS over environments
(i.e $(\Ga, \De) \trans{\ell} (\Ga', \De')$) in \S~\ref{sec:tupedLTS}.

\begin{figure}[t]
\[
\small
	\begin{array}{c}
		\tree{
			s \notin \dom{E} \quad
			(\Ga, \De_1) \trans{\actacc{a}{A}{s}} (\Ga, \De_2)
		}{
			(E, \Ga, \De_1) \trans{\actacc{a}{A}{s}} (E \cat s:G, \Ga, \De_2)
		}
		\quad \eltsrule{Acc}
		\qquad \quad
		\tree{
		s \notin \dom{E} \quad
		(\Ga, \De_1) \trans{\actreq{a}{A}{s}} (\Ga, \De_2)
		}{
		(E, \Ga, \De_1) \trans{\actreq{a}{A}{s}} ( E \cat s:G, \Ga, \De_2)
		}
		\quad \eltsrule{Req}
		\\[9mm]

		\tree{
		(\Ga, \De_1) \trans{\actout{s}{\p}{\q}{v}} (\Ga, \De_2) \quad 
		E_1 \trans{\actval{s}{\p}{\q}{U}} E_2
		}{
		(E_1, \Ga, \De_1) \trans{\actout{s}{\p}{\q}{v}} (E_2, \Ga, \De_2)
		}  
		\quad \eltsrule{Out}
		\\[9mm]

		\tree{
			(\Ga, \De_1) \trans{\actinp{s}{\p}{\q}{v}} (\Gacat v:U, \De_2) \quad 
			E_1 \trans{\actval{s}{\q}{\p}{U}} E_2
		}{
			(E_1, \Ga, \De_1) \trans{\actinp{s}{\p}{\q}{v}} (E_2, \Gacat v:U, \De_2)
		}
		\quad \eltsrule{In}
		\\[9mm]

		\tree{
			(\Ga, \De_1) \trans{\actbout{s}{\p}{\q}{a}} (\Ga \cat a:\chtype{G}, \De_2) \quad
			E_1 \trans{\actval{s}{\q}{\p}{\chtype{G}}} E_2
		}{
			(E_1, \Ga, \De_1) \trans{\actbout{s}{\p}{\q}{a}} (E_2, \Ga \cat a:\chtype{G}, \De_2)
		} 
		\quad \eltsrule{ResN}
		\\[9mm]

		\tree{
			(\Ga, \De_1) 
\trans{\actbdel{\s}{\p}{\q}{\s'}{\p'}} (\Ga, \De_2 \cat \set{s'[\p_i]:
  T_i}_{i\in I}) \quad 
\forall i\in I. \proj{G}{\p_i} = T_i
\quad s'\notin \dom{E_1} \quad
			E_1 \trans{\actval{s}{\q}{\p}{T_{\p'}}} E_2
		}{
			(E_1, \Ga, \De_1) \trans{\actbdel{\s}{\p}{\q}{\s'}{\p'}} (E_2 \cat s':G, \Ga, \De_2 \cat \set{s'[\p_i]: T_i}_{i\in I})
		}
		\quad \eltsrule{ResS}
		\\[9mm]

		\tree{
			(\Ga, \De_1) \trans{\actsel{s}{\p}{\q}{l}} (\Ga, \De_2) \quad E_1 \trans{\actval{s}{\p}{\q}{l}} E_2
		}{
			(E_1, \Ga, \De_1) \trans{\actsel{s}{\p}{\q}{l}} (E_2, \Ga, \De_2)
		} 
		\quad \eltsrule{Sel}
		\qquad \quad
		\tree{
			(\Ga, \De_1) \trans{\actbra{s}{\p}{\q}{l}} (\Ga, \De_2) \quad 
			E_1 \trans{\actgsel{s}{\q}{\p}{l}} E_2
		}{
			(E_1, \Ga, \De_1) \trans{\actbra{s}{\p}{\q}{l}} (E_2, \Ga, \De_2)
		}
		\quad \eltsrule{Bra}
		\\[9mm]

		\tree{
			(\De_1 =\De_2, E_1 = E_2)  \vee
			(\De_1 \trans{} \De_2, E_1 \trans{\lambda} E_2)
		}{
			(E_1, \Ga, \De_1) \trans{\tau} (E_2, \Ga, \De_2)
		}
		\quad \eltsrule{Tau}  
		\\[9mm]

		\tree{
			E_1 \red^\ast E_1' \quad
			(E_1', \Ga_1, \De_1) \trans{\ell} (E_2, \Ga_2, \De_2)
		}{
			(E_1, \Ga_1, \De_1) \trans{\ell} (E_2, \Ga_2, \De_2)
		}
		\quad \eltsrule{Inv}
	\end{array}
\]
\caption{Labelled transition system for environment configuations}
\label{fig:ELTS}
\end{figure}

Each rule requires
a corresponding environment transition
(Figure~\ref{fig-synch_envtrans} in \S~\ref{sec:tupedLTS}) and a corresponding
labelled global environment transition
in order to
control a transition following the global protocol.
Rule $\eltsrule{Acc}$ defines the acceptance of a session initialisation
by creating a new mapping $s:G$ which
matches $\Gamma$ in a governed environment $E$.
Rule $\eltsrule{Req}$ defines the request for a new session and
it is dual to $\eltsrule{Acc}$.

The next six rules are the transition relations on session
channels and we assume the condition
$\projset{E_1} \supseteq \Delta_1$ to ensure
the base action of the environment matches one in a global environment.
$\eltsrule{Out}$ is a rule for the output where the type of the
value and the action of $(\Gamma,\Delta)$ meets those in
$E$. $\eltsrule{In}$ is a rule for the input and dual to
$\eltsrule{Out}$. $\eltsrule{ResN}$ is a scope opening rule for
a name so that the environment
can perform the corresponding type $\ENCan{G}$ of $a$.
$\eltsrule{ResS}$ is a scope opening rule for
a session channel which creates a set of mappings for the opened
session channel $s'$ corresponding to the LTS of the environment.
$\eltsrule{Sel}$ and $\eltsrule{Bra}$ are the rules for
select and branch, which are similar to $\eltsrule{Out}$ and
$\eltsrule{In}$. Rule $\eltsrule{Tau}$ defines the silent
action for environment configurations, where we require that
reduction on global environments matches reduction on 
the linear typing. At the same time we allow a silent action
with no effect on the environment configuration.
Rule $\eltsrule{Inv}$ closes the labelled transition system
with respect to the global environment. Global environment 
$E_1$ reduces to $E_1'$ to perform the observer's actions, hence
the observed process can perform the action w.r.t. $E_1'$.

Hereafter we write $\red$ for $\trans{\tau}$. 
\begin{exa}[LTS for environment configuration]
\label{ex:envconf}
$ $

	\noi Let:
	\begin{eqnarray*}
		E &=& s: \valuegt{\p}{\q}{U} \valuegt{\p}{\q}{U} G\\
		\Ga &=& \va:U\\
		\De &=& \typedrole{s}{\p} \tout{\q}{U} T_\p
	\end{eqnarray*}
	\noi with $\proj{G}{\p} = T_\p$, $\proj{G}{\q} = T_\q$ and $\roles{G} = \set{\p, \q}$.

	Tuple $(E, \Ga, \De)$ is an environment configuration since
	there exists $E'$ such that:
	\[
		E \red E' \textrm{ implies } \projset{E'} \supset \De
	\]
	\noi Recall that we can write  $E \red E'$ for $E \trans{\lambda} E'$.
	\noi Indeed we can see that:
	\begin{eqnarray*}
		E &\trans{\actval{s}{\p}{\q}{U}}& s: \valuegt{\p}{\q}{U} G\\
		\projset{s: \valuegt{\p}{\q}{U} G} &=&
		\typedrole{s}{\p} \tout{\q}{U} T_\p \cat \typedrole{s}{\q} \tinp{\p}{U} T_\q \\
		\projset{s: \valuegt{\p}{\q}{U} G} &\supset& \De
	\end{eqnarray*}
	\noi An environment configuration transition takes a place 
	on environment configuration $(E, \Ga, \De)$
	if we apply the condition of rule $\eltsrule{Out}$ to obtain:
	\begin{eqnarray*}
		s: \valuegt{\p}{\q}{U} G &\trans{\actval{s}{\p}{\q}{U}}& s:G\\
		(\Ga, \typedrole{s}{\p} \tout{\q}{U} T_\p) &\trans{\actout{s}{\p}{\q}{\va}}&
		(\Ga, \typedrole{s}{\p} T_\p)
	\end{eqnarray*}
	thus we can obtain:
	\[
		(s: \valuegt{\p}{\q}{U} G, \Ga, \De)
		\trans{\actout{s}{\p}{\q}{\va}}
		(s:G, \Ga, \typedrole{s}{\p} T_\p)
	\]
	By last result and the fact that:
	\[
		E \red s: \valuegt{\p}{\q}{U} G
	\]
	\noi we use rule $\eltsrule{Inv}$,
	to obtain:
	\[
		(E, \Ga, \De) \trans{\actout{s}{\p}{\q}{\va}} (s:G, \Ga,
		\typedrole{s}{\p} T_\p)
	\]
	\noi as required.
\end{exa}

\paragraph{\bf Governed reduction-closed congruence.}
\label{subsec:grc}
To define the reduction-closed congruence, we first refine
the barb,
which is controlled
by the global witness where observables of a configuration
are defined with the global environment of the observer.
%
%
%
\begin{defi}[Governed barb]
\[
	\begin{array}{c}
		\tree {
			\srole{\s}{\q} \notin \dom{\De} \quad \exists E' \textrm{ such that } E \red^\ast E' \trans{\actval{\s}{\p}{\q}{U}}
			\quad
			\Decat \typedrole{\s}{\p} \tout{\q}{U} T \subseteq \projset{E'}
		}{
			(E, \Ga, \Decat \typedrole{\s}{\p} \tout{\q}{U} T) \barbout{\s}{\p}{\q}
		}
		\\[6mm]

		\tree {
			\srole{\s}{\q} \notin \dom{\De} \quad \exists E' \textrm{ such that } E \red^\ast E' \trans{\actgsel{\s}{\p}{\q}{l_k}},
			\quad \Decat \typedrole{\s}{\p} \tsel{\q}{l_i:T_i}_{i \in I} \subseteq \projset{E'} \quad k \in I 
		}{
			(E, \Ga, \Decat \typedrole{\s}{\p} \tsel{\q}{l_i:T_i}_{i \in I}) \barbout{\s}{\p}{\q}
		}
		\\[6mm]

		\tree {
			a\in \dom{\Ga}
		}{
			(E, \Ga, \De) \barbreq{a}
		}
	\end{array}
\]
\end{defi}
We write $(E, \Ga, \De) \Barb{m}$ if $(E, \Ga, \De) \red^\ast (\Ga, \De', E')$ and
$(\Ga, \De', E') \barb{m}$.

We define the binary operator $\sqcup$ over
global environments based on the inclusion
of the syntax tree for global types. The operation is used to define 
the typed relation with respect to 
a global witness and the governed bisimulation. 
The operator $\sqcup$ is used to relate two different, but 
compatible observers, $E_1$ and $E_2$. 
\begin{defi}
	Let $T_1$ and $T_2$ denote local types as defined in \S~\ref{sec:typing}. We write $T_1\sqsubseteq T_2$ 
	if the syntax tree of $T_2$ includes one of $T_1$ as a leaf.
	We extend to  $G_1 \sqsubseteq G_2$ by defining
	$\forall \srole{s}{\p}:T_1 \in \projset{s:G_1}, \exists 
         \srole{s}{\p}:T_2 \in \projset{s:G_2}$ and $T_1 \sqsubseteq T_2$.
	We define:
	$E_1\sqcup E_2 = \set{s:E_i(s) \setbar  E_j(s) \sqsubseteq E_i(s), i, j \in \set{1,2}, i\not=j}\cup
	\set{s:E_1(s), s':E_2(s') \setbar s \notin \dom{E_2}, s' \notin \dom{E_1}}$.
\end{defi}
As an example for global types inclusion consider that: 
	\[
		\tinp{\q}{\U'} T \sqsubseteq \tout{\p}{\U} \tinp{\q}{\U'} T
	\]
As an example of $E_1 \sqcup E_2$, let us define:
\begin{eqnarray*}
	E_1 &=& s_1: \valuegt{\p}{\q}{U_1} \valuegt{\p'}{\q'}{U_2} \valuegt{\p}{\q}{U_3} \inactgt \cat s_2: \valuegt{\p}{\q}{W_2} \inactgt\\
	E_2 &=& s_1: \valuegt{\p}{\q}{U_3} \inactgt \cat s_2: \valuegt{\p'}{\q'}{W_1} 
	\valuegt{\p}{\q}{W_2} \inactgt
\end{eqnarray*}
Then
\[
	E_1 \sqcup E_2 = s_1: \valuegt{\p}{\q}{U_1} \valuegt{\p'}{\q'}{U_2} \valuegt{\p}{\q}{U_3} \inactgt \cat
	s_2: \valuegt{\p'}{\q'}{W_1} \valuegt{\p}{\q}{W_2} \inactgt
\]
The behavioural relation w.r.t. a global
witness is defined below.
\begin{defi}[Configuration relation]
\label{def:configuration}
	The relation $\relfont{R}$ is a {\em configuration relation}
	between two configurations $\tprocess{E_1,\Ga}{\PP_1}{\De_1}$
	and $\tprocess{E_2,\Ga}{\PP_2}{\De_2}$, written
	\[
		\rel{\tprocess{E_1\sqcup E_2, \Ga}{\PP}{\De_1}}{R}{\noGtprocess{\PP_2}{\De_2}}
	\]
	if $E_1\sqcup E_2$ is defined.
\end{defi}
\begin{prop}[Decidability]
\label{pro:decidability}
$ $
	\begin{enumerate}
		\item	Given $E_1$ and $E_2$, a problem whether $E_1\sqcup E_2$ is defined or
			not is decidable and if it is defined, the calculation of $E_1\sqcup
			E_2$ terminates.
		\item	Given $E$, a set $\{ E' \ | \ E\trans{}^\ast E'\}$ is finite.
	\end{enumerate}
\end{prop}

\begin{proof}
	\label{app:pro:decidability}
	{\em (1)} since $T_1 \sqsubseteq T_2$ is a syntactic tree inclusion, it
	is reducible to a problem to check the isomorphism between two types.
	This problem is decidable \cite{yoshida.vasconcelos:language-primitives}.

	{\em (2)} the global LTS has one-to-one correspondence with the LTS of
	global automata in \cite{DY12} whose reachability set is finite.
\end{proof}





\begin{defi}[Global configuration transition]
\label{def:glts}
	We write
	$\tprocess{E_1, \Ga_1}{\PP_1}{\De_1} \trans{\ell} \tprocess{E_2, \Ga_2}{\PP_2}{\De_2}$ if $\tprocess{E_1, \Ga_1}{\PP_1}{\De_1}$, 
	$\tprocess{\Ga_1}{\PP_1}{\De_1} \trans{\ell} \tprocess{\Ga_2}{\PP_2}{\De_2}$ and $(E_1, \Ga_1, \De_1) \trans{\ell} (E_2, \Ga_2,\De_2)$.
\end{defi}
Note that $\tprocess{\Ga_2}{\PP_2}{\De_2}$ immediately holds by Definition 
\ref{def:lts}.



The proposition below states that the configuration LTS preserves the
well-formedness.

\begin{prop}[Invariants]
	\label{pro:invariants}
$ $
	\begin{enumerate}
		\item	$(E_1, \Ga, \De_1) \trans{\ell} (E_2, \Ga_2, \De_2)$ implies that
			$(E_2, \Ga_2, \De_2)$ is an environment configuration.
		\item	If $\Gtprocess{P}{\De}$ and $P \red P'$ with $\coherent{\De}$,
			then $\tprocess{E, \Ga}{P}{\De} \red \tprocess{E, \Ga}{P'}{\De'}$ and 
			$\coherent{\De'}$.
		\item	If $\tprocess{E_1, \Ga_1}{P_1}{\De_1} \trans{\ell} \tprocess{E_2, \Ga_2}{P_2}{\De_2}$ then
			$\tprocess{E_2, \Ga_2}{P_2}{\De_2}$ is a governance judgement.
	\end{enumerate}
\end{prop}
\begin{proof}
	The proof for Part 1 and Part 3 can be found in
\iflong
	Appendix \ref{appsubsec:invariant}.
\else
	Appendix \ref{appsubsec:invariant}.
\fi
	Part 2 is verified by simple transitions using
	$\eltsrule{Tau}$ in Figure~\ref{fig:ELTS}.
	$\coherent{\De'}$ is derived by Theorem~\ref{the:subject}.
\end{proof}

The definition of the reduction-closed congruence for governance follows.
Below we define $\tprocess{E, \Ga}{P}{\De} \Barb{m}$ if
$\PP \Barb{m}$ and $(E, \Ga, \De) \Barb{m}$.

\begin{defi}[Governed reduction-closed congruence]
\label{def:grc}
A configuration relation $\relfont{R}$ is {\em governed reduction-closed congruence}
if $\rel{\tprocess{E, \Ga}{P_1}{\De_1}}{R}{\noGtprocess{P_2}{\De_2}}$
then
	\begin{enumerate}
		\item	$\tprocess{E, \Ga}{P_1}{\De_1} \Barb{n}$
			if and only if
			$\tprocess{E, \Ga}{P_2}{\De_2} \Barb{n}$

		\item	\begin{itemize}
				\item	$P_1 \Red P_1'$ if there exists $P_2'$ such that $P_2 \Red P_2'$
					and $\rel{\tprocess{E, \Ga}{P_1'}{\De_1'}}{R}{\noGtprocess{P_2'}{\De_2'}}$.
				\item	the symmetric case.
			\end{itemize}

		\item	For all closed context $\C$,  such that
			$\tprocess{E, \Ga}{\Ccontext{P_1}}{\De_1'}$ and
			$\tprocess{E, \Ga}{\Ccontext{P_2}}{\De_2'}$
			then
			$\rel{\tprocess{E, \Ga}{\Ccontext{P_1}}{\De_1'}}{R}{\noGtprocess{\Ccontext{P_2}}{\De_2'}}$.		
\end{enumerate}
The union of all governed reduction-closed congruence relations is denoted as $\govcongs$.
\end{defi}


\subsection{\bf Globally governed bisimulation and its properties}
\label{subsec:gbsim}
This subsection introduces the
globally governed bisimulation relation definition
and studies its main properties.


\begin{defi}[Globally governed bisimulation]
\label{def:govwb}
A configuration relation $\relfont{R}$
is a {\em globally governed weak bisimulation} (or governed
bisimulation) if whenever
$\rel{\tprocess{E, \Ga}{\PP_1}{\De_1}}{R}{\noGtprocess{\PP_2}{\De_2}}$, it holds:
\begin{enumerate}
	\item	$\tprocess{E, \Ga}{\PP_1}{\De_1} \trans{\ell} \tprocess{E_1', \Ga'}{\PP_1'}{\De_1'}$ implies
		$\tprocess{E, \Ga}{\PP_2}{\De_2} \Trans{\hat{\ell}} \tprocess{E_2', \Ga'}{\PP_2'}{\De_2'}$
		such that $\rel{\tprocess{E_1'\sqcup E_2',\Ga'}{\PP_1'}{\De_1'}}{R}{\noGtprocess{\PP_2'}{\De'_2}}$.
		\item The symmetric case.
	\end{enumerate}
The maximum bisimulation exists which we call {\em governed
bisimilarity}, denoted by $\govwbs$. We sometimes leave
environments implicit, writing e.g.~$P\govwbs Q$.
\end{defi}




\iflong
\begin{lem}
$ $
\label{lem:bis-cong}
\label{lem:bis-complete}
	\begin{enumerate}
		\item $\govwbs$ is congruent.
		\item $\govcongs \ \subseteq \ \govwbs$.
	\end{enumerate}
\end{lem}

\begin{proof}
	The proof of (1) is by a case analysis on the
	context structure. The interesting
	case is the parallel composition,
	which uses Proposition \ref{pro:invariants}.
	See Appendix \ref{app:bis-cong}.

\fi
	The proof uses the technique from \cite{Hennessy07} 
(the external actions can be always tested).
	The proof can be found in Appendix \ref{app:bis-complete}.
\end{proof}

\begin{thm}[Soundness and completeness]
	\label{the:gcoincidence}
	$\govwbs \ = \ \govcongs$.
\end{thm}

\begin{proof}
	The fact that $\govwb \subseteq \govcong$ comes directly from
	the first part of Lemma~\ref{lem:bis-cong}.
	The proof is completed using the second part of Lemma~\ref{lem:bis-cong}.
\end{proof}


\noi The next theorem clarifies the relation between the locally controlled 
bisimilarity $\wb^s$ and globally governed bisimilarity $\wb_g^s$.

\begin{thm}
\label{lem:full-abstraction}
	If for all $E$ such that $\tprocess{E, \Ga}{P_1}{\De_1} \wb_g^s \noGtprocess{P_2}{\De_2}$ then
	$\Gtprocess{P_1}{\De_1} \wb^s \noGtprocess{P_2}{\De_2}$.
	Also if $\Gtprocess{P_1}{\De_1} \wb^s \noGtprocess{P_2}{\De_2}$, then for all
	$E$, $\tprocess{E, \Ga}{P_1}{\De_1} \wb_g^s \noGtprocess{P_2}{\De_2}$.
\end{thm}

\begin{proof}
	The proof is based on the properties that exist between semantics of the
	environment tuples $(\Ga, \De)$ and 
	the semantics of the environment configurations $(E, \Ga, \De)$.
	The full proof can be found in Appendix~\ref{app:full-abstraction}.
\end{proof}

To clarify the above theorem, consider the following processes:
\[
\begin{array}{rcl}
	\small
	P_1 & = & \out{s_1}{1}{3}{v} \out{s_2}{1}{2}{w} \inact \Par \out{s_1}{2}{3}{v} \inp{s_2}{2}{1}{x} 
	\out{s_2}{2}{3}{x} \inact\\
	P_2 & = & \out{s_1}{1}{3}{v} \inact \Par \out{s_2}{1}{2}{w} \inact \Par \out{s_1}{2}{3}{v} \inp{s_2}{2}{1}{x} \out{s_2}{2}{3}{x} \inact
\end{array}
\]
\noi We can show that $P_1 \wb^s P_2$. 
By Theorem~\ref{lem:full-abstraction}, 
we expect that 
for all $E$, we have $\tprocess{E, \Ga}{P_1}{\De_1}$ and $\tprocess{E, \Ga}{P_2}{\De_2}$ then
$E \proves P_1 \wb_g^s P_2$.
This is in fact true because 
the possible $E$ that can type 
$P_1$ and $P_2$ are:
\[
\begin{array}{rcl}
\small
	E_1 &=& s_1: \valuegt{1}{3}{U} \valuegt{2}{3}{U} \inactgt \cat s_2: \valuegt{1}{2}{W} \valuegt{2}{3}{W} \inactgt\\
	E_2 &=& s_1: \valuegt{2}{3}{U} \valuegt{1}{3}{U} \inactgt \cat s_2: \valuegt{1}{2}{W} \valuegt{2}{3}{W} \inactgt
\end{array}
\]
\noi and all the up-to weakening instances $E$
(see Lemma~\ref{lemma:weakening})
of $E_1$ and $E_2$.

To clarify the difference between $\wb^s$ and $\wb_g^s$, we introduce
the notion of
a {\em simple multiparty  process} defined in \cite{HYC08}.
A simple process contains only a single session so that it satisfies
the progress property as proved in \cite{HYC08}.
Formally a process $P$ is {\em simple} when it is typable
with a type derivation where the session typing in the premise and
the conclusion of each prefix rule is restricted to at most
a single session (i.e.~any
$\tprocess{\Ga}{P}{\De}$ which appears in a
derivation, $\Delta$ contains at most one session channel in its
domain, see \cite{HYC08}).
Since there is no interleaving of sessions in simple
processes, the difference between $\wb^s$ and $\wb_g^s$ disappears.

\begin{thm}[Coincidence]
\label{thm:coincidence}
	Assume $P_1$ and $P_2$ are simple.
	If there exists $E$ such that $\tprocess{E, \Ga}{P_1}{\De_1} \wb_g^s
	\noGtprocess{P_2}{\De_2}$, then
	$\tprocess{\Ga}{P_1}{\De_1} \wb^s \noGtprocess{P_2}{\De_2}$.
\end{thm}

\begin{proof}
	The proof follows the fact that if
	$P$ is simple and
	$\tprocess{\Ga}{P}{\De} \trans{\ell} \noGtprocess{P'}{\De'}$
	then $\exists E$ such that
	$\tprocess{E, \Ga}{P}{\De} \trans{\ell} \noGtprocess{P'}{\De'}$
	to continue that
	if $P_1$ and $P_2$ are simple
	and there exists $E$ such that
	$\tprocess{E, \Ga}{P_1}{\De_1} \wb_g^s \noGtprocess{P_2}{\De_2}$ then
	$\forall E, \tprocess{E, \Ga}{P_1}{\De_1} \wb_g^s \noGtprocess{P_2}{\De_2}$.
	The result then comes by applying
	Theorem~\ref{lem:full-abstraction}.
	The details of the proof are in the
	Appendix~\ref{app:thm_coincidence}.
\end{proof}

To clarify the above theorem, consider:
\begin{eqnarray*}
	P_1  &=&  \inp{s}{1}{2}{x} \out{s}{1}{3}{x} \inact \Par
	\out{s}{2}{1}{v} \inact\\
	P_2  &=&  \out{s}{1}{3}{v} \inact
\end{eqnarray*}
It holds that for 
\[
	E = s: \valuegt{2}{1}{U} \valuegt{1}{3}{U} \inactgt
\]
We can easily reason that $E \proves P_1 \wb_g^s P_2$ hence $P_1 \wb^s P_2$.

\begin{exa}[Governed bisimulation]
\label{ex:govern}
Recall the example from \S~\ref{sec:intro} and Example \ref{ex:swb}.
$Q_1$ is the process corresponding to 
a sequential thread (this corresponds to $P_1 \Par P_2$ in \S~\ref{sec:intro}), while
$Q_2$ has a parallel thread instead of the sequential composition (this corresponds to $P_1 \Par R_2$ in \S~\ref{sec:intro}). 
\[
	\begin{array}{rcl}
		Q_1 & = &
\newsp{s_b}{\out{s_a}{1}{3}{v} \out{s_b}{1}{2}{w} \inact \Par
		\inp{s_b}{2}{1}{x} \inact \Par \out{s_a}{2}{3}{v} \inact}\\
		Q_2 & = &
\newsp{s_b}{\out{s_a}{1}{3}{v} \out{s_b}{1}{2}{w} \inact \Par
				\inp{s_b}{2}{1}{x} \out{s_a}{2}{3}{v} \inact}
	\end{array}
\]
Assume:
\[
\small
	\begin{array}{rcl}
		\Ga & = &	a: G_a \cat b: G_b\\
		\De_0 & = &	\typedrole{s_a}{1} \tout{3}{S} \tinact \cat
				\typedrole{s_a}{2} \tout{3}{S} \tinact
\end{array}
\]
Then we have $\Gtprocess{Q_1}{\De_0}$ and $\Gtprocess{Q_2}{\De_0}$.
Now assume the two global witnesses as:
\[
\small
	\begin{array}{rcl}
		E_1 & = &	s_a: \valuegt{1}{3}{S} \valuegt{2}{3}{S} \inactgt\\

		E_2 & = &	s_a: \valuegt{2}{3}{S} \valuegt{1}{3}{S} \inactgt 
	\end{array}
\]
Then the projection of $E_1$ and $E_2$ is given as:
\[
\small
\begin{array}{rcl}
\projset{E_1} & = &	\typedrole{s_a}{1} \tout{3}{S} \tinact \cat
			\typedrole{s_a}{2} \tout{3}{S} \tinact \cat
			\typedrole{s_a}{3} \tinp{1}{S} \tinp{2}{S} \tinact
\\

		\projset{E_2} & = &	\typedrole{s_a}{1} \tout{3}{S} \tinact \cat
					\typedrole{s_a}{2} \tout{3}{S} \tinact \cat
			\typedrole{s_a}{3} \tinp{2}{S} \tinp{1}{S} \tinact
	\end{array}
\]
	\noi	with $\De_0 \subset \projset{E_1}$ and $\De_0 \subset
        \projset{E_2}$. The reader should note that the difference
between $E_1$ and $E_2$ is the type of the participant $3$ at $s_a$
(the third mapping in $E_1$ and $E_2$).

By definition of the global environment configuration,
we can write:
\begin{center}
	$\tprocess{E_i, \Ga}{Q_1}{\De_0}$ and
	$\tprocess{E_i, \Ga}{Q_2}{\De_0}$ for $i=1,2$.
\end{center}
Both processes are well-formed global configurations under both witnesses.
Now we can observe
\[
\tprocess{\Ga}{Q_1}{\De_0} \trans{\actout{s_a}{2}{3}{v}}
\tprocess{\Ga}{Q_1'}{\De'_0}
\] 
but
\[	\tprocess{\Ga}{Q_2}{\De_0}
        \stackrel{\actout{s_a}{2}{3}{v}}{\not\longrightarrow}
\]
Hence
	$\Gtprocess{Q_1}{\De_0} \not\wb^s \noGtprocess{Q_2}{\De_0}$
as detailed in Example \ref{ex:swb}.

Similarly, we have:
\[
\tprocess{E_2, \Ga}{Q_1}{\De_0} \not\wb_g^s \noGtprocess{Q_2}{\De_0}
\]
because $E_2$ allows to output
action $\actout{s_a}{2}{3}{v}$ by $\eltsrule{Out}$ in 
Figure~\ref{fig:ELTS} (since $E_2 \trans{s_a:2 \to 3:S} E_2'$).

On the other hand, since $E_1$ {\em forces} to wait for
$\actout{s_a}{2}{3}{v}$,  
\[\tprocess{E_1,\Ga}{Q_1}{\De_0}
\stackrel{\actout{s_a}{2}{3}{v}}{\not\longrightarrow}\] 
because we cannot apply 
$\eltsrule{Out}$ in Figure~\ref{fig:ELTS}.   
$E_1$ does not allow to output
action $\actout{s_a}{2}{3}{v}$ 
 (since $E_1 \not\trans{s_a:2 \to 3:S}$).
Hence
	$Q_1$ and $Q_2$ are bisimilar under $E_1$, 
i.e.~$\tprocess{E_1,\Ga}{Q_1}{\De_0} \wb_g^s
        \noGtprocess{Q_2}{\De_0}$. This concludes
the optimisation illustrated in \S~\ref{sec:intro} is correct.
\end{exa}

\section{Usecase: UC.R2.13 ``Acquire Data
From Instrument'' from the 
Ocean Observatories Initiative (OOI)~\cite{ooi}}
\label{sec:app:ooi}

\begin{figure}[t]
\usetikzlibrary{positioning,calc}
\tikzset{
  role/.style={rounded corners,draw,text width=3.2em,minimum height=1em,text height=1em,text centered},
  action/.style={midway,yshift=5pt,text=black} 
}
\begin{center}
\begin{tikzpicture}

  \node (uc1-instrument) [role] {\scriptsize $\textsf{Instrument}$};
  \node (uc1-agent) [role,right=1em of uc1-instrument] {\scriptsize $\textsf{Agent}$};
  \node (uc1-user) [role,right=1em of uc1-agent] {\scriptsize $\textsf{User}$};

  \draw[dash pattern= on 8pt off 4pt]
    (uc1-instrument.south) -- ($(uc1-instrument.south)+(0,-3)$);
  \draw[dash pattern= on 8pt off 4pt]
    (uc1-agent.south) -- ($(uc1-agent.south)+(0,-3)$);
  \draw[dash pattern= on 8pt off 4pt]
    (uc1-user.south) -- ($(uc1-user.south)+(0,-3)$);

  \draw[->, very thick, violet]
    ($(uc1-instrument)+(0,-1)$) -- ($(uc1-agent)+(0,-1)$)
    node [action] {\scriptsize $\actout{s_2}{\instrument}{\agentone}{rd}$};
  \draw[->, very thick, blue]
    ($(uc1-agent)+(0,-1.5)$) -- ($(uc1-user)+(0,-1.5)$)
    node [action] {\scriptsize $\actout{s_1}{\agentone}{\user}{pd1}$};
  \draw[<-, very thick, violet]
    ($(uc1-instrument)+(0,-2)$) -- ($(uc1-agent)+(0,-2)$)
    node [action] {\scriptsize $\actout{s_2}{\agentone}{\instrument}{ack}$};
  \draw[->, very thick, blue]
    ($(uc1-instrument)+(0,-2.5)$) -- ($(uc1-user)+(0,-2.5)$)
    node [action] {\scriptsize $\actout{s_1}{\instrument}{\user}{pd2}$};
  \node [below=3 of uc1-agent] {\footnotesize $\textsf{Usecase 1}$};


  \node (uc2-instrument) [role,right=2em of uc1-user] {\scriptsize $\textsf{Instrument}$};
  \node (uc2-agent1) [role,right=1em of uc2-instrument] {\scriptsize $\textsf{Agent}_1$};
  \node (uc2-agent2) [role,right=1em of uc2-agent1] {\scriptsize $\textsf{Agent}_2$};
  \node (uc2-user)  [role,right=1em of uc2-agent2] {\scriptsize $\textsf{User}$};

  \draw[dash pattern= on 8pt off 4pt]
    (uc2-instrument.south) -- ($(uc2-instrument.south)+(0,-3.4)$);
  \draw[dash pattern= on 8pt off 4pt]
    (uc2-agent1.south) -- ($(uc2-agent1.south)+(0,-3.4)$);
  \draw[dash pattern= on 8pt off 4pt]
    (uc2-agent2.south) -- ($(uc2-agent2.south)+(0,-3.4)$);
  \draw[dash pattern= on 8pt off 4pt]
    (uc2-user.south) -- ($(uc2-user.south)+(0,-3.4)$);

  \draw[->, very thick, violet]
    ($(uc2-instrument)+(0,-1)$) -- ($(uc2-agent1)+(0,-1)$)
    node [action] {\scriptsize $\actout{s_2}{\instrument}{\agentone}{rd}$};
  \draw[->, very thick, blue]
    ($(uc2-agent1)+(0,-1.5)$) -- ($(uc2-user)+(0,-1.5)$)
    node [action] {\scriptsize $\actout{s_1}{\agentone}{\user}{pd1}$};
  \draw[<-, very thick, violet]
    ($(uc2-instrument)+(0,-2)$) -- ($(uc2-agent1)+(0,-2)$)
    node [action] {\scriptsize $\actout{s_2}{\agentone}{\instrument}{ack}$};
  \draw[->, very thick, violet]
    ($(uc2-instrument)+(0,-2.5)$) -- ($(uc2-agent2)+(0,-2.5)$)
    node [action] {\scriptsize $\actout{s_2}{\instrument}{\agenttwo}{rd}$};
  \draw[->, very thick, blue]
    ($(uc2-agent2)+(0,-3)$) -- ($(uc2-user)+(0,-3)$)
    node [action] {\scriptsize $\actout{s_1}{\agenttwo}{\user}{pd2}$};
  \draw[<-, very thick, violet]
    ($(uc2-instrument)+(0,-3.2)$) -- ($(uc2-agent2)+(0,-3.2)$)
    node [action] {\scriptsize $\actout{s_2}{\agenttwo}{\instrument}{ack}$};

  \path (uc2-agent1) -- (uc2-agent2) node [midway] (uc2-agent12) {};
  \node [below=3.5 of uc2-agent12] {\footnotesize $\textsf{Usecase 2}$};

\end{tikzpicture}
\end{center}

\vspace{3mm}

\begin{center}
\begin{tikzpicture}

  \node (uc3-instrument) [role] {\scriptsize $\textsf{Instrument}$};
  \node (uc3-agent1) [role,right=1em of uc3-instrument] {\scriptsize $\textsf{Agent}_1$};
  \node (uc3-agent2) [role,right=1em of uc3-agent1] {\scriptsize $\textsf{Agent}_2$};
  \node (uc3-user)  [role,right=1em of uc3-agent2] {\scriptsize $\textsf{User}$};

  \draw[dash pattern= on 8pt off 4pt]
    (uc3-instrument.south) -- ($(uc3-instrument.south)+(0,-3.4)$);
  \draw[dash pattern= on 8pt off 4pt]
    (uc3-agent1.south) -- ($(uc3-agent1.south)+(0,-3.4)$);
  \draw[dash pattern= on 8pt off 4pt]
    (uc3-agent2.south) -- ($(uc3-agent2.south)+(0,-3.4)$);
  \draw[dash pattern= on 8pt off 4pt]
    (uc3-user.south) -- ($(uc3-user.south)+(0,-3.4)$);

  \draw[->, very thick, violet]
    ($(uc3-instrument)+(0,-0.7)$) -- ($(uc3-agent1)+(0,-0.7)$)
    node [action] {\scriptsize $\actout{s_2}{\instrument}{\agentone}{rd}$};
  \draw[->, very thick, violet]
    ($(uc3-instrument)+(0,-1.2)$) -- ($(uc3-agent2)+(0,-1.2)$)
    node [action] {\scriptsize $\actout{s_2}{\instrument}{\agenttwo}{rd}$};
  \draw[->, very thick, blue]
    ($(uc3-agent1)+(0,-1.7)$) -- ($(uc3-user)+(0,-1.7)$)
    node [action] {\scriptsize $\actout{s_1}{\agentone}{\user}{pd1}$};
  \draw[->, very thick, blue]
    ($(uc3-agent2)+(0,-2.2)$) -- ($(uc3-user)+(0,-2.2)$)
    node [action] {\scriptsize $\actout{s_1}{\agenttwo}{\user}{pd2}$};
  \draw[<-, very thick, violet]
    ($(uc3-instrument)+(0,-2.7)$) -- ($(uc3-agent1)+(0,-2.7)$)
    node [action] {\scriptsize $\actout{s_2}{\agentone}{\instrument}{ack}$};	
  \draw[<-, very thick, violet]
    ($(uc3-instrument)+(0,-3.2)$) -- ($(uc3-agent2)+(0,-3.2)$)
    node [action] {\scriptsize $\actout{s_2}{\agenttwo}{\instrument}{ack}$};

  \path (uc3-agent1) -- (uc3-agent2) node [midway] (uc3-agent12) {};
  \node [below=3.5 of uc3-agent12] {\footnotesize $\textsf{Usecase 3}$};

\end{tikzpicture}
\end{center}
\caption{Three usecases from UC.R2.13 ``Acquire Data
From Instrument'' in \cite{ooi}}
\label{fig:ooi}
\end{figure}

The running example for the thread transformation in the previous sections 
is the minimum to demonstrate a difference between
$\wb^s_g$ and $\wb^s$. This
discipline can be applied to general situations where
multiple agents need to interact following a global specification.
Our governance bisimulation
can be useful in other large applications, for example, it can be applied to 
the optimisation and verification of distributed systems,  
and the correctness of service communication.
In this section, we present a reasoning example based on
the real world usecase, UC.R2.13 ``Acquire Data
From Instrument'', from the 
Ocean Observatories Initiative (OOI)~\cite{ooi},
and show the optimisation and
verification of network services.

In this usecase, we assume a user program (\CODE{U})
which is connected to the Integrated Observatory Network (ION).
The ION provides the interface between users and remote sensing
instruments. 
The user requests, via the ION agent services (\CODE{A}), 
the acquisition of processed data from an instrument 
(\CODE{I}). More specifically the user requests from the ION
two different formats of the instrument data.
In the above usecase we distinguish two points of
communication coordination: i) an internal ION multiparty
communication and ii) an external communication between
ION instruments and agents and the user. In other
words it is natural to require the initiation of two
multiparty session types to coordinate the services
and clients involved in the usecase.
The behaviour of the multiparty session connection 
between the User (\CODE{U}) and ION is dependent on
the implementation and the synchronisation
of the internal ION session.

Below we present three possible implementation scenarios
and compare their behaviour with respect to the user
program. Depending on the ION requirements we can
choose the best implementation with the correct behaviour.

\subsection{Usecase Scenario 1}

In the first scenario (depicted in Usecase 1 in Figure~\ref{fig:ooi}) 
the user program 
(\CODE{U}) wants to acquire the first format of
data from  the instrument (\CODE{I}) 
and at the same time acquire the second format of the 
data from an agent service (\CODE{A}).
The communication between the agent (\CODE{A}) 
and the instrument happens internally in the ION
on a separate private session. 

\begin{enumerate}
	\item	A new session connection $s_1$ is established between 
		(\CODE{U}), (\CODE{I}) and (\CODE{A}).
	\item	A new session connection $s_2$ is established 
		between (\CODE{A}) and (\CODE{I}).
	\item	(\CODE{I}) sends raw data through $s_2$ to (\CODE{A}).
	\item	(\CODE{A}) sends processed data (format 1) through $s_1$ to (\CODE{U}).
	\item	(\CODE{A}) sends the acknowledgement through $s_2$ to (\CODE{I}).
	\item	(\CODE{I}) sends processed data (format 2) through $s_1$ to (\CODE{U}).
\end{enumerate}

The above scenario is implemented as follows:
\begin{eqnarray*}
	I_0 \Par A \Par U
\end{eqnarray*}
\noi where
\begin{eqnarray*}
	I_0 &=&	\acc{a}{\instrumentzero}{s_1} \req{b}{\instrumentzero}{s_2} 
		\out{s_2}{\instrumentzero}{\agentone}{\mathtt{rd}} \inp{s_2}{\instrumentzero}{\agentone}{x} 
		\out{s_1}{\instrumentzero}{\user}{\mathtt{pd}} \inact\\
	A &=&	\acc{a}{\agentone}{s_1} \acc{b}{\agentone}{s_2} 
		\inp{s_2}{\agentone}{\instrumentzero}{x} \out{s_1}{\agentone}{\user}{\mathtt{pd}} 
		\out{s_2}{\agentone}{\instrumentzero}{\mathtt{ack}} \inact\\
	U &=&	\req{a}{\user}{s_1} \inp{s_1}{\user}{\agentone}{x} \inp{s_1}{\user}{\instrumentzero}{y} \inact
\end{eqnarray*}
\noi and $\instrumentzero$ is the instrument role, 
$\agentone$ is the agent role and $\user$ is the user role.

\subsection{Usecase scenario 2}

Use case scenario 1 implementation requires from the instrument program to process raw data
in a particular format (format 2) before sending them to the user program. 
In a more modular and 
fine-grain implementation, the instrument program should only send raw data to the ION
interface for processing and forwarding to the user. A separate session 
between the instrument and the ION interface and a separate session between
the ION interface and the user make a distinction into different logical
and processing levels. 

To capture the above implementation we assume a scenario
(depicted in Usecase 2 in Figure~\ref{fig:ooi}) 
with the user program (\CODE{U}), the instrument (\CODE{I}) 
and agents (\CODE{A}$_1$) and (\CODE{A}$_2$):

\begin{enumerate}
	\item	A new session connection $s_1$ is established between 
		(\CODE{U}), (\CODE{A$_1$}) and (\CODE{A$_2$}).
	\item	A new session connection $s_2$ is established 
		between (\CODE{A$_1$}), (\CODE{A$_2$}) and (\CODE{I}).
	\item	(\CODE{I}) sends raw data through $s_2$ to (\CODE{A$_1$}).
	\item	(\CODE{A$_1$}) sends processed data (format 1) through $s_1$ to (\CODE{U}).
	\item	(\CODE{A$_1$}) sends the acknowledgement through $s_2$ to (\CODE{I}).
	\item	(\CODE{I}) sends raw data through $s_2$ to (\CODE{A$_2$}).
	\item	(\CODE{A$_2$}) sends processed data (format 2) through $s_1$ to (\CODE{U}).
	\item	(\CODE{A$_2$}) sends the acknowledgement through $s_2$ to (\CODE{I}).
\end{enumerate}
The above scenario is implemented as follows:
\begin{eqnarray*}
	I_1 \Par A_1 \Par A_2 \Par U
\end{eqnarray*}
\noi where
\begin{eqnarray*}
	I_1 &=&	\req{b}{\instrument}{s_2} 
		\out{s_2}{\instrument}{\agentone}{\mathtt{rd}} 
		\inp{s_2}{\instrument}{\agentone}{x} 
		\out{s_2}{\instrument}{\agenttwo}{\mathtt{rd}}
		\inp{s_2}{\instrument}{\agentone}{x} \inact \\
	A_1 &=&	\acc{a}{\agentone}{s_1} \acc{b}{\agentone}{s_2} 
		\inp{s_2}{\agentone}{\instrument}{x} 
		\out{s_1}{\agentone}{\user}{\mathtt{pd}} 
		\out{s_2}{\agentone}{\instrument}{\mathtt{ack}} \inact\\
	A_2 &=&	\acc{a}{\agenttwo}{s_1} \acc{b}{\agenttwo}{s_2} 
		\inp{s_2}{\agenttwo}{\instrument}{x} 
		\out{s_1}{\agenttwo}{\user}{\mathtt{pd}} 
		\out{s_2}{\agenttwo}{\instrument}{\mathtt{ack}} \inact\\
	U &=&	\req{a}{\user}{s_1} \inp{s_1}{\user}{\agentone}{x} \inp{s_1}{\user}{\agenttwo}{y} \inact
\end{eqnarray*}
\noi and $\instrument$ is the instrument role, 
$\agentone$ and $\agenttwo$ are the agent roles and $\user$ is the user role.
Furthermore, for session $\s_1$ we let role 
$\instrumentzero$ (from scenario 1) as $\agenttwo$, 
since we maintain the session $s_1$ as it is defined in the scenario 1.

\subsection{Usecase scenario 3}
A step further is to enhance the performance of usecase
scenario 2 if the instrument (\CODE{I}) 
code in usecase scenario 2
can have a different implementation, where raw data
\NY{is} sent to both agents (\CODE{A$_1$}, \CODE{A$_2$})
before any acknowledgement is received. ION agents
can process data in parallel resulting in an
optimised implementation. This scenario is 
depicted in Usecase 3 in Figure~\ref{fig:ooi}. 



\begin{enumerate}
	\item	A new session connection $s_1$ is established between 
		(\CODE{U}), (\CODE{A$_1$}) and (\CODE{A$_2$}).
	\item	A new session connection $s_2$ is established 
		between (\CODE{A$_1$}), (\CODE{A$_2$}) and (\CODE{I}).
	\item	(\CODE{I}) sends raw data through $s_2$ to (\CODE{A$_1$}).
	\item	(\CODE{I}) sends raw data through $s_2$ to (\CODE{A$_2$}).
	\item	(\CODE{A$_1$}) sends processed data (format 1) through $s_1$ to (\CODE{U}).
	\item	(\CODE{A$_1$}) sends acknowledgement through $s_2$ to (\CODE{I}).
	\item	(\CODE{A$_2$}) sends processed data (format 2) through $s_1$ to (\CODE{U}).
	\item	(\CODE{A$_2$}) sends acknowledgement through $s_2$ to (\CODE{I}).
\end{enumerate}

The process is now refined as
\begin{eqnarray*}
	I_2 \Par A_1 \Par A_2 \Par U
\end{eqnarray*}
\noi where
\begin{eqnarray*}
	I_2 &=&	\req{b}{\instrument}{s_2} 
		\out{s_2}{\instrument}{\agentone}{\mathtt{rd}} 
		\out{s_2}{\instrument}{\agenttwo}{\mathtt{rd}}
		\inp{s_2}{\instrument}{\agentone}{x} 
		\inp{s_2}{\instrument}{\agentone}{x} \inact \\
\end{eqnarray*}
\noi and $\instrument$ implements the instrument role, $\agentone$ and $\agenttwo$ are the agent roles 
and $\user$ is the user role.

\subsection{Bisimulations}
The main concern of the three scenarios is to 
implement the Integrated Ocean Network interface 
respecting the multiparty communication protocols.

Having the user process as the observer we can see that typed
processes:
\begin{eqnarray*}
	\Gtprocess{I_0 \Par A}{\De_0}& \text{and} 
	&\Gtprocess{I_1 \Par A_1 \Par A_2}{\De_1}
\end{eqnarray*}
\noi are bisimilar (using $\swb$)
since in both process we observe
the following transition relations 
(recall that $\instrumentzero = \agenttwo$) :

\[
	\Gtprocess{I_0 \Par A}{\De_0} 
					\ \trans{\actacc{a}{s}{\agentone, \instrumentzero}}
					\ \trans{\tau} 
					\ \trans{\actout{s_1}{\agentone}{\user}{\mathtt{pd}}}
					\ \trans{\actout{s_1}{\instrumentzero}{\user}{\mathtt{pd}}}
\]
and
\[
	\Gtprocess{I_1 \Par A_1 \Par A_2}{\De_1} 
							\ \trans{\actacc{a}{s}{\agentone, \agenttwo}}
							\ \trans{\tau} 
							\ \trans{\actout{s_1}{\agentone}{\user}{\mathtt{pd}}}
							\ \trans{\actout{s_1}{\agenttwo}{\user}{\mathtt{pd}}}
\]

Next we give the bisimulation
closure. Let:
\[
\begin{array}{rclcl}
	\Gtprocess{I_0 \Par A}{\De_0} &\trans{\actacc{a}{s}{\agentone, \agenttwo}}& \Gtprocess{P_1}{\De_{01}}
	& \trans{\tau} & \Gtprocess{P_2}{\De_{02}}\\
	& \trans{\tau} & \Gtprocess{P_3}{\De_{03}}
	& \trans{\actout{s_1}{\agentone}{\user}{\mathtt{pd}}} & \Gtprocess{P_4}{\De_{04}}\\
	& \trans{\tau} & \Gtprocess{P_5}{\De_{05}}
	& \trans{\actout{s_1}{\instrumentzero}{\user}{\mathtt{pd}}} & \Gtprocess{P_6}{\De_{06}}\\
	\\
	\Gtprocess{I_1 \Par A_1 \Par A_2}{\De_1} &\trans{\actacc{a}{s}{\agentone, \agenttwo}}& 
	\Gtprocess{Q_1}{\De_{11}}
	& \trans{\tau} & \Gtprocess{Q_2}{\De_{12}}\\
	& \trans{\tau} & \Gtprocess{Q_3}{\De_{13}}
	& \trans{\actout{s_1}{\agentone}{\user}{\mathtt{pd}}} & \Gtprocess{Q_4}{\De_{14}}\\
	& \trans{\tau} & \Gtprocess{Q_5}{\De_{15}}
	& \trans{\tau} & \Gtprocess{Q_6}{\De_{16}}\\
	& \trans{\actout{s_1}{\agenttwo}{\user}{\mathtt{pd}}} & \Gtprocess{P_7}{\De_{17}}
	& \trans{\tau} & \Gtprocess{Q_8}{\De_{18}}
\end{array}
\]
The bisimulation closure is:
\[
\begin{array}{rcl}
	\mathcal{R} &=& 
	\set{(\Gtprocess{I_0 \Par A}{\De_0}, \Gtprocess{I_1 \Par A_1 \Par A_2}{\De_1}),
	(\Gtprocess{P_1}{\De_{01}}, \Gtprocess{Q_1}{\De_{11}})\\
	& &(\Gtprocess{P_2}{\De_{02}}, \Gtprocess{Q_2}{\De_{12}}),
	(\Gtprocess{P_3}{\De_{03}}, \Gtprocess{Q_3}{\De_{13}})\\
	& &(\Gtprocess{P_4}{\De_{04}}, \Gtprocess{Q_4}{\De_{14}}),
	(\Gtprocess{P_5}{\De_{05}}, \Gtprocess{Q_5}{\De_{15}})\\
	& &(\Gtprocess{P_5}{\De_{05}}, \Gtprocess{Q_6}{\De_{16}}),
	(\Gtprocess{P_6}{\De_{06}}, \Gtprocess{Q_7}{\De_{17}})\\
	& &(\Gtprocess{P_6}{\De_{06}}, \Gtprocess{Q_8}{\De_{18}})
	}
\end{array}
\]

The two implementations (scenario 1 and scenario 2)
are completely interchangeable with respect to $\swb$.

If we proceed with the case of the scenario 3 we can 
see that typed process $\Gtprocess{I_2 \Par A_1 \Par A_2}{\De_2}$
cannot be simulated (using $\swb$) by
scenarios 1 and 2, since we can observe the execution:

\[	
	\Gtprocess{I_1 \Par A_1 \Par A_2}{\De_1} 	\ \trans{\tau} 
							\ \trans{\actacc{a}{s}{\agentone, \agenttwo}}
							\ \trans{\actout{s_1}{\agenttwo}{\user}{\mathtt{pd}}}
\]

By changing the communication ordering in the ION private
session $s_2$ we change the communication behaviour on
the external session channel $s_1$. Nevertheless, the
communication behaviour remains the same if we take into
account the global multiparty protocol of $s_1$ and the way
it governs the behaviour of the three usecase scenarios.

Hence we use $\govwbs$. The definition of the global environment is as follows:
\begin{eqnarray*}
E & = & s_1: \valuegt{\agentone}{\user}{\mathtt{PD}} \valuegt{\agenttwo}{\user}{\mathtt{PD}}
\end{eqnarray*}
The global protocol governs processes $I_1 \Par A_1 \Par A_2$ (similarly, $I_0 \Par A$) 
and $I_2 \Par A_1 \Par A_2$
to always observe action $\trans{\actout{s_1}{\agenttwo}{\user}{\mathtt{pd}}}$
after action $\trans{\actout{s_1}{\agentone}{\user}{\mathtt{pd}}}$ for
both processes.

Also note that the global protocol for $s_2$ is not present in the
global environment, because $s_2$ is restricted. 
The specification and implementation of
session $s_2$ are abstracted from the behaviour of
session $s_1$.

\section{Related and Future Work}
\label{sec:related}
Session types \cite{THK,honda.vasconcelos.kubo:language-primitives}
have been studied over the last decade for a wide range of
process calculi and programming languages, 
as a typed foundation for structured communication programming. 
Recently several works developed
multiparty session types and their extensions. While typed behavioural
equivalences are one of the central topics of the $\pi$-calculus,
surprisingly the typed behavioural semantics based on session types
have been less explored
and focusing only on binary (two-party) sessions.

In this section we first compare our work in a broader context
in relation with
the previous work on typed behavioural theories in the
$\pi$-calculus.  We then discuss and compare our work 
with more specific results: behavioural theories in the binary session types and
bisimulations defined with environments.

\paragraph{\bf Typed behavioural theories in the $\pi$-calculus}
An effect of types to behaviours of processes was first studied with 
the IO-subtyping in \cite{PiSa96b}. Since types can limit
contexts (environments) where processes can interact, typed equivalences
usually offer {\em coarse} semantics than untyped semantics. 
After \cite{PiSa96b}, many works on typed $\pi$-calculi 
have investigated correctness of encodings of known concurrent and
sequential calculi in order to examine semantic
effects of proposed typing systems. 

The type discipline closely related
to session types is a family of linear typing systems. The
work \cite{LinearPi} first proposed a linearly typed barbed congruence and 
reasoned a tail-call optimisation of higher-order functions which are
encoded 
as processes. 
The work \cite{Yoshida96} had
used a bisimulation of graph-based types to prove the full abstraction
of encodings of the polyadic synchronous $\pi$-calculus into the
monadic synchronous $\pi$-calculus. 
Later typed equivalences of a
family of linear and affine calculi \cite{BHY,YBH04,BergerHY05} 
were used to encode 
PCF \cite{Plotkin1977223,Milner19771}, the simply typed $\lambda$-calculi with sums and products, and system F \cite{GirardJY:protyp}
fully abstractly (a fully abstract encoding of the $\lambda$-calculi 
was an open problem in \cite{MilnerR:funp}).  
The work \cite{YHB02} proposed a new bisimilarity
method associated with linear type structure and strong
normalisation. It presented applications to reason secrecy in
programming languages. A subsequent work \cite{HY02} adapted these results
to a practical direction. It proposes new typing
systems for secure higher-order and multi-threaded programming 
languages. 
In these works, typed properties, linearity and liveness, 
play a fundamental role in the analysis. In general, linear types 
are suitable to encode ``sequentiality'' in the sense of 
\cite{HylandJME:fulapi,AbramskyS:fulap}.

Our first bisimulation $\wb^s$ is classified 
as one of linear bisimulations, 
capturing a mixture between shared behaviours (interactions at 
shared names) and linear behaviours (interactions at session names). 
Hence it is coarser than the untyped semantics 
(see Example~\ref{ex:swb}). 
Contrast to these linear bisimulations, 
the governance bisimulation offers more {\em fine-grained} 
equivalences since the same typable processes are observed in different ways 
depending on a witness (global types). 
See the last paragraph for a relationship with environment bisimulations. 

\paragraph{\bf Behavioural theories in the binary session types}
Our work in \cite{DBLP:conf/forte/KouzapasYH11,KouzapasYHH13}
develops an {\em asynchronous binary} session typed behavioural theory with
event operations. A 
labelled transition system 
is defined on session type
process judgements and ensures properties,
such as linearity in the presence of asynchronous queues.
We then apply the theory to validate
a transformation between threaded and event servers based on
the Lauer-Needham duality \cite{LauerNeedham79}.
For reasoning this transformation, we use a confluence technique 
developed in  \cite{confPi}. We have established several up-to 
techniques using confluence and determinacy 
properties on reductions on typed session names.  
These useful up-to techniques are still applicable to our standard 
and governed bisimulations since the up-to bisimulation obtained 
in \cite{DBLP:conf/forte/KouzapasYH11,KouzapasYHH13} is only concerned 
on the local $\tau$-actions on session names. It is an interesting future work to investigate the up-to techniques or useful axioms which are specific to 
the governed bisimulation. 

The work \cite{PCPT12} proves that
the proof conversions induced by a Linear Logic interpretation
of session types 
coincide with an observational equivalence over
a strict subset of the binary synchronous session processes.
The approach is extended to 
the binary asynchronous and binary synchronous polymorphic session processes 
in \cite{DeYoungCPT12} and \cite{CairesPPT13}, respectively.

The main focus of our paper is {\em multiparty} session types
and governed bisimulation, whose definitions and
properties crucially depend on
information of global types.
In the first author's PhD thesis \cite{dkphdthesis},
we studied how governed bisimulations can be systematically developed
under various semantics including three kinds of asynchronous
semantics 
by modularly changing the LTS for processes, environments and global types.
For governed bisimulations, we can reuse all of the definitions
among four semantics
by only changing the conditions of the LTS of global types to suit each semantics. 

Another recent work \cite{DemangeonH11} gives a fully abstract encoding of a {\em binary} synchronous
session typed calculus into a linearly typed $\pi$-calculus \cite{BHY}.\footnote{The work \cite{Dardha:2012:STR:2370776.2370794} also 
uses linear types to 
encode binary session types for the first and higher-order 
$\pi$-calculi \cite{tlca07}, but it does not 
study full abstraction results with respect to 
a behavioural equivalence or bisimulation.} 
We believe the same encoding method is smoothly
applicable to $\swb$ since it is
defined solely based on the projected types (i.e.~local types).
However a governed bisimulation requires a global witness, hence
the additional global information would be required for full abstraction. 

\paragraph{\bf Behavioural semantics defined with environments}
The constructions of our work are hinted by
\cite{DBLP:journals/mscs/HennessyR04}
which studies typed behavioural semantics for the $\pi$-calculus
with IO-subtyping where
an LTS for pairs of typing environments
and processes is used for defining typed
testing equivalences and barbed congruence.
On the other hand, in \cite{DBLP:journals/mscs/HennessyR04}, 
the type environment indexing the
observational equivalence resembles more a dictator 
where the refinement can be obtained by the fact that 
the observer has only partial knowledge on the typings, 
than a coordinator like our approach.  
Several papers have developed bisimulations
for the higher-order $\pi$-calculus or its variants
using the information of the environments. 
In \cite{DBLP:conf/lics/SangiorgiKS07} the authors take a
general approach for developing a behavioural theory for
higher order processes, both in the $\lambda$-calculus and
the $\pi$-calculus. The bisimulation relations are developed
in the presence of an environment knowledge for higher
order communication. Congruence and compositionality of
processes are restricted with respect to the environment.
A recent paper
\cite{DBLP:conf/esop/KoutavasH11} uses
a pair of a process and an observer knowledge
set for the LTS. The knowledge set contains a mapping from first order
values to the higher-order processes, which allows a tractable
higher-order behavioural theory using the first-order LTS. 

We record
a choreographic type as the witness in the environment to obtain
fine-grained bisimulations of multiparty processes.
The highlight of our bisimulation construction is
an effective use of the semantics of
global types for LTSs of processes (cf.~[Inv] in Figure~\ref{fig:ELTS} and Definition \ref{def:glts}).
Global types can guide the coordination among parallel
threads giving explicit protocols, hence it is applicable to
a semantic-preserving optimisation (cf.~Example \ref{ex:govern} and 
\S~\ref{sec:app:ooi}).

\paragraph{\bf Future work}
While it is known that it is undecidable to check
$P\wb Q$ in the full $\pi$-calculus, it is an interesting future topic to
investigate automated bisimulation-checking techniques or 
finite axiomatisations 
for the governed bisimulations for some subset of
multiparty session processes. 

More practical future direction is incorporating with, not only well-known 
subtyping of session types \cite{GH05,DemangeonH11} but also 
advanced refinements for communication optimisation 
(such as asynchronous subtyping  
\cite{mostrous_yoshida_honda_esop09,mostrous09sessionbased} 
and asynchronous distributed states \cite{ChenH12}) 
to seek practical applications of governed bisimulations 
to, e.g.~parallel algorithms \cite{NYH12} and 
distributed computing \cite{ooi}.


\bibliographystyle{abbrv}
\bibliography{session}

\ifcamera\else
\iflong\else 
\fi
\appendix
\section{Proof for Theorem~\ref{the:subject}}
\label{app:sr}
We first state 
a substitution lemma, used in the proof of the subject 
reduction theorem (Theorem~\ref{the:subject}). The statement is 
from the substitution lemma in \cite{BettiniCDLDY08}. 

\begin{lem}[Substitution]
\label{lem:subst}
$ $
	\begin{itemize}
		\item	If $\Ga \cat x: S \proves P \hastype \De$ and $\Gproves{v}{S}$ then
			$\Ga \proves P\subst{v}{x} \hastype \De$.
		\item	If $\Ga \proves P \hastype \De \cat y: T$ then
			$\Ga \proves P \subst{\srole{s}{\p}}{y} \hastype \De \cat \srole{s}{\p}: T$
	\end{itemize}

\end{lem}

\begin{proof}
	The proof is a standard induction on the typing derivation.
\end{proof}

We state Subject congruence.

\begin{lem}[Subject Congruence]
	Let $\Gamma \proves P_1 \hastype \Delta$ and $P_1 \scong P_2$.
	Then $\Gamma \proves P_2 \hastype \Delta$.
\end{lem}

\begin{proof}
	The proof is standard and straightforward by induction.
\end{proof}

\subsection{Proof for Theorem~\ref{the:subject}}

\begin{proof}
	For the subject reduction proof we apply
	induction on the structure of the reduction relation.
	We present the two key cases: the rest is 
	similar with the subject reduction theorem for the communication 
	typing system in \cite{BettiniCDLDY08}.\\

	\Case{$\orule{Link}$}
	\noi Let:
	\[
		P = \acc{a}{\p}{x_1} P_1 \Par \dots \Par \req{a}{n}{x_n} P_n
	\]
	\noi We apply the
	typing rules $\trule{MAcc}$, $\trule{MReq}$ and $\trule{Conc}$ to obtain $\Gtprocess{P}{\De}$ with $\coherent{\De}$.
	Next assume that:
	\[
		P \red P' = \newsp{s}{P_1 \subst{\srole{s}{1}}{x_1} \Par \dots \Par P_n \subst{\srole{s}{n}}{x_n}}
	\]
	\noi From rule $\trule{Conc}$ and Lemma~\ref{lem:subst}, we obtain:
	\[ 
		\Gtprocess{P_1 \subst{\srole{s}{1}}{x_1} \Par \dots \Par P_n 
		\subst{\srole{s}{n}}{x_n}}
		{\Decat \typedrole{\s}{1} T_1 \dots \typedrole{\s}{n} T_n}
	\] 
with $\fcoherent{\set{\typedrole{\s}{1} T_1 \dots  \typedrole{\s}{n}
    \T_n}}$ (since each $T_i$ is a projection of $\Gamma(a)$).
	We apply rule $\trule{SRes}$
	to obtain $\Gtprocess{P'}{\De}$, as required.
	\\

	\Case{$\orule{Comm}$}
	\noi Let:
	\[
		P = \out{\s}{\p}{\q}{\va} P_1 \Par \inp{\s}{\q}{\p}{\x} P_2
	\]
	\noi and
	\[
		P \red P_1 \Par P_2 \subst{v}{x}
	\]
	\noi We apply typing rules
	$\trule{Send}$, $\trule{Rcv}$ and $\trule{Conc}$ to
        obtain: 
	$\Gtprocess{P}{\De}$ with:
	\[
		\De = \De_1 \cat \typedrole{\s}{\p} \tout{\q}{\U} T_\p \cat \typedrole{\s}{\q} \tinp{\p}{U} T_\q
	\]
	\noi and using Lemma~\ref{lem:subst} we obtain that
	$\Ga \proves P_1 \Par P_2 \subst{v}{x} \hastype \De_1$.

	\noi From the induction hypothesis we know that $\coherent{\De}$. 
	From the coherency of $\De$ and from Proposition~\ref{prop:bindual} we obtain: 
	\begin{eqnarray*}
		\coherent{\De_1}\\
		\proj{(\tout{\q}{\U} T_\p)}{\q} &=& \dual{(\tinp{\p}{U} \proj{\T_\q)}{\p}}
	\end{eqnarray*}
	\noi The latter result implies that 
	$\btout{U} (\proj{T_\p}{\q}) = \dual{\btinp{U} (\proj{\T_\q}{\p})}$,
	which in turn implies that 
	$\proj{T_\p}{\q} = \dual{\proj{T_\q}{\p}}$.
	From the last result and the coherency of $\De_1$ we get:
	\[
		\coherent{\De_1 \cat \typedrole{\s}{\p} T_\p \cat \typedrole{\s}{\q} T_\q}
	\]
Hence $\Gtprocess{P_1 \Par P_2 \subst{\va}{x}}{\De'}$ with 
$\De' = \De_1 \cat \typedrole{\s}{\p} T_\p \cat \typedrole{\s}{\q} T_\q$, and 
$\De'$ is coherent.
\end{proof}

\section{Proofs for Bisimulation Properties}
\label{app:bis-properties}

%

\subsection{Proof for Lemma \ref{lemma:Bisimulation_typing_relation}}
\label{app:Bisimulation_typing_relation}

\begin{proof}
	We use the coinductive method based on 
	the bisimilarity definition.
	Assume that for 
	$\Gtprocess{P_1}{\De_1} \swb \noGtprocess{P_2}{\De_2}$, 
	we have 
	$\De_1 \bistyp \De_2$. 
	Then by the definition of $\bistyp$, there exists $\De$ such that:
	\begin{eqnarray}
		\De_1 \red^* \De \textrm{ and } \De_2 \red^* \De \label{proof:bisimulation_typing_relation_1}
	\end{eqnarray}
	\noi Now assume that 
	$\Gtprocess{P_1}{\De_1} \trans{\ell} \noGtprocess{P_1'}{\De_1'}$ then,
	$\Gtprocess{P_2}{\De_2} \Trans{\ell} \noGtprocess{P_2'}{\De_2'}$ and 
	by the typed transition definition we obtain
	$(\Ga, \De_1) \trans{\ell} (\Ga, \De_1')$ and $(\Ga, \De_2) \Trans{\ell} (\Ga, \De_2')$.
	We need to show that $\De_1' \bistyp \De_2'$.

	We prove by a case analysis on the transition 
	$\trans{\ell}$ on $(\Ga, \De_1)$ and $(\Ga, \De_2)$.\\

		\noi {\bf Case $\ell = \tau$:}

		\noi We use the fact that $\trans{\tau}$ with $\equiv$
                coincides with $\red$.
		By Theorem~\ref{the:subject}, we obtain 
		that if $\Gtprocess{P_1}{\De_1}$ and $P_1 \red P_1'$ then
		$\Gtprocess{P_1'}{\De_1'}$ and $\De_1 \red \De_1'$ or $\De_1 = \De_1'$.

		For environment $\De_2$ we obtain that
		if $\Gtprocess{P_2}{\De_2}$ and $P_2 \Red P_2'$ then
		$\Gtprocess{P_2'}{\De_2'}$ and $\De_2 \red^\ast \De_2'$.
		From the coinductive hypothesis in 
(\ref{proof:bisimulation_typing_relation_1}), 
		we obtain that there exists $\De$ such that:
\[
		\begin{array}{c}
			\De_1 \red \De_1' \red^\ast \De \\
			\De_2 \red^\ast \De_2' \red^\ast \De
		\end{array} 
\]
		\noi or 
\[
		\begin{array}{c}
			\De_1 = \De_1' \red^\ast \De \\
			\De_2 \red^\ast \De_2' \red^\ast \De
		\end{array}
\]
		\noi as required.
		\\

		\noi {\bf Case $\ell = \actacc{a}{\p}{s}$ or $\ell = \actreq{a}{\p}{s}$:}

		\noi The environment tuple transition on $\ell$ is:
		\begin{eqnarray*}
			(\Ga, \De_1) &\trans{\ell}& 
			(\Ga, \De_1 \cat \typedrole{s}{\p} T_\p \cat \dots \cat \typedrole{s}{\q} T_\q)\\
			(\Ga, \De_2) &\Trans{} \trans{\ell} \Trans{}&
			(\Ga, \De_2'' \cat \typedrole{s}{\p} T_\p \cat \dots \cat \typedrole{s}{\q} T_\q)
		\end{eqnarray*}

		\noi We set 
		\[
			\De' = \De \cat \typedrole{s}{\p} T_\p \cat \dots \cat \typedrole{s}{\q} T_\q
		\]
		\noi to obtain:
		\begin{eqnarray*}
			\De_1 \cat \typedrole{s}{\p} T_\p \cat \dots \cat \typedrole{s}{\q} T_\q &\red^*& \De'\\
			\De_2'' \cat \typedrole{s}{\p} T_\p \cat \dots \cat \typedrole{s}{\q} T_\q &\red^*& \De'
		\end{eqnarray*}
		\noi by the coinductive hypothesis (\ref{proof:bisimulation_typing_relation_1}).
		\\

		\noi {\bf Case $\ell = \actout{s}{\p}{\q}{v}$:}

		\noi We know
		from the definition of environment transition, 
		$\srole{s}{\q} \notin \dom{\De_1}$ and $\srole{s}{\q} \notin \dom{\De_2}$ and
		thus $\srole{s}{\q} \notin \dom\De$.

		From the typed transition we know that $\De_1$ and $\De_2$
		have the form:
		\begin{eqnarray*}
			\De_1 &=& \typedrole{s}{\p} \tout{\q}{v} T \cat \De_1''\\
			\De_2 &=& \typedrole{s}{\p} \tout{\q}{v} T \cat \De_2''
		\end{eqnarray*}
		\noi and from the coinductive hypothesis (\ref{proof:bisimulation_typing_relation_1}),
		there exists $\De = \typedrole{s}{\p} \tout{\q}{v} T \cat \De''$ such that:
		\begin{eqnarray*}
			\De_1 &\red^*& \De\\
			\De_2 &\red^*& \De
		\end{eqnarray*}
		\noi Note that there is no reduction on $s[{\p}]$ 
		because $\srole{s}{\q} \notin \dom\De$.

		\noi From the environment transition relation we obtain that:
		\begin{eqnarray*}
			\De_1' &=& \typedrole{s}{\p} T \cat \De_1''\\
			\De_2' &=& \typedrole{s}{\p} T \cat \De_2''\\
		\end{eqnarray*}
		
		The last step is to set $\De' = \typedrole{s}{\p} T \cat \De''$
		to obtain
		$\De_1' \red^\ast \De'$ and $\De_2' \red^*{} \De'$ as required.
		\\

		\noi {\bf Case $\ell = \actbdel{s}{\p}{\q}{\s'}{\p'}$:}

		\noi This case follows a similar argumentation with the case $\ell = \actout{s}{\p}{\q}{v}$.

		\noi We know
		from the definition of environment transition, 
		$\srole{s}{\q} \notin \dom{\De_1}$ and $\srole{s}{\q} \notin \dom{\De_2}$, thus 
		$\srole{s}{\q} \notin \dom\De$. 

		From the typed transition we know that $\De_1$ and $\De_2$
		have the form:
		\begin{eqnarray*}
			\De_1 &=& \typedrole{s}{\p} \tout{\q}{T'} T \cat \De_1''\\
			\De_2 &=& \typedrole{s}{\p} \tout{\q}{T'} T \cat \De_2''
		\end{eqnarray*}
		\noi and from the coinductive hypothesis (\ref{proof:bisimulation_typing_relation_1}),
		there exists $\De = \typedrole{s}{\p} \tout{\q}{T'} T \cat \De''$ such that:
		\begin{eqnarray*}
			\De_1 &\red^*& \De\\
			\De_2 &\red^*& \De
		\end{eqnarray*}
		\noi From the environment transition relation we obtain that:
		\begin{eqnarray*}
			\De_1' &=& \typedrole{s}{\p} T \cat \De_1''\\
			\De_2' &=& \typedrole{s}{\p} T \cat \De_2''
		\end{eqnarray*}
		The last step is to set $\De' = \typedrole{s}{\p} T \cat \De''$
		to obtain
		$\De_1' \red^\ast \De'$ and $\De_2' \red^*{} \De'$, as required.
		\\

		\noi	The remaining cases  on session channel actions are similar.
\end{proof}

\subsection{Weakening and strengthening}
The following lemmas are essential for invariant properties. 

\begin{lem}[Weakening]
	\label{lemma:weakening}
	\begin{enumerate}
		\item	If $\tprocess{E, \Ga}{\PP}{\De}$ then
		\begin{itemize}
		\item 	$\tprocess{E \cat s:G, \Ga}{P}{\De}$.
		\item 	$E = E' \cat s:G$ and $\exists G'$ such that $\set{s:\G'} \red^\ast \set{s:\G}$ then
			$\tprocess{E' \cat s:G', \Ga}{P}{\De}$.
		\end{itemize}
	\item	If $(E, \Ga, \De) \trans{\ell} (E, \Ga', \De')$ then
	\begin{itemize}
		\item	$(E \cat s:G, \Ga, \De) \trans{\ell} (E \cat s:G, \Ga', \De')$
		\item	If $E = E' \cat s:G$ and $\set{s:G'} \red^\ast \set{s:G}$ then
				$(E' \cat s:G', \Ga, \De) \trans{\ell} (E' \cat s:G', \Ga', \De')$
		\end{itemize}
		\item	If	$\tprocess{E, \Ga}{P_1}{\De_2} \govwbs \noGtprocess{P_2}{\De_2}$
		\begin{itemize}
			\item	$\tprocess{E \cat s:G, \Ga}{P_1}{\De_2} \govwbs \noGtprocess{P_2}{\De_2}$
			\item	If $E = E' \cat s:G$ and $\set{s:G'} \red^\ast \set{s:G}$ then
				$\tprocess{E' \cat s:G', \Ga}{P_1}{\De_2} \govwbs \noGtprocess{P_2}{\De_2}$
		\end{itemize}
	\end{enumerate}
\end{lem}

\begin{proof}
	We only show Part 1. Other parts are similar. 

	\begin{itemize}

		\item	From the governance judgement definition we have that
			$E \red^* E'$ and $\projset{E'} \supseteq 
			\ifsynch \De \fi$.
Let $E \cat s:G \red E' \cat s:G$. Then
			$\projset{E' \cat s:G} = 
\projset{E'} \cup \projset{s:G}
			\supseteq \projset{E'} \supseteq 
			\ifsynch \De \fi$.

		\item	From the governance judgement definition we have that
			there exist 
$E_1$ an $G_1$ such that 
$E' \cat s:G \red^* E_1 \cat s:G_1$ and 
			$\projset{E_1 \cat s:G_1} \supseteq 
			\ifsynch \De \fi$.
Let $E' \cat s:G' \red^* E' \cat s:G \red^* E_1 \cat s:G_1$.
			Hence the result is immediate.\qedhere
	\end{itemize}
\end{proof}

\begin{lem}[Strengthening]
	\label{lemma:strengthening}
	\begin{enumerate}
		\item	If $\tprocess{E \cat s:G, \Ga}{\PP}{\De}$ and
		\begin{itemize}
			\item 	If $s \notin \fn{P}$ then $\tprocess{E, \Ga}{P}{\De}$
			\item	If $\exists G',G_1$ s.t. 
$E\cdot s:G \red^\ast E_2 \cdot s:G' \red^\ast E_1 \cdot s:G_1$ with 
$\projset{E_1 \cat s:G_1}\supseteq \Delta$, 
				then $\tprocess{E \cat s:G', \Ga}{P}{\De}$
		\end{itemize}

		\item	If $(E \cat s:G, \Ga, \De)
			\trans{\ell} (E' \cat s:G, \Ga', \De')$ then
		\begin{itemize}
			\item	$(E, \Ga, \De) \trans{\ell} (E', \Ga', \De')$
			\item	
If $\exists G'$ s.t.  
$E\cdot s:G \red^\ast E_2 \cdot s:G' \red^\ast E_1 \cdot s:G_1$ 
with $\projset{E_1 \cat s:G_1}\supseteq \Delta$,  
		$(E \cat s:G', \Ga, \De) \trans{\ell} (E' \cat s:G', \Ga', \De')$
		\end{itemize}

		\item	If $\tprocess{E \cat s:G, \Ga}{P_1}{\De_2} \govwbs 
			\noGtprocess{P_2}{\De_2}$
		\begin{itemize}
			\item	If $s \notin \fn{P}$ then 
		$\tprocess{E, \Ga}{P_1}{\De_2} \govwbs \noGtprocess{P_2}{\De_2}$
\item	
If $\exists G'$ s.t.  
$E\cdot s:G \red^\ast E_2 \cdot s:G' \red^\ast E_1 \cdot s:G_1$ 
with $\projset{E_1 \cat s:G_1}\supseteq \Delta$,  
			$\tprocess{E \cat s:G', \Ga}{P_1}{\De_2} \govwbs \noGtprocess{P_2}{\De_2}$
		\end{itemize}
	\end{enumerate}

\end{lem}

\begin{proof}
	We prove for part 1. Other parts are similar.

	\begin{itemize}	
		\item	From the governance judgement definition we have that
			$E \cat s:G \red^* E_1 \cat s:G_1$ and 
			$\projset{E_1 \cat s:G_1} = \projset{E_1} \cup \projset{s:G_1} \supseteq 
			\ifsynch \De \fi$.
			Since $s \notin \fn{P}$ then $s \notin \dom{\De}$, then
			$\projset{s:G_1} \cap \ifsynch \De \fi = \es$.
			So $\projset{E_1} \supseteq \ifsynch \De \fi$ and $E \red^* E_1$.

		\item	The result is immediate from the definition of	
			governance judgement.\qedhere
	\end{itemize}
\end{proof}

\subsection{Configuration Transition Properties}
\label{appsubsec:invariant}

\begin{lem}
\label{lem:global_subset_inv}
$ $ \\
	\begin{itemize}
		\item	If $E \trans{\actval{s}{\p}{\q}{U}} E'$ then
			$\set{\typedrole{s}{\p} \tout{\q}{U} T_\p, \typedrole{s}{\q} \tinp{\p}{U} T_\q} 
			\subseteq \projset{E}$
			and $\set{\typedrole{s}{\p} T_\p, \typedrole{s}{\q} T_\q} \subseteq \projset{E'}$.

		\item	If $E \trans{\actgsel{s}{\p}{\q}{l}} E'$ then
			$\set{\typedrole{s}{\p} \tsel{\q}{l_i: T_{i\p}}, 
			\typedrole{s}{\q} \tbra{\p}{l_i: T_{i\q}}} 
			\subseteq \projset{E}$
			and $\set{\typedrole{s}{\p} T_{k\p}, \typedrole{s}{\q} T_{k\q}} \subseteq \projset{E'}$
	\end{itemize}
\end{lem}

\begin{proof}
	Part 1: We apply inductive hypothesis on the structure of the definition of 
$\actval{s}{\p}{\q}{U}$.
	The base case
	\[
		\set{s:\valuegt{\p}{\q}{U} G} \trans{\actval{s}{\p}{\q}{U}} \set{s:G}
	\]
	\noi is easy since 
	\begin{eqnarray*}
		\set{\typedrole{s}{\p} \proj{(\valuegt{\p}{\q}{U} G)}{\p}, 
		\typedrole{s}{\q} \proj{(\valuegt{\p}{\q}{U} G)}{\q}} =\\ 
		\set{\typedrole{s}{\p} \tout{\q}{U} T_\p, \typedrole{s}{\q} \tinp{\p}{U} T_\q} 
		\subseteq \projset{s: \valuegt{\p}{\q}{U} G}
	\end{eqnarray*}
	\noi and
	\begin{eqnarray*}
		\set{\typedrole{s}{\p} \proj{G}{\p}, 
		\typedrole{s}{\q} \proj{G}{\q}} =
		\set{\typedrole{s}{\p} T_\p, \typedrole{s}{\q} T_\q} 
		\subseteq \projset{s: G}
	\end{eqnarray*}

	\noi The main induction rule concludes that:
	\[
		\set{s:\valuegt{\p'}{\q'}{U} G} \trans{\actval{s}{\p}{\q}{U}} \set{s:G}
	\]
	\noi if $\p \not= \p'$ and $\q \not= \q'$ and
	$\set{s:G} \trans{\actval{s}{\p}{\q}{U}} \set{s:G'}$.
	From the induction hypothesis we know that:
	\begin{eqnarray*}
		\set{\typedrole{s}{\p} \tout{\q}{U} T_\p, \typedrole{s}{\q} \tinp{\p}{U} T_\q} 
		&\subseteq& \projset{s:G}\\
		\set{\typedrole{s}{\p} T_\p, \typedrole{s}{\q} T_\q} 
		&\subseteq& \projset{s: G'}
	\end{eqnarray*}
	\noi to conclude that:
	\begin{eqnarray*}
		\set{\typedrole{s}{\p} \proj{(\valuegt{\p'}{\q'}{U} G)}{\p}, 
		\typedrole{s}{\q} \proj{(\valuegt{\p'}{\q'}{U} G)}{\q}} =\\
		\set{\typedrole{s}{\p} \proj{G}{\p}, 
		\typedrole{s}{\q} \proj{G}{\q}} =\\  
		\set{\typedrole{s}{\p} \tout{\q}{U} T_\p, \typedrole{s}{\q} \tinp{\p}{U} T_\q} 
		\subseteq \projset{s: G}
	\end{eqnarray*}
	\noi and
	\begin{eqnarray*}
		\set{\typedrole{s}{\p} \proj{(\valuegt{\p'}{\q'}{U} G')}{\p}, 
		\typedrole{s}{\q} \proj{(\valuegt{\p'}{\q'}{U} G')}{\q}} =\\
		\set{\typedrole{s}{\p} \proj{G'}{\p}, 
		\typedrole{s}{\q} \proj{G'}{\q}} =\\
		\set{\typedrole{s}{\p} T_\p, \typedrole{s}{\q} T_\q} 
		\subseteq \projset{s: G'}
	\end{eqnarray*}
	\noi as required.

	Part 2: Similar.

	Part 3:
	From the global configuration transition relation (Definition~\ref{def:glts}), we obtain that
	\begin{eqnarray*}
		(E_1, \Ga_1, \De_1) &\trans{\ell}& (E_2, \Ga_2, \De_2)\\
		\Ga_1 \proves P_1 \hastype \De_1 &\trans{\ell}& \Ga_2 \proves P_2 \hastype \De_2
	\end{eqnarray*}
	Then from the definition of governed environment, we can show that
	$E_2, \Ga_2 \proves P_2 \hastype \De_2$ is a governed judgement: this is because 
	$(E_2, \Ga_2, \De_2)$ is an environment configuration, hence 
	$\exists E_2', \ E_2 \red E_2'$ and $\projset{E_2'} \supseteq \De_2$. 
\end{proof}

\paragraph{\bf Proof for Proposition \ref{pro:invariants}}
\begin{proof}
	{\bf (1)} We apply induction on the definition structure of $\trans{\ell}$.

\noi	{\bf Basic Step:}\\[1mm]
	\noi {\bf Case: $\ell = \actreq{a}{s}{A}$}.
From rule $\eltsrule{Acc}$ we obtain
	\[ 
		(E_1, \Ga_1, \De_1) \trans{\ell} 
		(E_1 \cat s:G, \Ga_1, \De_1 \cat \set{s[\p_i]: \proj{s}{\p_i}}_{\p_i \in A})
	\]
	\noi From the environment configuration definition we obtain that
	\[
		\exists E_1' \textrm{ such that } E_1 \red^* E_1', 
\quad 
\ifsynch
\projset{E_1'} \supseteq \De_1
\fi
	\]
	\noi We also obtain that $\projset{s:G} \supseteq \set{s[\p_i]: \proj{G}{\p_i}}_{i \in A}$.
	Thus we can conclude that 
	\[
\begin{array}{l}
		E_1 \cat s:G \red^* E_1' \cat s:G\\
		\projset{E_1 \cat s:G} \supseteq \De_1 \cat \set{s[\p_i]: 
\proj{G}{\p_i}}_{\p_i \in A}
\end{array}
	\]

	\noi {\bf Case: $\ell = \actreq{a}{s}{A}$}. Similar as above.
	\\

	\noi {\bf Case: $\ell = \actout{s}{\p}{\q}{v}$}.
From rule $\eltsrule{Out}$ we obtain
	\begin{eqnarray}
		(E_1, \Ga, \Decat \typedrole{s}{\p} \tout{\q}{U} T) &\trans{\ell}& (E_2, \Ga, \Decat \typedrole{s}{\p} T) \label{elts:inv1}\\ 	
		\projset{E_1} &\supseteq& \Decat \typedrole{s}{\p} \tout{\q}{U} T \label{elts:inv2}\\
		E_1 &\trans{\actval{s}{\p}{\q}{U}}& E_2 \label{elts:inv3}
	\end{eqnarray}

	From (\ref{elts:inv2}), we obtain 
	$\projset{E_1} \supseteq \Decat \set{\typedrole{s}{\p} \tout{\q}{U} T \cat \typedrole{s}{\q} \tinp{\p}{U} T'}$
	and from (\ref{elts:inv3}) and Lemma
        \ref{lem:global_subset_inv}, we obtain that $\projset{E_2}
        \supseteq \Decat \set{\typedrole{s}{\p} T \cat
          \typedrole{s}{\q} T'}$. \\

	\noi {\bf Case: $\ell = \actbdel{s}{\p}{\q}{s'}{\p'}$}.
	\begin{eqnarray}
		(E_1, \Ga, \Decat \typedrole{s}{\p} \tout{\q}{T\p'} T) &\trans{\ell}& 
		(E_2 \cat s:G, \Ga, \Decat \typedrole{s}{\p} T \cat \set{\typedrole{s}{\p_i} \proj{G}{\p_i}}) \label{elts:inv4}\\ 	
		\projset{E_1} &\supseteq& \Decat \typedrole{s}{\p} \tout{\q}{T_\p'} T \label{elts:inv5}\\
		E_1 &\trans{\actval{s}{\p}{\q}{T_\p'}}& E_2 \label{elts:inv6}\\
		\projset{s:G} &\supseteq& \set{\typedrole{s}{\p_i} \proj{G}{\p_i}} \label{elts:inv7}
	\end{eqnarray}
	From (\ref{elts:inv5}) we obtain
	$\projset{E_1} \supseteq \Decat \set{\typedrole{s}{\p} \tout{\q}{U} T \cat \typedrole{s}{\q} \tinp{\p}{U} T'}$
	and from (\ref{elts:inv6}) and Lemma \ref{lem:global_subset_inv} we obtain that $\projset{E_2} \supseteq \Decat \set{\typedrole{s}{\p} T \cat \typedrole{s}{\q} T'} \supset \Decat \typedrole{s}{\p} T$. From (\ref{elts:inv7}) we obtain that
	$\projset{E_2 \cat s:G} \supseteq \Decat \typedrole{s}{\p} T \cat \set{\typedrole{s}{\p_i} \proj{G}{\p_i}}$, as required.

	\noi The rest of the base cases are similar.\\

	\noi {\bf Inductive Step:}

	\noi The inductive rule for environment configuration is $\eltsrule{Inv}$.
	Let $(E_1, \Ga_1, \De_1) \trans{\ell} (E_2, \Ga_2, \De_2)$. From rule $\eltsrule{Inv}$
	we obtain:
	\begin{eqnarray}
		E_1 &\red^*& E_1' \label{elts:inv8}\\
		(E_1', \Ga_1, \De_1) &\trans{\ell}& (E_2, \Ga_2, \De_2) \label{elts:inv9}
	\end{eqnarray}
	From the inductive hypothesis we know that, from (\ref{elts:inv9}),
	there exists $E_3$ such that $E_2 \red^* E_3$ and $\De_2 \subseteq \projset{E_3}$.
	Then the result is by (\ref{elts:inv8}).
\end{proof}

\begin{lem}
	\label{lemma:configuration_transition}
$ $\\
	\begin{enumerate}
		\item	If	$(E, \Ga, \De_1) \trans{\ell} (E', \Ga', \De_2)$ then
				$(\Ga, \De_1) \trans{\ell} (\Ga', \De_2)$
		\item 	If	$(E, \Ga, \De_1) \trans{\ell} (E', \Ga', \De_1')$ and
				$\De_1 \bistyp \De_2$ then
				$(E, \Ga, \De_2) \Trans{\ell} (E', \Ga', \De_2')$

		\item	If	$(\Ga, \De_1) \trans{\ell} (\Ga', \De_2)$ then 
				there exists $E$ such that 
				$(E, \Ga, \De_1) \trans{\ell} (E', \Ga', \De_2)$

		\item	If	$(E, \Ga, \De \cat \typedrole{s}{\p} T_\p) \trans{\ell} 
				(E', \Ga, \De' \cat \typedrole{s}{\p} T_\p)$ then
				$(E, \Ga, \De) \trans{\ell} (E', \Ga, \De')$

		\item	If	$(E, \Ga, \De_1) \trans{\ell} (E', \Ga, \De_2)$
			then
				$(E, \Ga, \De_1 \cat \De)
				\transs{\ell} 
				(E', \Ga, \De_2 \cat \De)$ 
				provided that if
				$(E, \Ga, \De) \trans{\ell'} (E, \Ga, \De')$
				then $\ell \not\coh \ell'$
	\end{enumerate}
\end{lem}

\begin{proof}

	Part 1:

	\noi	The proof for part 1 is easy to be implied
		by a case analysis on the configuration
		transition definition with respect to
		environment transition definition.

	Part 2:

	\noi	By the case analysis on $\ell$.

	\noi	{\bf Case $\ell = \tau$:} The result is trivial.

	\noi	{\bf Case $\ell = \actreq{a}{\p}{s}$ or $\ell = \actacc{a}{\p}{s}$:}
		The result comes from a simple transition.

	\noi	{\bf Case $\ell = \actout{s}{\p}{\q}{v}$:} 
		$\De_1 \bistyp \De_2$ implies $\De_1 \red^* \De$ and
		$\De_2 \red^* \De$ for some $\De$ and
		$\De = \De' \cat \typedrole{s}{\p} \tout{\q}{U} T$
.

Hence $(E, \Ga, \De_2) \Trans{} (E, \Ga, \De) \trans{\ell}$ as required.

	\noi	{\bf Case $\ell = \actbdel{s}{\p}{\q}{\s'}{\p'}$:} 
		$\De_1 \bistyp \De_2$ implies $\De_1 \red^* \De$ and
		$\De_2 \red^* \De$ for some $\De$ and
		$\De = \De' \cat \typedrole{s}{\p} \tout{\q}{T'} T$
.

	\noi	$(E, \Ga, \De_2) \Trans{} (E, \Ga, \De) \trans{\ell}$ as required.
	The remaining cases are similar.

	Part 3:

	\noi	We do a case analysis on $\ell$.

	\noi	{\bf Cases $\ell = \tau, \ell = \actreq{a}{\p}{s}, \ell = \actacc{a}{\p}{s}$:}
		The result holds for any $E$.

	\noi	{\bf Case $\ell = \actout{s}{\p}{\q}{v}:$}
		$\De_1 = \De_1' \cat \De_1''$ with
		$\De_1'' = 
		\typedrole{s}{\p} \tout{\q}{U} T_\p \cat \dots \cat \typedrole{s}{\gtfont{r}} T_\gtfont{r}$ 
		Choose $E = E' \cat s:G$ with 
		$\ifsynch \De_1'' \fi \subseteq \projset {s:G}$ and
		$\typedrole{s}{\q} \tinp{\p}{U} T_\q \in \projset {s:G}$ and
		$\ifsynch \De_1 \fi \subseteq \projset{E}$. 
		By the definition of configuration transition relation, we obtain 
		$(E, \Ga, \De) \trans{\ell} (E, \Ga', \De_2)$, as required.

%
Remaining cases are similar.

	Part 4:

	\noi	$(E, \Ga, \De \cat \typedrole{s}{\p} T_\p)
		\trans{\ell}
		(E', \Ga, \De' \cat \typedrole{s}{\p} T_\p)$
		implies that $\srole{s}{\p} \notin \subj{\ell}$.
		The result then follows from the definition of
		configuration transition.

	Part 5:

	\noi	{\bf Case $\ell = \tau, \ell = \actreq{a}{\p}{s}, \ell = \actacc{a}{\p}{s}$:}
		The result holds by definition of the configuration transition.

	\noi	{\bf Case $\ell = \actout{s}{\p}{\q}{U}$:}
		we have that
		$\De_1 = \De_1' \cat \typedrole{s}{\p} \tout{\q}{U} T$
		and $E \trans{\actval{s}{\p}{\q}{U}} E'$.
		$\srole{s}{\q} \in \De$, then by
		definition of weak configuration pair
		we have $\De = \De'' \cat \typedrole{s}{\q}\tinp{\p}{U}T$
		and 
		$(E, \Ga, \De) \trans{\actinp{\s}{\q}{\p}{U}}$. But this
		contradicts with the assumption $\ell \not\coh \ell'$, so
		$\srole{s}{\q} \notin \De$.
		By the definition of configuration pair transition
		we obtain that 
		$(E, \Ga, \De_1 \cat \De) \trans{\actout{s}{\p}{\q}{U}} (E, \Ga, \De_2 \cat \De)$.
	\noi	Remaining cases are similar.
\end{proof}


\subsection{Proof for Lemma \ref{lem:bis-cong} (1)}
\label{app:bis-cong}

\begin{proof}
	Since we are dealing with closed processes, the interesting
	case is parallel composition. We need to show that if
	$\tprocess{E, \Ga}{P}{\De_1} \govwbs \tprocess{E, \Ga}{Q}{\De_2}$ then
	for all 
	$R$ such that 
	$\tprocess{E, \Ga}{P \Par R}{\De_3},
	\tprocess{E, \Ga}{Q \Par R}{\De_4}$
	then $\tprocess{E, \Ga}{P \Par R}{\De_3} \govwbs \noGtprocess{Q \Par R}{\De_4}$.

	We define the following configuration relation.
	\[
		\begin{array}{rl}
		S = &	\set{(\tprocess{E, \Ga}{P \Par R}{\De_3},\ \tprocess{E, \Ga}{Q \Par R}{\De_4}) \setbar \\
			& \tprocess{E, \Ga}{P}{\De_1} \govwbs \noGtprocess{Q}{\De_2},\\
			& \forall R \textrm{ such that }
			\tprocess{E, \Ga}{P \Par R}{\De_3}, \tprocess{E, \Ga}{Q \Par R}{\De_4}
			}
		\end{array}
	\]

	Before we proceed to a case analysis, we extract 
	general results.
	Let
	$\Gtprocess{P}{\De_1}, 
	\Gtprocess{Q}{\De_2},
	\Gtprocess{R}{\De_5},
	\Gtprocess{P \Par R}{\De_3},
	\Gtprocess{Q \Par R}{\De_4}$
	\noi then from typing rule $\trule{Conc}$ we obtain
	\begin{eqnarray}
		\De_3 &=& \De_1 \cup \De_5 \label{bis_is_cong:1}\\
		\De_4 &=& \De_2 \cup \De_5 \label{bis_is_cong:2}\\
		\De_1 \cap \De_5 & = & \es \label{bis_is_cong:3}\\
		\De_2 \cap \De_5 & = & \es \label{bis_is_cong:4}
	\end{eqnarray}

	We prove that $S$ is a bisimulation. There are three cases:

	\Case{1}
	\[
		\tprocess{E, \Ga}{P \Par R}{\De_3} \trans{\ell} \tprocess{E_1', \Ga}{P' \Par R}{\De_3'}
	\]
with  $\bn{\ell} \cap \fn{R} = \es$. 

	\noi	From typed transition definition we have that:
		\begin{eqnarray}
			P \Par R &\trans{\ell}& P' \Par R \label{bis_is_cong:case1_1} \\
			(E, \Ga, \De_3) &\trans{\ell}& (\E_1', \Ga, \De_3') \label{bis_is_cong:case1_2}
		\end{eqnarray}

	\noi	Transition (\ref{bis_is_cong:case1_1}) and rule $\ltsrule{Par}$
		(LTS in Figure \ref{fig-synch-lts}) imply:
		\begin{eqnarray}
			P &\trans{\ell}& P' \label{bis_is_cong:case1_3}
		\end{eqnarray}
	\noi	From (\ref{bis_is_cong:1}),
		transition (\ref{bis_is_cong:case1_2}) can be written as
		$(E, \Ga, \De_1 \cup \De_5) \trans{\ell} (\E_1', \Ga, \De_1' \cup \De_5)$, 
		to conclude from part 4 of Lemma \ref{lemma:configuration_transition}, that:
		\begin{eqnarray}
			(E, \Ga, \De_1) &\trans{\ell}& (E_1', \Ga, \De_1') \label{bis_is_cong:case1_4}\\
			\subj{\ell} &\notin& \dom{\De_5} \label{bis_is_cong:case1_6}
		\end{eqnarray}

	\noi	Transitions (\ref{bis_is_cong:case1_3}) and (\ref{bis_is_cong:case1_4})
		imply
		$\tprocess{E, \Ga}{P}{\De_1} \trans{\ell} \tprocess{E_1', \Ga}{P'}{\De_1'}$.
		From the definition of set $S$ we obtain 
		$\tprocess{E, \Ga}{Q}{\De_2} \Trans{\ell} \tprocess{E_2', \Ga}{Q'}{\De_2'}$.

	\noi	From the typed transition definition we have that:
		\begin{eqnarray}
			Q &\Trans{\ell}& Q' \label{bis_is_cong:case1_5} \\
			(E, \Ga, \De_2) &\Trans{\ell}& (E_2', \Ga, \De_2')
		\end{eqnarray}

	\noi	From (\ref{bis_is_cong:case1_6}) and part 5 of Lemma
		\ref{lemma:configuration_transition} we can write:
		$(E, \Ga, \De_2 \cup \De_5) \Trans{\ell} (E_2', \Ga, \De_2' \cup \De_5)$,
		which implies, from (\ref{bis_is_cong:case1_5}), 
		$\tprocess{E, \Ga}{Q \Par R}{\De_4} \Trans{\ell} \tprocess{E_2', \Ga}{Q' \Par R}{\De_4'}$.
		Furthermore, we can see that:
		\[
			(\tprocess{E_1' \sqcup E_2', \Ga}{P' \Par R}{\De_4'}\ , \ \tprocess{E_1' \sqcup E_2', \Ga}{Q' \Par R}{\De_4'}) \in S
		\]
		as required, noting $E_1' \sqcup E_2'$ is 
defined by the definition of $S$.\\

\noi	\Case{2}
	\[
		\tprocess{E, \Ga}{P \Par R}{\De_3} \trans{\tau} \tprocess{E'}{P' \Par R'}{\De_3'}
	\]
	\noi From the typed transition definition, we have that:
	\begin{eqnarray}
		P \Par R &\trans{\tau}& P' \Par R' \label{bis_is_cong:case2_1}\\
		(E, \Ga, \De_3) &\trans{\tau}& (E', \Ga, \De_3') \label{bis_is_cong:case2_2}
	\end{eqnarray}

	\noi From (\ref{bis_is_cong:case2_1}) and rule $\ltsrule{Tau}$, we obtain
	\begin{eqnarray}
		P &\trans{\ell}& P' \label{bis_is_cong:case2_3}\\
		R &\trans{\ell'}& R' \label{bis_is_cong:case2_4}
	\end{eqnarray}
	\noi From (\ref{bis_is_cong:1}), transition (\ref{bis_is_cong:case2_2}) can be written
	$(E, \Ga, \De_1 \cup \De_5) \trans{\tau} (E', \Ga, \De_1' \cup \De_5')$, to conclude that
	\begin{eqnarray}
		(E, \Ga, \De_1) \trans{\ell} (E', \Ga, \De_1') \label{bis_is_cong:case2_5}\\
		(E, \Ga, \De_5) \trans{\ell'} (E', \Ga, \De_5') \label{bis_is_cong:case2_6}
	\end{eqnarray}
	\noi From (\ref{bis_is_cong:case2_3}) and (\ref{bis_is_cong:case2_5}), we conclude that
	$\tprocess{E, \Ga}{P}{\De_1} \trans{\ell} \tprocess{E', \Ga}{P'}{\De_1'}$ and from 
	(\ref{bis_is_cong:case2_4}) and (\ref{bis_is_cong:case2_6}), 
	we have $\tprocess{E, \Ga}{R}{\De_5} \trans{\ell'} \tprocess{E', \Ga}{R'}{\De_5'}$.

	\noi From the definition of set $S$, we obtain that 
	$\tprocess{E, \Ga}{Q}{\De_2} \Trans{\ell} \tprocess{E, \Ga}{Q'}{\De_2'}$ implies
	\begin{eqnarray}
		Q &\Trans{\ell}& Q' \\
		(E, \Ga, \De_2) &\Trans{\ell}& (E', \Ga, \De_2')
	\end{eqnarray}
	\noi From (\ref{bis_is_cong:case2_4}), we obtain that $Q \Par R \Trans{\tau} Q' \Par R'$
	and $(E, \Ga, \De_2 \cup \De_5) \Trans{\tau} (E', \Ga, \De_2' \cup \De_5')$, implies:
	\[
		\tprocess{E, \Ga}{Q \Par R}{\De_4} \Trans{\tau} \tprocess{E'}{Q' \Par R'}{\De_4'}
	\]
	with
	\[
		(\tprocess{E', \Ga}{P' \Par R'}{\De_4'}\ , \ \tprocess{E', \Ga}{Q' \Par R'}{\De_4'}) \in S
	\]
	as required.\\

\noi	\Case{3}
	\[
		\tprocess{E, \Ga}{P \Par R}{\De_3} \trans{\ell} \tprocess{E',\Ga}{P \Par R'}{\De_3'}
	\]
This case is similar with the above cases.

\subsection{Proof for Lemma \ref{lem:bis-cong} (2)}
\label{app:bis-complete}
The proof for the completeness follows the technique which uses 
the testers in \cite{Hennessy07}. 
We need to adapt the testers to multiparty session types. 

\begin{defi}[Definability]\label{def:definability}
Let $\obj{\ell}$ and $\subj{\ell}$ to denote a set of 
object and subject of $\ell$, respectively. 
	Let $N$ be a finite set of shared names and session endpoints for testing the
	receiving objects defined as   
	$N \bnfis \es \bnfbar N \cat \srole{s}{\p} \bnfbar N \cat a$. 
	An external action $\ell$ is {\em definable} if for a set of names
	$N$, fresh session $\suc$, $n$ is the dual endpoint of $\subj{\ell}$ (i.e. the dual endpoint of $s[\p][\q]$ is $s[\q]$), 
	there is a {\em testing process} $\funcbr{T}{N, \suc, \ell}$ 
	with the property that for every process $\PP$ and $\fn{\PP} \subseteq N$,
	\begin{itemize}
		\item	$\tprocess{E_1, \Ga}{P}{\De_1} \trans{\ell} \tprocess{E_1', \Ga'}{P'}{\De_1'}$
			implies that
			$\tprocess{E, \Ga}{\funcbr{T}{N, \suc, \ell} \Par P}{\De}
			\Red\\
			\tprocess{E', \Ga}{\newsp{\bn{\ell}, b}{\PP' \Par \out{\suc}{1}{2}{\obj{\ell}, n} \inact}}{\De'}$ 
		\item	$\tprocess{E, \Ga}{\funcbr{T}{N, \suc, \ell} \Par P}{\De}
			\Red 
			\tprocess{E', \Ga}{Q}{\Delta'}$ and 
			$\tprocess{E', \Ga}{Q}{\Delta'}\Downarrow \suc$ implies 
			for some 
			$\tprocess{E_1, \Ga}{\PP}{\De_1}\trans{\ell}
			\tprocess{E'_1, \Ga'}{\PP'}{\De_1'}$, we have
			$\Q \equiv \newsp{\bn{\ell}, b}{\PP'\Par \out{\suc}{1}{2}{\obj{\ell}, n} \inact}$. 
	\end{itemize}
%
\end{defi}
\noindent Hereafter we omit the environments if they are obvious from the
context. 

\begin{lem}[Definability]
	\label{lem:definibility}
	Every external action is definable.
\end{lem}
\begin{proof}
	The cases of the input actions and the session initialisation (accept
	and request)  are straightforward \cite{Hennessy07}:
	\begin{enumerate}
		\item	$\funcbr{T}{N, \suc, \actacc{a}{A}{s}} =$


			$\newsp{b}{\req{a}{n}{x}\suc !\ENCan{\true};\acc{b}{1}{x}R \Par 
			\acc{a}{\p_1}{x}\acc{b}{1}{x}R_1 \Par \cdots \Par 
			\acc{a}{\p_m}{x}\acc{b}{1}{x}R_m}$ 

		\noi with $\p_1, \dots, \p_m \notin A$ and 
		$\{\p_1, \dots, \p_m \}\cup A$ complete w.r.t. 
		$n = \max(\{\p_1, \dots, \p_m \}\cup A)$.  

		\item	$\funcbr{T}{N, \suc, \actinp{s}{\p}{\q}{v}} = \out{s}{\q}{\p}{v}\out{\suc}{1}{2}{\srole{s}{\q}} \inact$

		\item	$\funcbr{T}{N, \suc, \actbra{s}{\p}{\q}{l}} = \sel{s}{\q}{\p}{l}\out{\suc}{1}{2}{\srole{s}{\q}} \inact$

		\item 	$\funcbr{T}{N, \suc, \actreq{a}{A}{s}}= \newsp{b}{\acc{a}{\p_1}{x}\suc !\ENCan{\true};\acc{b}{1}{x} R_1  \Par \dots \Par \acc{a}{\p_m}{x};\acc{b}{1}{x} R_m}$

		\noi with $\p_1, \dots, \p_m \notin A$ and 
		$\{\p_1, \dots, \p_m \}\cup A$ complete w.r.t.  
		$\max(\{\p_1, \dots, \p_m \}\cup A)$. 
	\end{enumerate}
	The requirements of Definition \ref{def:definability} is verified 
	straightforwardly. 

	For the output cases, we use the matching operator as  
	\cite[\S~2.7]{Hennessy07}. 

	\begin{enumerate}
		\setcounter{enumi}{4}
		\item	$\funcbr{T}{N, \suc, \actout{s}{\p}{\q}{v}}=$

			$\inp{s}{\q}{\p}{x}$

			$\ifthen{x = v}{\out{\suc}{1}{2}{x, \srole{s}{\q}} \inact}
			\Else
			\newsp{b}{\acc{b}{1}{x} \out{\suc}{1}{2}{x, \srole{s}{\q}} \inact}$

		\item 
			$\funcbr{T}{N, \suc, \actbout{s}{\p}{\q}{v}}=$

			$\inp{s}{\q}{\p}{x}$

			$\ifthen{x \notin N}{\out{\suc}{1}{2}{x, \srole{s}{\q}} \inact}
			\Else
			\newsp{b}{\acc{b}{1}{x} \out{\suc}{1}{2}{x, \srole{s}{\q}} \inact}$

		\item 
			$\funcbr{T}{N, \suc, \actsel{s}{\p}{\q}{l_k}}=$

			$\bra{s}{\q}{\p}{l_k: \out{\suc}{1}{2}{\srole{s}{\q}} \inact,
			\quad l_i: \newsp{b}{\acc{b}{1}{x}\out{\suc}{1}{2}{\srole{s}{\q}} \inact}}_{i\in I\setminus k}$
	\end{enumerate}
	The requirements of Definition \ref{def:definability} are straightforward to verify. 
	Note that we need to have process
	$\out{\suc}{1}{2}{\srole{s}{\q}} \inact$
	on both conditions in the if-statement since $\suc$ is a session
	channel (see $\trule{If}$ in Figure~\ref{fig:synch-typing} 
	in \S~\ref{sec:typing}). 
\end{proof}

The next lemma  follows \cite[Lemma 2.38]{Hennessy07}. 

\begin{lem}[Extrusion]\label{lem:extruction}
	Assume $\suc$ is fresh and $b\not\in \{\vec{m}\}\cup\fn{P}\cup
	\fn{Q} $ and $\{\vec{m}\}\subseteq \fn{v} \subseteq \set{\vec{n}}$. 
	\begin{eqnarray}
		\tprocess{E, \Ga}{(\nu \vec{m}, b) (P \Par
                  \out{\suc}{1}{2}{\vec{n}, \srole{s}{\q}} \inact \Par
                  \prod_i R_i)}{\Delta_1}\\
		\label{definability-cong-assumption} 
		\quad\quad\quad \cong\ (\nu \vec{m}, b) (Q \Par \out{\suc}{1}{2}{\vec{n}, \srole{s}{\q}} \inact \Par \prod_i R_i) \hastype \Delta_2
	\end{eqnarray}
	with $R_i=\acc{b}{1}{x}R_i'$ then
	\begin{eqnarray}
		\label{definability-cong-result}
		\tprocess{E', \Ga'}{P}{\Delta_1'} \cong Q \hastype \Delta_2'
	\end{eqnarray}
\end{lem}

\begin{proof}
\noi Let relation
\[
	\begin{array}{rcl}
		\relfont{S} & = & \set{ (\tprocess{E',\Ga'}{P}{\Delta_1'}, 
		\tprocess{E',\Ga'}{Q}{\Delta_2'}) \setbar \\ 
		& & \tprocess{E, \Ga}{(\nu \vec{m}, b) (P \Par \out{\suc}{1}{2}{\vec{n}, \srole{s}{\q}} \inact \Par \prod_i R_i)}{\Delta_1}
		\govcongs \\
		&&(\nu \vec{m}, b) (Q \Par \out{\suc}{1}{2}{\vec{n}, \srole{s}{\q}} \inact \Par \prod_i R_i) \hastype \Delta_2 }
	\end{array}
\]
where we assume 
$\suc$ is fresh and $b\not\in \{\vec{m}\}\cup\fn{P}\cup
\fn{Q}$.
We will show that $\relfont{S}$ is governed reduction-closed.

\paragraph{\bf Typability}
	We show $\relfont{S}$ is a typed relation. From the definition of 
	$\relfont{S}$, we have 
	$\De_1' \ordercup \De'_2$. 
	By using 
	typing rules $\trule{NRes}$, $\trule{SRes}$, $\trule{Conc}$,
	we obtain 
	$(\tprocess{\Ga'}{P}{\Delta_1'},  
	\tprocess{\Ga'}{Q}{\Delta_2'})$ is in the typed 
	relation. Then we can set $E' =E\cup 
	E_0\cup \{\suc:1\to 2:U.\inactgt\}$ 
	where $E_0 = \{s:G_s\}$ if $s=\vec{m}$; otherwise $E_0=\es$ 
	to make 
	$\relfont{S}$ a governed relation. 

\paragraph{\bf Reduction-closedness}
	Immediate by the assumption that $\suc$ is fresh and $\prod_i R_i\not\red$. 

\paragraph{\bf Barb preserving}
\noi	Suppose $\tprocess{E', \Ga'}{P}{\Delta_1'}\downarrow_n$ 
	with $n\not\in \vec{m}$. 
	Then 
	\[
		\tprocess{E, \Ga}{(\nu \vec{m}, b) (P \Par \out{\suc}{1}{2}{\vec{n}, \srole{s}{\q}} \inact \Par \prod_i R_i)}{\Delta_1} \downarrow_n
	\]
\noi	by the freshness of $\suc$. 
	This implies 
	$\tprocess{E, \Ga}{(\nu \vec{m}, b) (Q \Par \out{\suc}{1}{2}{\vec{n}, \srole{s}{\q}} \inact \Par \prod_i R_i)}{\Delta_2}\Downarrow_n$. 
	Since $\out{\suc}{1}{2}{\vec{n}, \srole{s}{\q}} \inact \Par \prod_i R_i$ does not reduce, 
	we have $\tprocess{E', \Ga'}{Q}{\Delta_2'}\Downarrow_n$, as required.  \\

	Suppose $\tprocess{E', \Ga'}{P}{\Delta_1'}\downarrow_{\srole{s}{\p}}$.
	Then we chose $T\lrangle{N, \suc', \ell}$ such that
\[
	\begin{array}{l}
	\tprocess{E, \Ga}{(\nu \vec{m}, b) (P \Par \out{\suc}{1}{2}{\vec{n}, \srole{s}{\q}} \inact \Par \prod_i R_i \Par \inp{\suc}{2}{1}{\vec{y}, x} T\lrangle{N, \suc', \ell} )}{\Delta_1''}
	\Red\\ P' \Par \prod_i R_i \Par \out{\suc'}{1}{2}{\vec{n},\srole{s}{\q}} \inact \hastype \Delta_1'''
	\end{array}
\]
	which implies 
\[
	\begin{array}{l}
	\tprocess{E, \Ga}{(\nu \vec{m}, b) (Q \Par \out{\suc}{1}{2}{\vec{n}, \srole{s}{\q}} \inact \Par \prod_i R_i \Par \inp{\suc}{2}{1}{\vec{y}, x} T\lrangle{N, \suc', \ell} )}{\Delta_2''}
	\Red\\ Q' \Par \prod_i R_i \Par \out{\suc'}{1}{2}{\vec{n},\srole{s}{\q}} \inact \hastype \Delta_1'''
	\end{array}
\]
	which implies
	$\tprocess{E, \Ga}{Q}{\Delta_2'} \Downarrow_{\srole{s}{\p}}$.\\

%

\paragraph{\bf Contextual property}
	The only interesting case is if 
	$\tprocess{E',\Ga'}{P}{\Delta_1'} \ \relfont{S}  \ 
	\tprocess{E',\Ga'}{Q}{\Delta_2'}$ then 
	$\tprocess{E'',\Ga'}{P \Par R}{\Delta_1''} \ \relfont{S}  \ 
	\tprocess{E'',\Ga'}{Q \Par R}{\Delta_2''}$ for all $R$. 

	We
	compose with $O=\inp{\suc}{2}{1}{\vec{y}, x} (R \Par \out{\suc'}{1}{2}{\vec{z}, x} \inact)$.
\[
	\begin{array}{lll}
		(\nu \vec{m}, b) (P \Par \out{\suc}{1}{2}{v, \srole{s}{\q}} \inact \Par \prod_i R_i \Par O)
		& \govcongs & 
		(\nu \vec{m}, b) (Q \Par \out{\suc}{1}{2}{v, \srole{s}{\q}} \inact \Par \prod_i R_i \Par O)\\
		\textrm{ implies }\\
		(\nu \vec{m}, b) (P \Par R \Par \out{\suc'}{1}{2}{v, \srole{s}{\q}} \inact \Par \prod_i R_i)
		& \govcongs & 
		(\nu \vec{m}, b) (Q \Par R \Par \out{\suc'}{1}{2}{v, \srole{s}{\q}} \inact \Par \prod_i R_i)\\
		\textrm{ implies }\\
		P \Par R \relfont{S} Q \Par R
	\end{array}\vspace{-\baselineskip}
\]
\end{proof}

We can now we prove the completness direction

We prove:
\begin{quote}
if 
$\tprocess{E,\Ga}{P}{\Delta_1}\govcongs
\tprocess{E,\Ga}{Q}{\Delta_2}$ and 
$\tprocess{E,\Ga}{P}{\Delta_1}\trans{\ell}
\tprocess{E',\Ga'}{P'}{\Delta_1'}$, then\\ 
$\tprocess{E,\Ga}{Q}{\Delta_2}\Trans{\hat{\ell}}
\tprocess{E',\Ga'}{Q'}{\Delta_2'}$ such that 
$\tprocess{E',\Ga'}{P'}{\Delta_1'}\govcongs
\tprocess{E',\Ga'}{Q'}{\Delta_2'}$
\end{quote} 

The case $\ell=\tau$ is trivial by the definition of $\govcongs$. 

For the case of $\ell\not=\tau$, we use Lemma~\ref{lem:definibility} and  
Lemma~\ref{lem:extruction}. 
Since the governed witness and session environments do 
not change the proof, we omit environments.
Suppose $P \govcongs Q$ and $P \trans{\ell} P'$. We must find a
matching weak transition from $Q$. Choose $N$ to contain all the free
names in both $P$ and $Q$ and choose $\suc$ and $b$ 
to be fresh for both $P$ and $Q$. We denote the tester by $T$ for
convenience. 

Because $\govcongs$ is contextual, we know $T \ | \ P \govcongs T \ | \
Q$. We also know 
\[
T \ | \ P \red^\ast 
\newsp{\bn{\ell}, b}{\PP' \Par \out{\suc}{1}{2}{\obj{\ell}, n} \inact}
\]
and therefore $T \ | \ Q \red^\ast Q''$ for some $Q''$ such that 
\[
\newsp{\bn{\ell}, b}{\PP' \Par \out{\suc}{1}{2}{\obj{\ell}, n} \inact} 
\govcongs
Q''
\]
and $Q''\Downarrow \suc$. Thus by Lemma \ref{lem:definibility}, 
we can set 
$Q'' \equiv \newsp{\bn{\ell}, b}{\Q'\Par \out{\suc}{1}{2}{\obj{\ell}, n} \inact}$
where $Q \Trans{\ell} Q'$. Since $\equiv$ is included in 
$\govcongs$, we have: 
\[
\newsp{\bn{\ell}, b}{\PP' \Par \out{\suc}{1}{2}{\obj{\ell}, n} \inact} 
\govcongs
\newsp{\bn{\ell}, b}{\Q' \Par \out{\suc}{1}{2}{\obj{\ell}, n} \inact} 
\]
By Lemma \ref{lem:extruction}, we have $\PP'\govcongs \Q'$, as
required. \vspace{-\baselineskip}

\end{proof}

\subsection{Proof for Lemma \ref{lem:full-abstraction}}
\label{app:full-abstraction}
\begin{proof}
	We prove direction if
	$\forall E, \tprocess{E, \Ga}{P_1}{\De_1} \govwbs \noGtprocess{P_2}{\De_2}$
	then
	$\Gtprocess{P_1}{\De_1} \swb \Gtprocess{P_2}{\De_2}$.

	\noi	If
		$\Gtprocess{P_1}{\De_1} \trans{\ell} \noGtprocess{P_1'}{\De_1'}$ then
		$P_1 \trans{\ell} P_1'$ and
		$(\Ga, \De_1) \trans{\ell} (\Ga', \De_1')$.

		From part 3 of Lemma \ref{lemma:configuration_transition}
		we choose $E$ such that
		$(E, \Ga, \De_1) \trans{\ell} (E', \Ga', \De_1')$.
		Since $\forall E, \tprocess{E, \Ga}{P_1}{\De_1} \govwbs \noGtprocess{P_2}{\De_2}$
		it can now be implied that,
		$\tprocess{E, \Ga}{P_1}{\De_1} \trans{\ell} \tprocess{E', \Ga}{P_1'}{\De_1'}$ implies
		$\tprocess{E, \Ga}{P_2}{\De_2} \Trans{\ell} \tprocess{E', \Ga}{P_2'}{\De_2'}$ which implies
		$P_2 \Trans{\ell} P_2'$ and
		$(E, \Ga, \De_2) \Trans{\ell} (E', \Ga', \De_2')$.

	\noi	From part 1 of Lemma  \ref{lemma:configuration_transition} we obtain
		$(\Ga, \De_2) \Trans{\ell} (\Ga', \De_2')$ implies
		$\Gtprocess{P_2}{\De_2} \Trans{\ell} \noGtprocess{P_2'}{\De_2'}$ as required.

$ $\\

		We prove direction if
		$\Gtprocess{P_1}{\De_1} \swb \Gtprocess{P_2}{\De_2}$ then
		$\forall E, \tprocess{E, \Ga}{P_1}{\De_1} \govwbs \noGtprocess{P_2}{\De_2}$.

	\noi	Let $\tprocess{E, \Ga}{P_1}{\De_1} \trans{\ell} \noGtprocess{P_1'}{\De_1'}$ then
		\begin{eqnarray}
			P_1 &\trans{\ell}& P_1' \\
			(E, \Ga, \De_1) &\trans{\ell}& (E', \Ga', \De_1') \label{proof:full_abstraction_1}
		\end{eqnarray}

	\noi 	If
		$\Gtprocess{P_1}{\De_1} \trans{\ell} \noGtprocess{P_1'}{\De_1'}$ then 
		$P_1 \trans{\ell} P_1',
		(\Ga, \De_1) \trans{\ell} (\Ga', \De_1')$ and 
		$\Gtprocess{P_2}{\De_2} \trans{\ell} \noGtprocess{P_2'}{\De_2'}$.

	\noi	From the last implication we obtain
		\begin{eqnarray}
			P_2 &\Trans{\ell}& P_2' \label{proof:full_abstraction_2} \\
			(\Ga, \De_2) &\Trans{\ell}& (\Ga', \De_2')\\ \label{proof:full_abstraction_3}
			\De_1 &\bistyp& \De_2 \label{proof:full_abstraction_4}
		\end{eqnarray}

	\noi	We apply part 2 of Lemma \ref{lemma:configuration_transition}
		to (\ref{proof:full_abstraction_1}) and
		(\ref{proof:full_abstraction_4}) to obtain
		 $(E, \Ga, \De_2) \Trans{\ell} (E', \Ga', \De_2')$.
		From the last result and (\ref{proof:full_abstraction_2}), we obtain
		$\tprocess{E, \Ga}{P_2}{\De_2} \Trans{\ell} \tprocess{E', \Ga}{P_2'}{\De_2'}$.
	\endproof
\end{proof}

\subsection{Proof for Theorem \ref{thm:coincidence}}
\label{app:thm_coincidence}

\begin{proof}
$\tprocess{\Ga}{P}{\De} \trans{\ell} \tprocess{\Ga'}{P'}{\De'}$ 
implies $(\Gamma,\Delta) \trans{\ell} (\Gamma',\Delta')$
by definition.  
This implies there exists 
$E$ such that $(E,\Gamma,\Delta) \trans{\ell} (E',\Gamma',\Delta')$
by  part 3 of Lemma~\ref{lemma:configuration_transition}. 
Then by definition, 
if $\tprocess{\Ga}{P}{\De} \trans{\ell} \tprocess{\Ga'}{P'}{\De'}$
then $\exists E$ such that $\tprocess{E, \Ga}{P}{\De} 
\trans{\ell} \tprocess{E', \Ga'}{P'}{\De'}$. 
Similarly, we have if 
$\tprocess{\Ga}{P}{\De} \Trans{\hat{\ell}} \tprocess{\Ga'}{P'}{\De'}$
then $\exists E$ such that $\tprocess{E, \Ga}{P}{\De} 
\Trans{\hat{\ell}} \tprocess{E', \Ga'}{P'}{\De'}$. 



Suppose $P$ is simple. Then by definition, we can set 
$P\equiv (\nu \ \vec{a}\vec{s})(P_1 \ | \  P_2 \ | \ \cdots \ | \ P_n)$
where $P_i$ contains either zero or a single session name $s$, which is
not used in process $P_j$, $i \not = j$. 

Suppose $\tprocess{\Ga}{P}{\De}$. Then it is derived from 
$\tprocess{\Ga\cdot\Ga_0}{P_i}{\De_i}$ where $\De_i$ contains 
a zero or single session name and $\De\cdot\De_0 = \De_1\cdots\De_n$ 
where 
$\Ga_0$ and $\De_0$ correspond to the environments 
of the restrictions 
$\vec{a}$ and $\vec{s}$, respectively. 
If $\tprocess{E_i,\Ga\cdot\Ga_0}{P_i}{\De_i}$ and 
$(\Ga\cdot\Ga_0, \De_i)\trans{\ell}(\Ga'\cdot\Ga_0', \De_i')$, then since $\De_i$ 
contains at most only a single session, the condition 
$E\trans{\lambda}E'$ in the premise 
in [Out, In, Sel, Bra, Tau] 
in Figure~\ref{fig:ELTS} is always true. Hence  
$(E_i,\Ga\cdot\Ga_0, \De_i)\trans{\ell}(E_i',\Ga'\cdot\Ga_0', \De_i')$. 
From here,  we obtain for all $E_i$ such that 
$\tprocess{E_i,\Ga\cdot\Ga_0}{P_i}{\De_i}$, 
if 
$\tprocess{\Ga\cdot\Ga_0}{P_i}{\De_i}\trans{\ell}
\tprocess{\Ga'\cdot\Ga_0'}{P_i'}{\De_i'}$, then 
we have 
$\tprocess{E_i,\Ga\cdot\Ga_0}{P_i}{\De_i}\trans{\ell}\tprocess{E_i',\Ga'\cdot\Ga_0'}{P_i'}{\De_i'}$. Similarly, for the case of 
$\tprocess{\Ga\cdot\Ga_0}{P_i}{\De_i}\Trans{\ell}
\tprocess{\Ga'\cdot\Ga_0'}{P_i'}{\De_i'}$.

Now by applying the parallel composition, 
we can reason for all $E$ such that $E \red^\ast \sqcup_i E_i$, 
if $\tprocess{E,\Ga}{P}{\De}$ and  
$\tprocess{\Ga}{P}{\De}\trans{\ell}\tprocess{\Ga'}{P'}{\De'}$, 
there exits 
$\tprocess{\Ga\cdot\Ga_0}{P_i}{\De_i}\trans{\ell}\tprocess{\Ga'\cdot\Ga_0'}{P_i'}{\De_i'}$, 
which implies 
$\tprocess{E,\Ga}{P}{\De}\trans{\ell}\tprocess{E',\Ga'}{P'}{\De'}$, 
with $E \red^\ast \sqcup_i E_i \trans{\lambda} \sqcup_i E_i' = E'$.  

By this, if $P$ is simple and 
$\tprocess{\Ga}{P}{\De}\trans{\ell}\tprocess{\Ga'}{P'}{\De'}$, 
for all $E$ such that $\tprocess{E,\Ga}{P}{\De}$, 
$\tprocess{E,\Ga}{P}{\De}\trans{\ell}\tprocess{E',\Ga'}{P'}{\De'}$. 
%
%
Hence if 
there exits $E$ such that 
$\tprocess{E,\Ga}{P}{\De}\trans{\ell}\tprocess{E',\Ga'}{P'}{\De'}$, 
then for all $E_0$, 
$\tprocess{E_0,\Ga}{P}{\De}\trans{\ell}\tprocess{E_0',\Ga'}{P'}{\De'}$. 
Similarly, for  $\tprocess{E,\Ga}{P}{\De}\Trans{\ell}\tprocess{E',\Ga'}{P'}{\De'}$.  

Now suppose if $P_1$ and $P_2$ are simple and 
	$\exists E$ such that $\tprocess{E, \Ga}{P_1}{\De_1} \wb_g^s
        \noGtprocess{P_2}{\De_2}$. 
	From the above result 
and part 2 of Lemma~\ref{lemma:configuration_transition}, 
	if $P_1$ and $P_2$ are simple and 
	$\exists E$ such that $\tprocess{E, \Ga}{P_1}{\De_1} \wb_g^s \noGtprocess{P_2}{\De_2}$ then $\forall E, \tprocess{E, \Ga}{P_1}{\De_1} \wb_g^s \noGtprocess{P_2}{\De_2}$. By applying
Lemma \ref{lem:full-abstraction} we are done.
\end{proof}

\fi
\end{document}